\documentclass[12pt, reqno, letterpaper]{amsart}

\usepackage{graphicx}    
\usepackage{bbm, xcolor}
\usepackage{xfrac}
\usepackage{multicol}
\usepackage{soul}

\usepackage[margin=1in]{geometry}

\usepackage[english]{babel}

\usepackage{amsmath,amsthm,amsfonts,amssymb}

\newtheorem{theorem}{Theorem}[section]
\newtheorem{corollary}{Corollary}[section]
\newtheorem{lemma}{Lemma}[section]
\newtheorem{proposition}{Proposition}[section]
\newtheorem{assumption}{Assumption}[section]

\newtheorem{remark}{Remark}[section]
\numberwithin{figure}{section}

\begin{document}

\title[Short-maturity VIX and European options with jumps]
{VIX and European options with jumps in the short-maturity regime}

\author{Desen Guo}
\address{Florida State University, Tallahassee, Florida, United States of America, gd19k@fsu.edu}

\author{Dan Pirjol}
\address{Stevens Institute of Technology, Hoboken, New Jersey, United States of America, dpirjol@gmail.com}

\author{Xiaoyu Wang}
\address{Hong Kong University of Science and Technology (Guangzhou), People's Republic of China, xiaoyuwang@hkust-gz.edu.cn}

\author{Lingjiong Zhu}
\address{Florida State University, Tallahassee, Florida, United States of America, zhu@math.fsu.edu}

\date{January 23, 2026}

\keywords{VIX option, short-maturity asymptotics, jumps}

\begin{abstract}
We present a study of the short-maturity asymptotics for VIX and European option prices in local-stochastic volatility models with compound Poisson jumps. 
Both out-of-the-money (OTM) and at-the-money (ATM) asymptotics are considered. 
The leading-order asymptotics are obtained in closed-form.
We apply our results to three examples: the Eraker model, 
a Kou-type model, and a folded normal model.
Numerical illustrations are provided for these three examples that show
the accuracy of predictions based on the asymptotic results.
\end{abstract}

\maketitle



\section{Introduction}

The CBOE Volatility Index (VIX) is the main volatility benchmark of the U.S. stock market and provides a measure of the implied volatility of options with a maturity of 30 days on the S\&P 500 index. It is defined in terms of an expectation in the risk-neutral measure 
\begin{equation}\label{VIX:formal:defn}
\mathrm{VIX}_t^2 = - \frac{2}{\tau} \mathbb{E}\left[\log\left(\frac{S_{t+\tau}}{S_t}\right)\Big|\mathcal{F}_t\right]+\frac{2}{\tau}\mathbb{E}\left[\int_{t}^{t+\tau}\frac{dS_{s}}{S_{s-}}\Big|\mathcal{F}_{t}\right],    
\end{equation} 
where $S_t$ is the equity index S\&P 500 at time $t$, and $\tau = 30$ days. The expectation is computed by replication in terms of market-observed SPX option prices - see the VIX White Paper \cite{VIXwp} for the details of the methodology.
Since 2022, CBOE has also started reporting 
the CBOE 1-day Volatility Index \texttt{(VIX1D)} \cite{VIX1Dwp}, which is an analog of the VIX index computed using the PM-settled weekly SPX options which mature on the same day and the next day $(\tau=1$ day) as the index date.

The volatility index VIX is used by market participants to speculate on and hedge volatility risk. Several volatility derivatives that can be used for this purpose are traded on the CBOE Options Exchange: futures contracts on VIX have been traded since 2004, and 
VIX options have been traded since 2006. In view of the popularity of these contracts, a great deal of work has been devoted in the literature to the valuation of volatility derivatives. 

\subsection{Literature review}

The study of VIX derivatives has gained considerable attention due to their role in hedging and managing market volatility risks. Early research focused on jump-diffusion models to incorporate sudden shifts in volatility and price dynamics. 
We distinguish three types of jumps: two independent jumps in the underlying and in the volatility, respectively, and simultaneous jumps in these two processes.
Furthermore, the jump intensity can be taken to be constant for simplicity, or it may be stochastic in order to reproduce the observed volatility clustering effect. 

Empirical studies of the equity markets by Duffie, Pan and Singleton (2000) \cite{Duffie2000} and Eraker (2004) \cite{Eraker2004} demonstrated the importance of extending jump-diffusion models by allowing also for jumps in the volatility process in order to obtain a good fit to data. Including jumps in both the asset price and volatility process is essential also for the accurate valuation of volatility derivatives.

Lin and Chang (2009) further explored VIX option pricing with state-dependent jumps, highlighting the importance of volatility jumps, especially for short-term options~\cite{Lin2009}. Their result was corrected by Lian and Zhu (2013) \cite{Lian2013}, who gave an exact solution for the VIX futures and option prices in the SVJJ model as a single integral.

Psychoyios et al. (2009) introduced a jump-diffusion model for VIX options and futures,
showing that jumps in both price and volatility significantly improve pricing accuracy~\cite{Psychoyios2009}. 
In particular, they studied the empirical properties of the variance process and considered several models which were also used to price VIX derivatives. Their favored model was the mean-reverting logarithmic diffusion with jumps.

Todorov and Tauchen (2011) emphasized the jump components in the S\&P 500 volatility and the VIX index, revealing that market volatility often arises from large, discrete jumps rather than continuous changes~\cite{Todorov2011}. Eraker (2010) investigated volatility and jump risk premiums embedded in VIX derivatives, providing evidence of their critical role in explaining option prices~\cite{Eraker2010}.

As models evolved, researchers incorporated more advanced jump structures. Sepp (2008) proposed a stochastic volatility model with random jumps, demonstrating its ability to price VIX options consistently alongside SPX options~\cite{Sepp2008}. Kokholm and Stisen (2013) extended this idea with a 3/2 stochastic volatility model incorporating jumps, achieving better consistency between equity and volatility derivatives~\cite{Kokholm2013}. Zang et al. (2017) introduced a double-jump diffusion model that used the VVIX index as a proxy for VIX volatility, capturing co-jumps and their clustering effects~\cite{Zang2017}. Cao et al. (2020) developed a two-factor model with infinite-activity jumps, which better describes the frequent small jumps observed in volatility markets compared to finite-activity models~\cite{Cao2020}.

Recent advances introduced volatility clustering and asymmetric jump effects through a stochastic jump intensity process. Jing et al. (2020) applied the Hawkes jump-diffusion model to VIX options, capturing the clustering of jumps during periods of financial stress~\cite{Jing2020}. Li et al. (2017) proposed a pure jump model for VIX dynamics, which allows for infinite activity and variation, achieving a better fit for the steep volatility skew in VIX options~\cite{Kokholm2015}. Park (2016) studied the effects of asymmetric volatility and separately modeled upward and downward jumps, finding that upward jumps and volatility asymmetry have a larger impact on VIX derivative prices~\cite{Li2017,Park2016}. 
Ye et al. (2023) \cite{Ye2023} proposed a model with stochastic jump intensity which allows more flexibility in calibrating to market data with non stationary jumps dynamics.

This line of research demonstrates a clear progression from basic jump-diffusion models to more advanced frameworks incorporating infinite-activity jumps, co-jumps, and clustering effects. These improvements have significantly enhanced the accuracy of VIX derivative pricing. In this paper, we extend this line of research by analyzing short-maturity VIX options and European options under a jump-diffusion framework, aiming to further refine the understanding of VIX dynamics in the presence of jumps.
Most of the literature on the VIX options with jumps assumes a Heston-like variance process and does not include a local volatility component. This is done in order to preserve the affine nature of the model, which allows the computation of the characteristic function in closed form. In this paper, we relax both of these constraints and do not limit ourselves to the class of affine models.

Short-maturity asymptotics for option pricing in the presence
of jumps have been studied in the literature, mostly for European options. 
For a wide class of models with jumps, the leading short-maturity asymptotics
of the European call options is given by Boyarchenko and Levendorskii (2002) \cite{BL02}.
Later, Figueroa-Lopez (2008) \cite{FL08} weakened the technical conditions required in \cite{BL02}. 
The next-to-leading order correction of $O(T^2)$ to the short-maturity asymptotics has been obtained in Figueroa-Lopez and Forde (2012) \cite{FLF}, and the leading asymptotics of the ATM skew was studied by Figueroa-Lopez and Olaffson (2016) \cite{FLO}.
Al\`os et al (2007) \cite{Alos2007} obtained short-maturity expansions in jump-diffusion models using Malliavin calculus methods. 
We refer to \cite{FLF} for a detailed list of references for short-maturity
European option pricing in the presence of jumps.
In this paper, we extend these results to local-stochastic volatility models with jumps, considering the general case of idiosyncratic jumps and common jumps in the asset and variance processes.  

Short-maturity asymptotics for path-dependent options in models with jumps are much less studied.
To the best of our knowledge, the only work in this domain
is the study of short-maturity asymptotics for Asian options in local volatility models in the presence of jumps presented in \cite{PZAsianJumps}. 
In this paper, we study the short-maturity asymptotics of VIX options in a combined limit of the small expiry and small averaging period.
In this limit the path dependence of the VIX options disappears and they 
reduce to European-style options on payoffs depending on both the asset and variance process. 
We derive the leading asymptotics of both VIX and European options in local-stochastic volatility models with jumps, under appropriate technical conditions.
We consider both the ATM and OTM options asymptotics and illustrate the theoretical predictions with explicit results for several models of jumps, including the popular model of Eraker \cite{Eraker2004} widely used in the literature.

The short maturity approach has the advantage that it leads to analytical results which are useful for practical application to valuation and model calibration. The theoretical predictions are tested by comparing with numerical simulations of VIX and European models. The agreement is reasonably good, for sufficiently small option maturity.


\subsection{Model setup}

In this paper, we are interested in the VIX call and put option prices 
in the presence of jumps. 
The VIX call and put option prices are given by
\begin{equation}
C_{V}(K,T)=e^{-rT}\mathbb{E}\left[\left(\mathrm{VIX}_{T}-K\right)^{+}\right],
\qquad
P_{V}(K,T)=e^{-rT}\mathbb{E}\left[\left(K-\mathrm{VIX}_{T}\right)^{+}\right],
\end{equation}
where $K>0$ is the strike price.

In \cite{VIXpaper}, the short-maturity VIX options and European options 
are studied in which the asset price $S_t$ is assumed to follow a local-stochastic volatility model under the
risk-neutral probability measure $\mathbb{Q}$:
\begin{eqnarray}\label{LSvol}
&& \frac{dS_t}{S_t} = \eta(S_t) \sqrt{V_t} dW_t + (r - q) dt \,, \\
&& \frac{dV_t}{V_t} =\sigma(V_t) dZ_t + \mu(V_t) dt \,,\nonumber
\end{eqnarray}
with initial conditions $S_0 > 0, V_0 > 0$, 
where $W_t, Z_t$ are correlated standard Brownian motions with correlation $\rho$, $r$ is the risk-free rate and $q$ is the dividend yield.
See also \cite{PZ2025} for short-maturity VIX options for SABR model.

We start by formulating technical conditions and assumptions for the parameters of the model (\ref{LSvol}).
First, we assume that $\eta(\cdot),\mu(\cdot)$ and $\sigma(\cdot)$ are
uniformly bounded. As in \cite{VIXpaper}, we impose
the following assumptions on $\eta(\cdot),\mu(\cdot)$ and $\sigma(\cdot)$.

\begin{assumption}\label{assump:bounded}
We assume that $\eta(\cdot),\mu(\cdot)$ and $\sigma(\cdot)$ are
uniformly bounded:
\begin{equation}
\sup_{x\in\mathbb{R}^{+}}\eta(x)\leq M_{\eta},
\qquad
\sup_{x\in\mathbb{R}^{+}}|\mu(x)|\leq M_{\mu},
\qquad
\sup_{x\in\mathbb{R}^{+}}\sigma(x)\leq M_{\sigma}.
\end{equation}
\end{assumption}

\begin{assumption}
\label{assump:lip}
We assume that $\eta$ is $L$-Lipschitz and $\sigma$ is $L'$-Lipschitz.
\end{assumption}

\begin{assumption}\label{assump:LDP}
We assume that $\inf_{x\in\mathbb{R}^{+}}\sigma(x)>0$
and $\inf_{x\in\mathbb{R}^{+}}\eta(x)>0$. Moreover, 
there exist some constants $M,\alpha>0$ such that
for any $x,y\in\mathbb{R}$,
$|\sigma(e^{x})-\sigma(e^{y})|\leq M|x-y|^{\alpha}$
and $|\eta(e^{x})-\eta(e^{y})|\leq M|x-y|^{\alpha}$.
\end{assumption}

In this paper, we study the short-maturity asymptotics of VIX and European options in the presence of jumps. 
We will show that with the presence of jumps, 
the leading-order short-maturity asymptotics
of VIX and European options
can behave very differently than
the one without jumps \eqref{LSvol} as studied in \cite{VIXpaper}.
In particular, we assume that
under the risk-neutral measure $\mathbb{Q}$,
a local-stochastic volatility model with jumps
has the dynamics:
\begin{eqnarray}\label{LSvol:jumps}
&& \frac{dS_t}{S_t} = \eta(S_t) \sqrt{V_t} dW_t + \left(r - q-\lambda^{S}\mu^{S}-\lambda^{C}\mu^{C,S}\right) dt+dJ_{t}^{S}+dJ_{t}^{C,S} \,, \\
&& \frac{dV_t}{V_t} =\sigma(V_t) dZ_t + \mu(V_t) dt+dJ_{t}^{V}+dJ_{t}^{C,V} \,,\nonumber
\end{eqnarray}
where
\begin{equation}
J_{t}^{S}:=\sum_{i=1}^{N_{t}^{S}}\left(e^{Y_{i}^{S}}-1\right),
\end{equation}
denotes the idiosyncratic jumps of the log-asset price, where $N_{t}^{S}$
is a Poisson process with intensity $\lambda^{S}$, and $Y_{i}^{S}$ are independent and identically 
distributed (i.i.d.) with probability distribution function $P^{S}(x)$, where $-\infty<x<\infty$,
and when its probability density function exists, we denote it as $p^{S}(x)$.
Similarly,
\begin{equation}
J_{t}^{V}:=\sum_{i=1}^{N_{t}^{V}}\left(e^{Y_{i}^{V}}-1\right),
\end{equation}
denotes the idiosyncratic jumps of the log-variance process, where $N_{t}^{V}$
is a Poisson process with intensity $\lambda^{V}$, and $Y_{i}^{V}$ are i.i.d.
distributed with probability distribution function $P^{V}(x)$, where $-\infty<x<\infty$,
and when its probability density function exists, we denote it as $p^{V}(x)$,
and $J_{t}^{C,S}$, $J_{t}^{C,V}$ are the common jump processes
for the log-asset price and log-variance process such that
\begin{equation}
J_{t}^{C,S}=\sum_{i=1}^{N_{t}^{C}}\left(e^{Y_{i}^{C,S}}-1\right),
\qquad
J_{t}^{C,V}=\sum_{i=1}^{N_{t}^{C}}\left(e^{Y_{i}^{C,V}}-1\right),
\end{equation}
where $N_{t}^{C}$
is a Poisson process with intensity $\lambda^{C}$, and $(Y_{i}^{C,S},Y_{i}^{C,V})$ are i.i.d.
distributed with probability distribution function $P^{C}(x,y)$. 
If it has the probability density function $p^{C}(x,y)$, then we can
assume that $p^{C}(x,y)$ is positive where $x<0$ and $y>0$ or $y<0$ and $x>0$.
This assumption is due to the leverage effect in finance
such that for the common jumps, when $S_{t}$ jumps downward,
$V_{t}$ will jump upward, and vice versa. 
We assume that the expectations of $e^{Y_{i}^{S}},e^{Y_{i}^{V}},e^{Y_{i}^{C,S}},e^{Y_{i}^{C,V}}$
all exist and denote $\mu^{S}$, $\mu^{V}$, $\mu^{C,S}$, $\mu^{C,V}$ as
\begin{equation}\label{mu:defn}
\mu^{S}:=\mathbb{E}[e^{Y_{1}^{S}}]-1, 
\quad\mu^{V}:=\mathbb{E}[e^{Y_{1}^{V}}]-1, 
\quad\mu^{C,S}:=\mathbb{E}[e^{Y_{1}^{C,S}}]-1, 
\quad
\mu^{C,V}:=\mathbb{E}[e^{Y_{1}^{C,V}}]-1.
\end{equation}
In addition, we impose the following assumption.

\begin{assumption}
Assume that 
$\mathbb{E}[e^{2Y_{1}^{V}}]$
and $\mathbb{E}[e^{2Y_{1}^{C,V}}]$ are both finite.
\end{assumption}

Under these technical assumptions, we will study the
short-maturity asymptotics for VIX and European options
with compound Poisson jumps. 
The rest of the paper is organized as follows.
In Section~\ref{sec:main}, we present
the main results.
In particular, we derive the leading-order
short-maturity asymptotics for both OTM and ATM cases
for the VIX options in Section~\ref{sec:VIX}, 
and for the European options in Section~\ref{sec:European}.
Next, we introduce and study several models for the jumps dynamics in Section~\ref{sec:models}: the
Eraker model (Section~\ref{sec:Eraker:model}), a Kou-type model (Section~\ref{sec:Kou:model}) and a folded normal model (Section~\ref{sec:folded:model}).
In Section~\ref{sec:pred}, we provide numerical illustrations for the efficiency of our short-maturity asymptotic predictions, comparing them with MC simulation of VIX and European options for the models presented in Section~\ref{sec:models}.
The proofs of all the results are provided in the Appendix.

\section{Main Results}\label{sec:main}

\subsection{VIX options}\label{sec:VIX}

In this section, we derive the short-maturity asymptotics
for VIX options.
We recall from \eqref{VIX:formal:defn} that VIX at time $T$ is defined as: 
\begin{equation}\label{VIX:formal:defn:2}
\mathrm{VIX}_T^2 = - \frac{2}{\tau} \mathbb{E}\left[\log\left(\frac{S_{T+\tau}}{S_T}\right)\Big|\mathcal{F}_T\right]+\frac{2}{\tau}\mathbb{E}\left[\int_{T}^{T+\tau}\frac{dS_{t}}{S_{t-}}\Big|\mathcal{F}_{T}\right].    
\end{equation}
We can compute from \eqref{LSvol:jumps} that
\begin{equation}\label{defn:second:term}
\frac{2}{\tau}\mathbb{E}\left[\int_{T}^{T+\tau}\frac{dS_{t}}{S_{t}}\Big|\mathcal{F}_{T}\right]=\frac{2}{\tau}\int_{T}^{T+\tau}(r-q)dt
=2(r-q),
\end{equation}
and
\begin{equation}
S_{T+\tau}=S_{T}e^{\int_{T}^{T+\tau}\eta(S_{t})\sqrt{V_{t}}dW_{t}
-\int_{T}^{T+\tau}\frac{1}{2}\eta^{2}(S_{t})V_{t}dt
+\left(r - q-\lambda^{S}\mu^{S}-\lambda^{C}\mu^{C,S}\right)\tau+(R^{S}_{T+\tau}-R^{S}_{T})+(R^{C,S}_{T+\tau}-R^{C,S}_{T})},
\end{equation}
where, for any $t\geq 0$,
\begin{equation}
R_{t}^{S}:=\sum_{i=1}^{N_{t}^{S}}Y_{i}^{S},
\qquad
R_{t}^{C,S}:=\sum_{i=1}^{N_{t}^{C}}Y_{i}^{C,S}.
\end{equation}
Hence, we can compute from \eqref{LSvol:jumps}, \eqref{VIX:formal:defn:2} and \eqref{defn:second:term} that
\begin{equation}\label{VIX:formula}
\mathrm{VIX}_T^2 = \frac{1}{\tau} \mathbb{E}\left[\int_{T}^{T+\tau}\eta^{2}(S_{t})V_{t}dt\Big|\mathcal{F}_T\right]
+\kappa,
\end{equation}
where
\begin{equation}\label{kappa:eqn}
\kappa:=2\lambda^{S}\mu^{S}+2\lambda^{C}\mu^{C,S}
-2\lambda^{S}m^{S}-2\lambda^{C}m^{C,S},    
\end{equation}
where $m^{S}:=\mathbb{E}[Y_{1}^{S}]$ and $m^{C,S}:=\mathbb{E}[Y_{1}^{C,S}]$.
Note that we can rewrite \eqref{kappa:eqn} as
\begin{align}
\kappa&=2\lambda^{S}\left(\mu^{S}-m^{S}\right)+2\lambda^{C}\left(\mu^{C,S}-m^{C,S}\right)
\nonumber
\\
&=2\lambda^{S}\left(\mathbb{E}[e^{Y_{1}^{S}}]-1-\mathbb{E}[Y_{1}^{S}]\right)
+2\lambda^{C}\left(\mathbb{E}[e^{Y_{1}^{C,S}}]-1-\mathbb{E}[Y_{1}^{C,S}]\right),
\end{align}
which implies that $\kappa\geq 0$.

As $\tau\rightarrow 0$, it follows from \eqref{VIX:formula} that we have
\begin{equation}
\mathrm{VIX}_{T}^{2}\rightarrow\eta^{2}(S_{T})V_{T}+\kappa,
\end{equation}
almost surely. Indeed, we have the following result that provides
an upper bound on the distance between $\mathrm{VIX}_{T}^{2}$
and its proxy $\eta^{2}(S_{T})V_{T}+\kappa$.

\begin{proposition}\label{prop:VIXsmalltau}
If Assumption~\ref{assump:bounded} and~\ref{assump:lip} hold, and 
\begin{equation}\label{eta:2:assumption}
\sup_{s\geq 0}|(\eta^{2})''(s)s^{2}|\leq M_{\eta,2} \,.
\end{equation}
Then we have
\begin{align}
\left|\mathrm{VIX}_T^2
-\eta^{2}(S_{T})V_{T}-\kappa\right|\leq C_{1}(\tau)S_{T}+C_{2}(\tau)V_{T}
+(\lambda^{S}+\lambda^{C})M_{\eta}^{2}\tau,
\end{align}
and moreover,
\begin{align}
&\mathbb{E}\left|\mathrm{VIX}_T^2-\eta^{2}(S_{T})V_{T}-\kappa\right|
\nonumber
\\
&\leq
C_{1}(\tau)S_{0}e^{(r-q)T}
+C_{2}(\tau)V_{0}e^{\lambda^{V}T\mu^{V}}e^{\lambda^{C}T\mu^{C,V}}e^{TM_{\mu}}
+(\lambda^{S}+\lambda^{C})M_{\eta}^{2}\tau,
\end{align}
where
\begin{equation}\label{C:1:tau}
C_{1}(\tau):=2LM_{\eta}|r-q-\lambda^{S}\mu^{S}-\lambda^{C}\mu^{C,s}|e^{|r-q|\tau}\tau,
\end{equation}
and
\begin{align}
C_{2}(\tau)&:=
M_{\eta}^{2}\Bigg(e^{\lambda^{V}\tau(\mathbb{E}[e^{2Y_{1}^{V}}]-1)}e^{\lambda^{C}\tau(\mathbb{E}[e^{2Y_{1}^{C,V}}]-1)}e^{2\tau M_{\mu}}
e^{4\tau M_{\sigma}^{2}}+1
\nonumber
\\
&\qquad\qquad\qquad\qquad\qquad-2e^{-\lambda^{V}\tau(\mathbb{E}[e^{Y_{1}^{V}}]-1)}e^{-\lambda^{C}\tau(\mathbb{E}[e^{Y_{1}^{C,V}}]-1)}e^{-\tau M_{\mu}-\frac{1}{2}\tau M_{\sigma}^{2}}\Bigg)^{1/2}
\nonumber
\\
&\qquad\qquad
+\frac{\tau}{2}M_{\eta,2}M_{\eta}^{2}e^{\lambda^{V}\tau\mu^{V}}e^{\lambda^{C}\tau\mu^{C,V}}e^{\tau M_{\mu}}.\label{C:2:tau}
\end{align}
\end{proposition}

As a corollary, we have the following result that provides
an upper bound on the distance between $\mathrm{VIX}_{T}$
and its proxy $\sqrt{\eta^{2}(S_{T})V_{T}+\kappa}$.

\begin{corollary}\label{cor:VIXsmalltau}
Suppose the same assumptions in Proposition~\ref{prop:VIXsmalltau} hold.
Then, we have
\begin{align}
\left|\mathrm{VIX}_T
-\sqrt{\eta^{2}(S_{T})V_{T}+\kappa}\right|\leq \frac{C_{1}(\tau)}{\sqrt{\kappa}}S_{T}+\frac{C_{2}(\tau)}{\sqrt{\kappa}}V_{T}
+\frac{1}{\sqrt{\kappa}}\left(\lambda^{S}+\lambda^{C}\right)M_{\eta}^{2}\tau,
\end{align}
where $C_{1}(\tau),C_{2}(\tau)$ are given in \eqref{C:1:tau}-\eqref{C:2:tau}
and moreover
\begin{align}
&\mathbb{E}\left|\mathrm{VIX}_T-\sqrt{\eta^{2}(S_{T})V_{T}+\kappa}\right|
\nonumber
\\
&\leq
\frac{C_{1}(\tau)}{\sqrt{\kappa}}S_{0}e^{(r-q)T}
+\frac{C_{2}(\tau)}{\sqrt{\kappa}}V_{0}e^{\lambda^{V}T\mu^{V}}e^{\lambda^{C}T\mu^{C,V}}e^{TM_{\mu}}+\frac{1}{\sqrt{\kappa}}\left(\lambda^{S}+\lambda^{C}\right)M_{\eta}^{2}\tau.\label{tau:order}
\end{align}
\end{corollary}

Therefore, the VIX call and put option prices can be approximated by
\begin{equation}
\hat{C}_{V}(K,T):=e^{-rT}\mathbb{E}\left[\left(\sqrt{\eta^{2}(S_{T})V_{T}+\kappa}-K\right)^{+}\right],
\end{equation}
and
\begin{equation}
\hat{P}_{V}(K,T):=e^{-rT}\mathbb{E}\left[\left(K-\sqrt{\eta^{2}(S_{T})V_{T}+\kappa}\right)^{+}\right],
\end{equation}
as $\tau\rightarrow 0$.
More precisely, $\hat{C}_{V}(K,T)=\lim_{\tau\rightarrow 0}C_{V}(K,T)$
and $\hat{P}_{V}(K,T)=\lim_{\tau\rightarrow 0}P_{V}(K,T)$.

In the rest of the paper, we will assume the $\tau\to 0$ limit and approximate $\mathrm{VIX}_T$ by $\sqrt{\eta^2(S_T) V_T + \kappa}$ as is justified in Corollary~\ref{cor:VIXsmalltau}. 
The error introduced by this approximation is expected to be negligible for options on the CBOE 1-day VIX index \texttt{(VIX1D)} \cite{VIX1Dwp}, which is an analog of the usual 30-day VIX index, and is computed using 
PM-settled weekly SPX options which mature on the
same day and the next day ($\tau$= 1 day) as the index date. 

When $\eta^{2}(S_{0})V_{0}+\kappa<K^{2}$, the VIX call option is OTM,
and when $\eta^{2}(S_{0})V_{0}+\kappa>K^{2}$, the VIX put option is OTM.
We have the following result for OTM short-maturity VIX options.

\begin{theorem}\label{thm:OTM}
(i) When $\eta^{2}(S_{0})V_{0}+\kappa<K^{2}$, 
\begin{align}
\lim_{T\rightarrow 0}\lim_{\tau\rightarrow 0}\frac{C_{V}(K,T)}{T}
&=\lambda^{S}\int_{\mathbb{R}}\left(\sqrt{\eta^{2}\left(S_{0}e^{x}\right)V_{0}+\kappa}-K\right)^{+}P^{S}(dx)
\nonumber
\\
&\qquad
+\lambda^{C}\int_{\mathbb{R}}\int_{\mathbb{R}}\left(\sqrt{\eta^{2}\left(S_{0}e^{x}\right)V_{0}e^{y}+\kappa}-K\right)^{+}P^{C}(dx,dy)
\nonumber
\\
&\qquad\qquad
+\lambda^{V}\int_{\mathbb{R}}\left(\sqrt{\eta^{2}\left(S_{0}\right)V_{0}e^{x}+\kappa}-K\right)^{+}P^{V}(dx).
\end{align}

(ii) When $\eta^{2}(S_{0})V_{0}+\kappa>K^{2}$,
\begin{align}
\lim_{T\rightarrow 0}\lim_{\tau\rightarrow 0}\frac{P_{V}(K,T)}{T}
&=\lambda^{S}\int_{\mathbb{R}}\left(K-\sqrt{\eta^{2}\left(S_{0}e^{x}\right)V_{0}+\kappa}\right)^{+}P^{S}(dx)
\nonumber
\\
&\qquad
+\lambda^{C}\int_{\mathbb{R}}\int_{\mathbb{R}}\left(K-\sqrt{\eta^{2}\left(S_{0}e^{x}\right)V_{0}e^{y}+\kappa}\right)^{+}P^{C}(dx,dy)
\nonumber
\\
&\qquad\qquad
+\lambda^{V}\int_{\mathbb{R}}\left(K-\sqrt{\eta^{2}\left(S_{0}\right)V_{0}e^{x}+\kappa}\right)^{+}P^{V}(dx).
\end{align}
\end{theorem}

It is easy to see from \eqref{C:1:tau}-\eqref{C:2:tau}
that $C_{1}(\tau)=O(\tau)$ and $C_{2}(\tau)=O(\sqrt{\tau})$
as $\tau\rightarrow 0$. Therefore, if $\tau=o(T^{2})$ as $T\rightarrow 0$, then the results in Theorem~\ref{thm:OTM} still hold
if one removes the limit $\tau\rightarrow 0$ in Theorem~\ref{thm:OTM}.

We observe from Theorem~\ref{thm:OTM} that
in the OTM case, $C_{V}(K,T)$ and $P_{V}(K,T)$ are of order $T$
as $T\rightarrow 0$, in contrast to $e^{-O(1/T)}$
for the local-stochastic volatility model without jumps (Theorem~4.1 in \cite{VIXpaper}).
Also, the leading-order terms in Theorem~\ref{thm:OTM}
only depend on the jumps, not the diffusion part 
in the local-stochastic volatility model.
The intuition is that 
the probability of having one single jump
that can
move the option from OTM to ITM is of order $T$,
which dominates the probability of order 
$e^{-O(1/T)}$
for the local-stochastic volatility model without jumps (Theorem~4.1 in \cite{VIXpaper}).
Therefore, the leading-order terms in Theorem~\ref{thm:OTM}
only depend on the jumps, whereas the diffusion part is negligible.

Next, let us consider the ATM case for VIX options, i.e. $\eta^{2}(S_{0})V_{0}+\kappa=K^{2}$.
We have the following result for ATM short-maturity VIX options.

\begin{theorem}\label{thm:ATM}
When $\eta^{2}(S_{0})V_{0}+\kappa=K^{2}$, we have    
\begin{align}
&\lim_{T\rightarrow 0}\lim_{\tau\rightarrow 0}\frac{1}{\sqrt{T}}C_{V}(K,T)
=\lim_{T\rightarrow 0}\lim_{\tau\rightarrow 0}\frac{1}{\sqrt{T}}P_{V}(K,T)
\nonumber
\\
&=\frac{1}{\sqrt{2\pi}}\frac{\sqrt{\eta^{2}(S_{0})V_{0}}}{\sqrt{\eta^{2}(S_{0})V_{0}+\kappa}}\nonumber
\\
&\qquad\cdot\sqrt{\left((\eta(S_{0})\frac{1}{2}\sigma(V_{0})\sqrt{V_{0}}+\eta'(S_{0})\eta(S_{0})S_{0}V_{0}\rho\right)^{2}+\left(\eta'(S_{0})\eta(S_{0})S_{0}V_{0}\sqrt{1-\rho^{2}}\right)^{2}}.
\end{align}
\end{theorem}

Since $C_{1}(\tau)=O(\tau)$ and $C_{2}(\tau)=O(\sqrt{\tau})$
as $\tau\rightarrow 0$ from \eqref{C:1:tau}-\eqref{C:2:tau}. Therefore, if $\tau=o(T)$ as $T\rightarrow 0$, then the results in Theorem~\ref{thm:ATM} still hold
if one removes the limit $\tau\rightarrow 0$ in Theorem~\ref{thm:ATM}.

We observe from Theorem~\ref{thm:ATM} that
in the ATM case, $C_{V}(K,T)$ and $P_{V}(K,T)$ are of order $\sqrt{T}$
as $T\rightarrow 0$, and the limits coincide
with those of the local-stochastic volatility model without jumps \cite{VIXpaper}.
The intuition is that 
the probability of having any jump
is of order $T$, which is negligible
compared to the order $\sqrt{T}$
for the local-stochastic volatility model without jumps (Theorem~4.1 in \cite{VIXpaper}).
Therefore, the leading-order terms in Theorem~\ref{thm:ATM}
only depend on the diffusion part, whereas the jumps part is negligible. 

\subsection{European options}\label{sec:European}

In this section, we derive the short-maturity asymptotics
for European options.
Although short-maturity asymptotics
for European options in the presence of jumps
have been studied in the literature \cite{GatheralBook,Alos2007}, 
to the best of our knowledge, 
the results on the short-maturity asymptotics of these options
in the 
local-stochastic volatility model with jumps presented here are new.

The prices of European call and put options are given by
\begin{equation}
C_{E}(K,T)=e^{-rT}\mathbb{E}[(S_{T}-K)^{+}],
\qquad
P_{E}(K,T)=e^{-rT}\mathbb{E}[(K-S_{T})^{+}],
\end{equation}
where $K>0$ is the strike price.
When $S_{0}<K$, the European call option is OTM,
and when $S_{0}>K$, the European put option is OTM.
We have the following result for OTM short-maturity European options.

\begin{theorem}\label{thm:OTM:European}
(i) When $S_{0}<K$, 
\begin{align}
\lim_{T\rightarrow 0}\frac{C_{E}(K,T)}{T}
=\lambda^{S}\int_{\mathbb{R}}\left(S_{0}e^{x}-K\right)^{+}P^{S}(dx)
+\lambda^{C}\int_{\mathbb{R}}\int_{\mathbb{R}}\left(S_{0}e^{x}-K\right)^{+}P^{C}(dx,dy).
\end{align}

(ii) When $S_{0}>K$,
\begin{align}
\lim_{T\rightarrow 0}\frac{P_{E}(K,T)}{T}
=\lambda^{S}\int_{\mathbb{R}}\left(K-S_{0}e^{x}\right)^{+}P^{S}(dx)
+\lambda^{C}\int_{\mathbb{R}}\int_{\mathbb{R}}\left(K-S_{0}e^{x}\right)^{+}P^{C}(dx,dy).
\end{align}
\end{theorem}

Similarly as OTM VIX options (Theorem~\ref{thm:OTM}), we observe from Theorem~\ref{thm:OTM:European} that
in the OTM case, $C_{E}(K,T)$ and $P_{E}(K,T)$ are of order $T$
as $T\rightarrow 0$, in contrast to $e^{-O(1/T)}$
for the local-stochastic volatility model without jumps (Theorem~3.1 in \cite{VIXpaper}).
Also, the leading-order terms in Theorem~\ref{thm:OTM:European}
only depend on the jumps, not the diffusion part 
in the local-stochastic volatility model.

Next, let us consider the ATM case for European options, i.e. $S_{0}=K$.
We have the following result for ATM short-maturity European options.

\begin{theorem}\label{thm:ATM:European}
When $S_{0}=K$, we have    
\begin{align}
\lim_{T\rightarrow 0}\frac{1}{\sqrt{T}}C_{E}(K,T)
=\lim_{T\rightarrow 0}\frac{1}{\sqrt{T}}P_{E}(K,T)
=\frac{\eta(S_{0})\sqrt{V_{0}}}{\sqrt{2\pi}}.
\end{align}
\end{theorem}

Similarly as ATM VIX options (Theorem~\ref{thm:ATM}), we observe from Theorem~\ref{thm:ATM:European} that
in the ATM case, $C_{E}(K,T)$ and $P_{E}(K,T)$ are of order $\sqrt{T}$
as $T\rightarrow 0$, and the limits coincide
with those of the local-stochastic volatility model without jumps \cite{VIXpaper}.

\section{Modeling the Jumps}
\label{sec:models}

In this section, we present and study a few models for the idiosyncratic $Y_1^S, Y_1^V$ and common jumps $Y_1^{C,S}, Y_1^{C,V}$. 
We will obtain asymptotic predictions for these models, 
applying the theoretical results 
obtained in Section~\ref{sec:main}.
In particular, we will present the short-maturity leading asymptotic predictions for both OTM VIX and European options.

The models we present are guided by the following empirical observations.

1. \textit{Idiosyncratic jumps.} There are idiosyncratic jumps in $S$ and $V$, driven by independent Poisson processes with intensities $\lambda^{S}, \lambda^{V}$.
The jumps in the underlying can take either sign, upward or downward. 
On the other hand, empirical evidence \cite{Eraker2004} suggests that jumps in the variance process can only be positive. 

2. \textit{Common jumps.} There are jumps that occur simultaneously in the underlying and in the variance process. They are driven by a common Poisson process with intensity $\lambda^{C}$. The size distribution has the same sign restriction as the idiosyncratic jumps. 

3. \textit{Negative correlation of jumps in $V$ and $S$.} 
Empirical evidence \cite{Eraker2004} shows that jumps in variance are always positive and are negatively correlated with jumps in $S$. Therefore, the joint distribution of the common jumps is modeled such that their correlation is negative.

4. \textit{The variance process $V_t$.}
It is an empirical fact that the variance process mean-reverts to a long-term finite value $V_L$. This can be achieved with a Heston-like process
$dV_t = (\theta - \kappa V_t) dt + \sigma_V V_t dZ_t$.
The technical conditions of our paper are satisfied by a geometric Brownian motion (gBM) process for $V_t$ of the form $dV_t/V_t = \mu_V dt + \sigma_V dZ_t$.
The gBM process approaches the long-term value $V_L=0$ if the drift is sufficiently small, since  $\lim_{t\to \infty} V_t = 0$ a.s. for $\mu_V < \frac12 \sigma_V^2$. 

We consider several models for the jump size distribution.

\subsection{Eraker model}\label{sec:Eraker:model}

One popular model in the literature for stochastic volatility with common jumps in $S,V$ is the Eraker model \cite{Eraker2004} (denoted in this reference as the SVCJ model). 
The model contains idiosyncratic jumps $Y_1^S$ for the asset $S$ and $Y_1^V$ for the variance process $V$. They are normally and exponentially distributed, respectively
\begin{align}
Y_1^S &= \alpha_S + \sigma_S Z_S\,,\quad Z_S \sim \mathcal{N}(0,1)\,,\\
Y_1^V &\sim \mathrm{Exp}(\eta_V) \,.
\end{align}
The common jumps to $S$ and $V$ are driven by a Poisson process with intensity $\lambda^{C}$. The jump distribution for $V$ is exponential: 
\begin{equation}
Y_1^{C,V} \sim \mathrm{Exp}(\eta_{C,V})\,.
\end{equation}
Conditional on  $Y_1^{C,V}$, there is a jump in the underlying which is normally distributed as 
\begin{align}
Y_1^{C,S} = \mu^{S} + \rho_J Y_1^{C,V} + \sigma_{C,S} Z_{C,S}\,,\quad Z_{C,S}\sim\mathcal{N}(0,1)\,.
\end{align}
The joint density of the jumps is given by
\begin{align}\label{pCEraker}
p^{C}(x,y) = \eta_{C,V} e^{-\eta_{C,V} y}\cdot
\frac{1}{\sqrt{2\pi} \sigma_{C,S}} e^{-\frac{1}{2\sigma_{C,S}^2} (x - \mu^{S} - \rho_J y)^2}\,.
\end{align}

This model is widely used in the literature on VIX modeling with jump-diffusions; see e.g.
Lin and Chang (2009) \cite{Lin2009}, Lian and Zhu (2013)
\cite{Lian2013} and Ye, Wu and Chen (2023) \cite{Ye2023}.

\textit{Compensators.} 
The compensators for the $S,V$ jumps are 
\begin{align}
\mu^{S}  &= \mathbb{E}\left[e^{Y_1^S}\right]-1 = e^{\alpha_S + \frac12 \sigma_S^2} - 1\,,\\
\mu^{V}  &= \mathbb{E}\left[e^{Y_1^V}\right]-1 = \frac{1}{\eta_V-1}\,,\quad \eta_V > 1\,.
\end{align}

The compensator for the common jumps $Y_1^{C,S}$ is
\begin{align}
\mu^{C,S} = \eta_{C,V} \int_0^\infty e^{-\eta_{C,V} y} \left( e^{\mu^{S} + \rho_J y + \frac12 \sigma_{C,S}^2} - 1 \right) dy = \frac{\eta_{C,V}}{\eta_{C,V} - \rho_J} 
 e^{\mu^{S} + \frac12 \sigma_{C,S}^2} - 1\,.
\end{align}

\subsubsection{Asymptotic predictions for European options} 

We take $\eta(x)\equiv 1$ since the model does not include a local volatility component. 
The contribution from the $S$ idiosyncratic jumps follows from Boyarchenko and
Levendorskii (2002) \cite{BL02}. 
The short-maturity asymptotic prediction for OTM European options 
in the Eraker model is obtained from Theorem~\ref{thm:OTM:European}.

\begin{corollary}\label{cor:Eraker:European}

Assume that the asset price follows a jump-diffusion model \eqref{LSvol:jumps} with jump distribution given by the Eraker model. The leading-order short maturity asymptotics of the OTM European options are given as follows.

a) When $S_{0}<K$,
\begin{align}
a_{E,C}(K) := \lim_{T\to 0}
\frac{C_{E}(K,T)}{T}
&= \lambda^{C} \eta_{C,V} \int_0^\infty e^{-\eta_{C,V} y}
c_{\mathrm{BS}}\left(K,S_0e^{\mu^{S}+\rho_J y}, \sigma_{C,S}\right) dy \nonumber \\
&\quad+ \lambda^{S} \left\{ S_0 e^{\alpha_S + \frac12 \sigma_S^2} 
\Phi\left(\frac{-k + \alpha_S + \sigma_S^2}{\sigma_S}\right) - K \Phi\left(\frac{-k+\alpha_S}{\sigma_S}\right) \right\}\,,\nonumber
\end{align}
with $c_{\mathrm{BS}}(K,F,v) := F \Phi(-\frac{1}{v} \log(K/F)+\frac12 v) - 
K \Phi(-\frac{1}{v} \log(K/F) - \frac12 v)$ and $k=\log(K/S_0)$.
We denoted $\Phi(x):=\frac{1}{\sqrt{2\pi}}\int_{-\infty}^{x}e^{-\frac{y^{2}}{2}}dy$ the cumulative distribution function of a standard normal random distribution $\mathcal{N}(0,1)$.

b) When $0<K<S_{0}$,
\begin{align}
a_{E,P}(K) := \lim_{T\to 0}
\frac{P_{E}(K,T)}{T} 
&= \lambda^{C} \eta_{C,V} \int_0^\infty e^{-\eta_{C,V} y}
p_{\mathrm{BS}}\left(K,S_0e^{\mu^{S}+\rho_J y}, \sigma_{C,S}\right) dy\nonumber \\
&\quad+ \lambda^{S} \left\{ K \Phi\left(\frac{k-\alpha_S}{\sigma_S}\right) - S_0 e^{\alpha_S + \frac12 \sigma_S^2} 
\Phi\left(\frac{k - \alpha_S - \sigma_S^2}{\sigma_S}\right)  \right\}\,,\nonumber
\end{align}
with $p_{\mathrm{BS}}(K,F,v) := - F \Phi(\frac{1}{v} \log(K/F)-\frac12 v) + 
K \Phi(\frac{1}{v} \log(K/F) + \frac12 v)$ and $k=\log(K/S_0)$. 
\end{corollary}

\subsubsection{Asymptotic predictions for VIX options}

The short maturity asymptotics for OTM VIX call and put options in the Eraker model
follow from Theorem~\ref{thm:OTM} and are given by the following corollary.

\begin{corollary}
\label{cor:Eraker:VIX}
Assume that the asset price follows a jump diffusion \eqref{LSvol:jumps} with jump dynamics given by the Eraker model. Then the leading-order short maturity asymptotics 
for the 
OTM VIX options are given as follows.

a) OTM VIX call options. When $V_{0}+\kappa<K^{2}$, 
\begin{align}
a_{V,C}(K) := \lim_{T\to 0}
\frac{C_{V}(K,T)}{T} 
&= \lambda^{C}  e^{-\eta_{C,V} y_0} \Big[\eta_{C,V} I_1(\kappa, V_0 e^{y_0}, \eta_{C,V}) - K \Big] \nonumber \\
&\qquad\qquad+ \lambda^{V} e^{-\eta_{V} y_0} \Big[\eta_{V} I_1(\kappa, V_0 e^{y_0}, \eta_{V}) - K \Big]\,, \nonumber 
\end{align}
where $y_0(K) = \log\frac{K^2-\kappa}{V_0}$ and 
$I_1(a,b,\eta)=\frac{2\sqrt{b}}{2\eta-1} {}_2F_1(-1/2, -1/2+\eta, 1/2+\eta,-a/b)$, where ${}_2 F_1(a,b,c;z)$ is the Gauss hypergeometric function, 
and 
\footnote{Note that $\lambda^{S}$ contributes indirectly through 
$\kappa$, even if the term proportional to $\lambda^{S}$ vanishes since $\eta(x)\equiv 1$.}
\begin{align}
\kappa= 2\lambda^{S} \left(e^{\alpha_S + \frac12 \sigma_S^2}-1 - \alpha_S\right)+ 2\lambda^{C} \left(\frac{\eta_{C,V}}{\eta_{C,V} - \rho_J} 
 e^{\mu^{S} + \frac12 \sigma_{C,S}^2} - 1 - \mu^{S} - \frac{\rho_J}{\eta_{C,V}}\right) \,.\nonumber
\end{align}

b) OTM VIX put options. 
When $0<K^2< V_{0}+\kappa $, 
\begin{align}\label{put:Eraker:0}
a_{V,P}(K) := \lim_{T\to 0}
\frac{P_{V}(K,T)}{T} =0.
\end{align}
\end{corollary}

Note that the short maturity asymptotics for OTM VIX put options in the Eraker model vanishes (see \eqref{put:Eraker:0}).
The reason for this null result is that both the idiosyncratic and common
jumps in $V$ are positive, and they do not contribute to the integral. 
The leading contribution to the OTM VIX put options as $T\to 0$ comes from 
the diffusive component which is exponentially suppressed in the order of $O(e^{-1/T})$.

Next, we describe two possible alternatives to this model, focusing on the common jumps contribution. The choice of the idiosyncratic jumps distribution is independent of the dynamics chosen for the common jumps, and can be treated in a similar way to the Eraker model. 
For simplicity, we omit their contribution to the asymptotic results, although we show explicitly the contribution from $\lambda^{S}$ to $\kappa$ which is independent of the distributional choice for $Y_1^S$.

\subsection{Kou-type model}\label{sec:Kou:model}

This model replaces the normal assumption for $Y_1^{C,S}$ with a doubly-exponential distribution, similar to the doubly-exponential Kou model \cite{kou2002jump}.
In contrast to the usual doubly-exponential model, we use this distribution for the common jumps, not for the idiosyncratic $S$ jumps.

The distribution of the $V$ jumps is exponential, just as in the Eraker model:
\begin{equation}
Y_1^{C,V} \sim\mathrm{Exp}(\eta_{C,V})\,.
\end{equation}
Conditional on  $Y_1^{C,V}$, the underlying $S$ has a double-exponential jump
\begin{align}
Y_1^{C,S} =
\begin{cases}
\mu^{S} + \rho_J Y_1^{C,V} + \mathrm{Exp}(\eta_{C,S}) & \text{ with probability $\alpha$}, \\
\mu^{S} + \rho_J Y_1^{C,V} - \mathrm{Exp}(\eta_{C,S}) & \text{ with probability $1- \alpha$}.
\end{cases}
\end{align}

In principle, the two exponential distributions for positive and negative jumps could have different parameters $\eta_{C,S}^+\neq 
\eta_{C,S}^-$ but we take them to be equal in analogy with the Eraker model where the normal distribution has symmetric tails. This also keeps the parameter count the same as in the Eraker model, and we can estimate the model parameters by matching the mean and variance. 

This model corresponds to the joint density of the common jumps
\begin{align}
p^{C}(x,y) &= \eta_{C,V} e^{-\eta_{C,V} y}
\cdot\Big\{
\alpha \eta_{C,S} e^{-\eta_{C,S} (x-\mu^{S}-\rho_J y)} 1_{x-\mu^{S} - \rho_J y \geq 0}\\
&\qquad\qquad\qquad\qquad\qquad+ (1-\alpha)
\eta_{C,S} e^{\eta_{C,S} (\mu^{S} + \rho_J y - x)} 1_{x-\mu^{S} - \rho_J y < 0}
\Big\}\,. \nonumber
\end{align}
As in the Eraker model, we take $\eta(x)\equiv 1$.

\textit{Compensators.} 
Since the testing in Section~\ref{sec:pred} focuses only on the common jumps, we give only the  compensator for $Y_1^{C,S}$.
\begin{align}
\mu^{C,S} = \mathbb{E}\left[e^{Y_1^{C,S}}\right]-1 = e^{\mu^{S}} \frac{\eta_{C,V}}{\eta_{C,V}-\rho_J}
\eta_{C,S} \frac{\eta_{C,S} + (2\alpha-1)}{\eta_{C,S}^2-1} - 1\,.
\end{align}
We give also the expectation of the common jump
\begin{equation}
m^{C,S} = \mathbb{E}\left[Y_1^{C,S}\right] = \mu^{S} + \frac{\rho_J}{\eta_{C,V}} + \frac{2\alpha-1}{\eta_{C,S}}\,.
\end{equation}

These expectations are required for the $\kappa$ parameter used for the short maturity asymptotics of the VIX options
\begin{equation}
\kappa = 2\lambda^{S} \left(\mu^{S} - m^{S}\right) + 2\lambda^{C} \left(\mu^{C,S} - m^{C,S}\right) \,.
\end{equation}

\subsubsection{Asymptotic predictions for European options}

The short-maturity asymptotic predictions for European options are obtained from 
Theorem~\ref{thm:OTM} and are parameterized in terms of the functions $a_{E,C}(K)$ 
and $a_{E,P}(K)$ defined as 
\begin{equation}
a_{E,C}(K) := \lim_{T\to 0} \frac{C_E(K,T)}{T}\,,\quad
a_{E,P}(K) := \lim_{T\to 0} \frac{P_E(K,T)}{T} \,.
\end{equation}
Denote $k:=\log(K/S_0)$ the option log-moneyness. We have the following short-maturity asymptotics for OTM European options for the Kou-type model.

\begin{corollary}\label{cor:Kou:European}
Assume that the asset price follows a jump-diffusion model \eqref{LSvol:jumps} with jump distribution given by the Kou-type model. 
Assume $\rho_J < 0 $ and $\mu^{S}<0$. 
Then the leading-order short maturity asymptotics 
for the OTM European options are given as follows.

a) OTM European call options. When $S_{0}<K$ (i.e $k>0$),
\begin{align}
a_{E,C}(K) = \lambda^{C} c_{R} f_c(k, \eta_{C,S}) + 
\lambda^{S} \mathcal{I}_{C,S}(K)\,,
\end{align}
where $\mathcal{I}_{C,S}(K):=\int_{\mathbb{R}}\left(S_{0}e^{x}-K\right)^{+}P^{S}(dx)$, 
\begin{align}
c_R := \alpha \frac{\eta_{C,V}}{\eta_{C,V}+\eta_{C,S} |\rho_J|} e^{\eta_{C,S} \mu^{S}}\,,
\end{align}
and
\begin{align}\label{f:c:eqn}
f_c(k,a) := S_0 \frac{1}{a-1} e^{-(a-1)k} \,.
\end{align}

b) OTM European put options. 
When $S_{0}>K$ (i.e. $k<0$), we distinguish two cases.
For $k<\mu^{S}<0$, we have
\begin{align}
a_{E,P}(K) = \lambda^{C} \left(  c_{1L} f_p\left(k, \frac{\eta_{C,V}}{|\rho_J|}\right) + c_{2L} f_p(k, \eta_{C,S}) 
\right) + \lambda^{S} \mathcal{I}_{P,S}(K) \,,
\end{align}
where $\mathcal{I}_{P,S}(K):=\int_{\mathbb{R}}\left(K-S_{0}e^{x}\right)^{+}P^{S}(dx)$,
\begin{align}
c_{1L} &:= \alpha \frac{\eta_{C,S} |\rho_J|}{\eta_{C,V} + \eta_{C,S} |\rho_J|} e^{-\frac{\eta_{C,V}}{|\rho_J|} \mu^{S}} - (1-\alpha)
\frac{\eta_{C,S} |\rho_J|}{\eta_{C,V} - \eta_{C,S} |\rho_J|} e^{-\frac{\eta_{C,V}}{|\rho_J|} \mu^{S}}\,,
\\
c_{2L} &:=(1-\alpha)
\frac{\eta_{C,V} }{\eta_{C,V} - \eta_{C,S} |\rho_J|} e^{- \eta_{C,S}  \mu^{S}} \,,
\end{align}
and
\begin{align}
f_p(k,a) := S_0 \frac{1}{a+1} e^{(a+1)k} \,,
\end{align}
and for $\mu^{S}<k<0$, we have
\begin{align}
a_{E,P}(K) &= \lambda^{C} \left( c_{1L} f_p\left(\mu^{S}, \frac{\eta_{C,V}}{|\rho_J|}\right) + c_{2L}  f_p(k, \eta_{C,S}) + 
c_{R} \left[ f_c\left(\mu^{S},\eta_{C,S}\right) - f_c(k,\eta_{C,S})\right] \right) \\
  &\qquad\qquad\qquad+ \lambda^{S} \mathcal{I}_{P,S}(K)\,, \nonumber
\end{align}
where $f_{c}$ is defined in \eqref{f:c:eqn}.
\end{corollary}

\begin{remark}
The marginal distribution $\int_{\mathbb{R}} p^{C}(x,y) dy$ of the common jumps 
has different analytical form for $x<\mu^S$ and $x>\mu^S$. This leads to different analytical expressions for the asymptotic coefficients $a_{E,C}(K), a_{E,P}(K)$ depending on the sign of $\mu^{S}$. We give here only the results for $\mu^{S}<0$ which corresponds to the practically relevant case of negative average jumps. The results for $\mu^{S}>0$ can be obtained analogously.
The results of Corollary \ref{cor:Kou:European}  assume also $\rho_J<0$.
This corresponds to anti-correlated $S,V$ common jumps which is the analog of the leverage effect for jumps. 
\end{remark}

In Corollary~\ref{cor:Kou:European}, the $\mathcal{I}_{C,S}(K),\mathcal{I}_{P,S}(K)$ terms depend on the distributional property assumed for the idiosyncratic jumps in $S$. Their distribution is independent of that of the common jumps so, as discussed, we omit them for simplicity.

\subsubsection{Asymptotic predictions for VIX options}

Keeping only the term proportional to $\lambda^{C}$ as discussed above, 
the short maturity asymptotics for OTM VIX call options in the Kou-type model has the same form as in the Eraker model (Corollary~\ref{cor:Eraker:VIX}), 
since the $V$ jumps have the same exponential distribution 
$Y_1^V \sim \mbox{Exp}(\eta_{C,V})$. 
The short maturity asymptotics for OTM VIX put options in the Kou-type model also 
has the same form as in the Eraker model.
We give it again below for completeness. 

\begin{corollary}
\label{cor:Kou:VIX}
Assume that the asset follows a jump diffusion model \eqref{LSvol:jumps} with Kou-type jump dynamics.
Then the leading-order short maturity asymptotics 
for the 
OTM VIX options are given as follows.

a) For OTM VIX call options, when $V_{0}+\kappa<K^{2}$,
\begin{align}
a_{V,C}(K) := \lim_{T\to 0}\frac{C_V(K,T)}{T}
= \lambda^{C}  e^{-\eta_{C,V} y_0} \left[\eta_{C,V} I_1(\kappa, V_0 e^{y_0}, \eta_{C,V}) - K 
\right] + \lambda^{V}\mathcal{I}_{C,V}(K) \,,
\nonumber
\end{align}
where $\mathcal{I}_{C,V}(K):=\int_{\mathbb{R}}\left(\sqrt{V_{0}e^{x}+\kappa}-K\right)^{+}P^{V}(dx)$, 
$y_0(K) = \log\frac{K^2-\kappa}{V_0}$
and
$$
I_1(a,b,\eta)=\frac{2\sqrt{b}}{2\eta-1} {}_2F_1(-1/2, -1/2+\eta, 1/2+\eta;-a/b), 
$$
where ${}_2 F_1(a,b,c;z)$ is the Gauss hypergeometric function.

b) For OTM VIX put options, when $0<K^2<V_{0}+\kappa $, 
\begin{align}\label{put:Kou:0}
a_{V,P}(K) := \lim_{T\to 0}
\frac{P_V(K,T)}{T} =0.
\end{align}
\end{corollary}

In Corollary~\ref{cor:Kou:VIX}, the $\mathcal{I}_{C,V}(K)$ term depends on the choice for the distribution of the idiosyncratic $V$ jumps, which is independent of that for the common jumps.
The short maturity asymptotics for OTM VIX put options in the Kou-type model vanishes (see \eqref{put:Kou:0}), just as in the Eraker model, since the common jumps in $V$ are strictly positive. 
The leading order asymptotics to the VIX put options comes from the diffusive contribution to the dynamics of $V$.

\subsection{Folded normal model}\label{sec:folded:model}

Another alternative model is to assume that $Y_1^{C,V}$ follows a 
folded normal distribution, which is positive definite
\begin{equation}
Y_1^{C,V} = \sigma_{C,V} |Z_{C,V}|\,,\quad Z_{C,V} \sim \mathcal{N}(0,1) \,.
\end{equation}
Conditional on  $Y_1^{C,V}$, the underlying $S$ has a normally distributed jump similar to the Eraker model
\begin{align}
Y_1^{C,S} = \mu^{S} + \rho_J Y_1^{C,V} + \sigma_{C,S} Z_{C,S}\,,\quad Z_{C,S}\sim\mathcal{N}(0,1).
\end{align}

The joint distribution of the common jumps is
represented by the probability density function:
\begin{align}
p^{C}(x,y) =  \frac{2}{\sqrt{2\pi } \sigma_{C,V}} e^{-\frac{y^2}{2\sigma_{C,V}^2}}\cdot
\frac{1}{\sqrt{2\pi} \sigma_{C,S}} e^{-\frac{1}{2\sigma_{C,S}^2} (x - \mu^{S} - \rho_J y)^2}\,.
\end{align}

Denote $f_{\mathrm{FN}}(y) := \int_{\mathbb{R}} p^{C}(x,y) dx$ the jump size probability density function for the $V$ common jump. This is given by
\begin{equation}\label{fFN}
f_{\mathrm{FN}}(y) = \frac{2}{\sqrt{2\pi} \sigma_{C,V}} e^{-\frac{y^2}{2\sigma_{C,V}^2}} \,.
\end{equation}

\textit{Compensators.} Since the testing focuses only on the common jumps, 
we give only the compensator for $Y_1^{C,S}$: 
\begin{align}
\mu^{C,S} = \mathbb{E}\left[e^{Y_1^{C,S}}\right]-1 = 
2e^{\mu^{S} + \frac12\sigma_{C,S}^2\rho_J^2} \Phi( \rho_J \sigma_{C,V}) - 1\,.
\end{align}
We give also the expectation of the common jump
\begin{equation}
m^{C,S} = \mathbb{E}\left[Y_1^{C,S}\right] = \mu^{S} + \sqrt{\frac{2}{\pi}}\rho_J \sigma_{C,V} \,.
\end{equation}
As in the Eraker model, we take $\eta(x)\equiv 1$.

\subsubsection{Asymptotic predictions for European options}

The short-maturity asymptotic predictions for European options are obtained from 
Theorem~\ref{thm:OTM} and are parameterized in terms of the functions $a_{E,C}(K)$ 
and $a_{E,P}(K)$ defined as 
\begin{equation}
a_{E,C}(K) := \lim_{T\to 0} \frac{C_E(K,T)}{T}\,,\quad
a_{E,P}(K) := \lim_{T\to 0} \frac{P_E(K,T)}{T} \,.
\end{equation}
Denote $k=\log(K/S_0)$ the option log-moneyness. 
We have the following short-maturity asymptotics for  OTM European options for the folded normal model.

\begin{corollary}\label{cor:FN:Eur}
Assume that the asset price follows a local-stochastic volatility model \eqref{LSvol:jumps} with folded normal common jumps distribution. Then the leading-order short maturity asymptotics 
for the OTM European options are given as follows.

a) 
For OTM call options with $K > S_0$, we have
\begin{align}
a_{E,C}(K) = \lambda^{C} \int_0^\infty f_{\mathrm{FN}}(y) c_{\mathrm{BS}}\left(K,S_0e^{\mu^{S} + \rho_J y}, \sigma_{C,S}\right) dy+\lambda^{S}\mathcal{I}_{C,S}(K)\,,
\end{align}
where $\mathcal{I}_{C,S}(K):=\int_{\mathbb{R}}\left(S_{0}e^{x}-K\right)^{+}P^{S}(dx)$ and
\begin{equation*}
c_{\mathrm{BS}}(K,F,v) := F \Phi\left(-\frac{1}{v} \log(K/F)+\frac12 v\right) - 
K \Phi\left(-\frac{1}{v} \log(K/F) - \frac12 v\right).    
\end{equation*}

b)
For OTM put options with $0<K<S_0$ we have
\begin{align}
a_{E,P}(K) = \lambda^{C} \int_0^\infty f_{\mathrm{FN}}(y) p_{\mathrm{BS}}\left(K,S_0e^{\mu^{S} + \rho_J y}, \sigma_{C,S}\right) dy+\lambda^{S}\mathcal{I}_{P,S}(K)\,,
\end{align}
where $\mathcal{I}_{P,S}(K):=\int_{\mathbb{R}}\left(K-S_{0}e^{x}\right)^{+}P^{S}(dx)$ and
\begin{align*}
p_{\mathrm{BS}}(K,F,v) := - F \Phi\left(\frac{1}{v} \log(K/F)-\frac12 v\right) + 
K \Phi\left(\frac{1}{v} \log(K/F) + \frac12 v\right).    
\end{align*}
\end{corollary}

In Corollary~\ref{cor:FN:Eur}, the $\mathcal{I}_{C,S}(K),\mathcal{I}_{P,S}(K)$ terms depend on the distributional property assumed for the idiosyncratic jumps in $S$. Their distribution is independent of that of the common jumps so, as discussed, we omit them for simplicity.

\subsubsection{Asymptotic predictions for VIX options}

We have the following short-maturity asymptotics for  OTM VIX options for the folded normal model.

\begin{corollary}
\label{cor:FN:VIX}
Assume that the asset follows a jump diffusion model \eqref{LSvol:jumps} with 
folded normal jump dynamics.
Then the leading-order short maturity asymptotics 
for the 
OTM VIX options are given as follows.

a) For OTM VIX call options, when $V_{0}+\kappa<K^{2}$, we have
\begin{align}
a_{V,C}(K) := \lim_{T\to 0} \frac{C_{V}(K,T)}{T}
= \lambda^{C}  
\int_0^\infty \left(\sqrt{V_0 e^y + \kappa} - K\right)^+ f_{\mathrm{FN}}(y) dy+\lambda^{V}\mathcal{I}_{C,V}(K) \,,
\end{align}
where $\mathcal{I}_{C,V}(K):=\int_{\mathbb{R}}\left(\sqrt{V_{0}e^{x}+\kappa}-K\right)^{+}P^{V}(dx)$.

b) For OTM VIX put options, when $0<K^2<V_{0}+\kappa $, we have
\begin{align}\label{put:FN:0}
a_{V,P}(K) := \lim_{T\to 0}
\frac{P_V(K,T)}{T} =0.
\end{align}
\end{corollary}

In Corollary~\ref{cor:FN:VIX}, the $\mathcal{I}_{C,V}(K)$ term depends on the choice for the distribution of the idiosyncratic $V$ jumps, which is independent of that for the common jumps.
The short maturity asymptotics for OTM VIX put options in the folded normal model vanishes (see \eqref{put:FN:0}), just as in the Eraker and Kou-type models, since the common jumps in $V$ are strictly positive. 
The leading order asymptotics to the OTM VIX put options comes from the diffusive contribution to the dynamics of $V$.

\section{Numerical Testing}
\label{sec:pred}

In Section~\ref{sec:models}, 
we obtained asymptotic predictions for the three models introduced in Section~\ref{sec:models}, applying the theoretical results 
obtained in Section~\ref{sec:main}.
The short-maturity asymptotics for both OTM VIX and European options have the general form
\begin{equation}
\lim_{T\to 0} \frac{C(K,T)}{T} = \lambda^{S} f_S(K) + \lambda^{C} f_C(K) + \lambda^{V} f_V(K)\,,
\end{equation}
where the three coefficients $f_{S,C,V}(K)$ are calculable.
The three terms correspond to idiosyncratic $S, V$ and common jumps, respectively. 

In this section, these predictions are tested by comparison with independent numerical option pricing obtained by Monte Carlo (MC) simulation of the models. 
The plan of testing is to compute the ratio $C(K,T)/T$ by MC simulation and compare with the analytical asymptotic prediction of this paper. We expect that they agree, for sufficiently small option maturity $T$. In order to streamline the testing and focus on the common jumps which are the novel feature of the model, we focus on the contribution of the common jumps with intensity $\lambda^{C}$.
Thus, we assume $\lambda^{S}=\lambda^{V}=0$. We also neglect the contribution of local volatility and take $\eta(x)\equiv 1$.

\subsection{Eraker model}

We start by presenting the test results for the Eraker model. 

\subsubsection{Parameters}


For simplicity, we take $S_0=1, r=q=0$. 
The numerical values of the model parameters for the common jumps sector are given in Table~\ref{tab:1}, following the paper of Lian and Zhu (2013) \cite{Lian2013}. The conversion to our notation is shown in the second column. The numerical value of the $\sigma_{C,S}$ parameter was taken larger than in \cite{Lian2013}, to amplify the common jumps effect for better test sensitivity.
Since we focus our testing on the contribution of the common jumps, the parameters of the idiosyncratic jumps distributions $\alpha_S,\sigma_S,\eta_V$ will not be required.

\begin{table}[h!]
  \centering
  \caption{The common jumps parameters in the Eraker model. The jumps model is $Y_{C,V} = \mathrm{Exp}(\mu^{V})$ and $Y_{C,S}|Y_{C,V} = \mu^{S} + \rho_J Y_{C,V} + \sigma_S \mathcal{N}(0,1)$.
  Note that $\mu^{V}$ defined by Lian and Zhu (2013) corresponds to our $1/\eta_{C,V}$. }
    \begin{tabular}{|c|c|}
    \hline
LZ parameters  & Our notation \\
    \hline\hline
$\lambda = 0.47$ & $\lambda^{C} = 0.47$ \\
\hline
$\mu^{V}  = 0.05$ & $\eta_{C,V} = 20.0 $ \\
$\mu^{S} = -0.0869$ & $\mu^{S} = -0.0869$ \\
$\rho_J = -0.38$ & $\rho_J = -0.38$ \\
$\sigma_S = 0.0001$ & $\sigma_{C,S}=0.1$ \\
\hline
$\sqrt{V_0} = 0.087$ & $V_0 = 0.0076$ \\
    \hline
    \end{tabular}%
  \label{tab:1}%
\end{table}%

\subsubsection{European options}

We present tests for the short-maturity asymptotic predictions for European options in the Eraker model provided in Corollary~\ref{cor:Eraker:European} by comparing with numerical pricing of the options.

Figure~\ref{Fig:E} shows the prices of European options in the Eraker model with maturity $T=0.1$ obtained by MC simulation of the model (red dots), rescaled as $1000 \frac{C_E(K,T)}{\lambda^{C} T}$ and $1000 \frac{P_E(K,T)}{\lambda^{C} T}$. 
The results for $k>1$ correspond to call options and for $k<1$ to put options.
The MC simulation used $10^5$ MC samples and the model was implemented using an Euler scheme for $\log S_t$ with $n=100$ time steps. 

The black curves show the asymptotic results for $a_{E,C}(K)$ and $a_{E,P}(K)$ from Corollary~\ref{cor:Eraker:European}.
The agreement of the asymptotic result in Corollary~\ref{cor:Eraker:European} with the MC simulation is reasonably good although we note larger differences for the call options. 
The disagreement is larger at the ATM point $K=S_0$ but this is explained by the large diffusive contributions which are formally sub-leading in the small maturity expansion $O(T^{1/2})$. For this reason a meaningful test should use only OTM strikes $K\neq S_0$ where the diffusive contribution is numerically negligible and the jumps contribution dominates.

In order to study the reason for the observed disagreement with the asymptotic result we consider also options with shorter maturity $T=0.01$.
Table~\ref{tab:ET0p01} shows the MC simulation
results for options with $T=0.01$. The agreement has improved, which shows that the reason for the differences observed is related to higher order contributions in the small maturity expansion. For most strikes away from the ATM point, the MC simulation agrees (within MC errors) with the asymptotic prediction in Corollary~\ref{cor:Eraker:European}.

\begin{figure}
\includegraphics[scale=0.65]{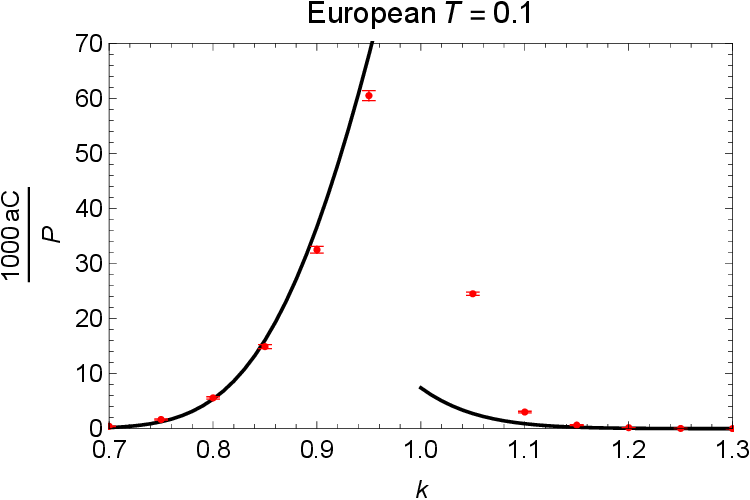}
\caption{Plots of the asymptotic results for European calls and puts with maturity $T=0.1$ in the Eraker model.
The black curves show the asymptotic prediction for
$1000 \frac{a_{E,C}(K)}{\lambda^{C}}$ and
$1000 \frac{a_{E,P}(K)}{\lambda^{C}}$. The theoretical results are compared with 
the MC simulation (red dots) described in text.}
\label{Fig:E}
\end{figure}

 \begin{table}[h!]
  \centering
  \caption{Numerical tests for European option pricing with maturity $T=0.01$ under the Eraker model. The MC simulation uses 100k paths. }
    \begin{tabular}{|c|c|c||c|c|c|}
    \hline
$K$ & $1000a_{E,C}/\lambda^{C}$ & MC simulation 
       & $K$ & $1000a_{E,P}/\lambda^{C}$ & MC simulation \\
    \hline\hline
1.0 & 7.409 & $785.2\pm 3.8$ 
      & 1.0 & 107.734 & $789.9 \pm 6.66$ \\
1.05 & 2.741 & $2.97\pm 0.75$ 
     & 0.95 & 67.8864 & $64.65\pm 4.23$ \\
1.1 & 0.897 & $1.14 \pm 0.38$
      & 0.9 & 36.6656 & $34.75\pm 2.74$  \\
1.15 & 0.262 & 0 
      & 0.85 & 16.0734 & $14.13 \pm 1.59$  \\
1.2 & 0.069 & 0 
      & 0.8 & 5.3573 & $4.517 \pm 0.816$ \\
1.25 & 0.017 & 0 
      & 0.75 & 1.2573 & $1.084 \pm 0.385$ \\
1.3 & 0.0004 & 0 
      & 0.7 & 0.1902 & $0.25 \pm 0.18$ \\
    \hline
    \end{tabular}%
  \label{tab:ET0p01}%
\end{table}%

\subsubsection{VIX options} 

We illustrate the short-maturity asymptotic prediction for VIX options in the Eraker model that is 
provided in Corollary~\ref{cor:Eraker:VIX}.

As discussed, for the purpose of numerical testing in this section we keep only the common jumps contribution and thus we take $\lambda^S=0$.
Using the numerical 
values of the parameters in Table~\ref{tab:1} we get
$\kappa = \lambda^C \cdot 0.020169 = 0.0095$.
We obtain for the strike of an ATM VIX option 
$K_{\mathrm{ATM}} = \sqrt{V_0 + \kappa} = 0.1308$.

The short maturity asymptotics of VIX options is described by the parameters $a_{V,C}(K)$ which are predicted by Corollary~\ref{cor:Eraker:VIX}.
The numerical results for $a_{V,C}(K)$ are shown in the second column of Table~\ref{tab:V} 
for several values of the strike ratio $k=K/K_{\mathrm{ATM}} > 1$ corresponding to OTM VIX call options.

\begin{table}[h!]
\centering
\caption{Short maturity predictions for the VIX call options under the Eraker model. The corresponding result for VIX put options vanishes, as explained in text. The last column shows the result of a simulation for VIX options with maturity $T=0.1$.
Model parameters as in Table~\ref{tab:1} and $\mu^{V}=0, \sigma_V=0.01$.}
\begin{tabular}{|c|c||c|}
\hline
$k=K/K_{\mathrm{ATM}}$ & $1000a_{V,C}/\lambda^{C}$ & $1000 C_V^{\mathrm{MC}}(K,T)/(\lambda^{C} T)$ \\
    \hline\hline
1.00 & 1.51125 & $2.193 \pm 0.030$ \\
1.02 & 0.28352  & $0.266\pm 0.013$ \\
1.04 & 0.05901 & $0.050\pm 0.006$  \\
1.06 &  0.01345 & $0.008\pm 0.003$ \\
1.08 & 0.00332 & $0.002 \pm 0.002$ \\
1.10 & 0.00088 & $0.001 \pm 0.001$ \\
1.12 & 0.00025 & $0.001\pm 0.001$ \\
    \hline
    \end{tabular}%
  \label{tab:V}%
\end{table}%

We test these asymptotic prediction in Corollary~\ref{cor:Eraker:VIX} by MC simulation of VIX option prices. 
For this simulation the VIX index $\mathrm{VIX}_T$ is approximated by
$\sqrt{V_T + \kappa}$, which corresponds to the $\tau \to 0$ limit. 
The same MC parameters are used as for the European options tests.
We take $\sigma_V=0.01$ to be small in order to test the model in a regime where the dynamics of $V$ is jump dominates.

As a first test of the simulation, we compute the
VIX forward at maturity $T=0.1$
\begin{equation}
F_V(T) = \mathbb{E}[\mathrm{VIX}_T] = 0.13046\pm 0.00002,
\end{equation}
This is close to the VIX ATM strike $K_{\mathrm{ATM}} = 0.1308$,
as expected in the short-maturity limit.

VIX call option prices obtained by simulation are used to compute $\frac{C_V(K,T)}{\lambda^{C} T}$. The results are shown in the last column of Table~\ref{tab:V}, and in Figure~\ref{Fig:aV}. The results are consistent with the asymptotic prediction for the Eraker model in Corollary~\ref{cor:Eraker:VIX}. 
The corresponding values for VIX puts (not shown) are very small and consistent with zero as predicted by Corollary~\ref{cor:Eraker:VIX}. 



\begin{figure}
\centering
\includegraphics[scale=1.0]{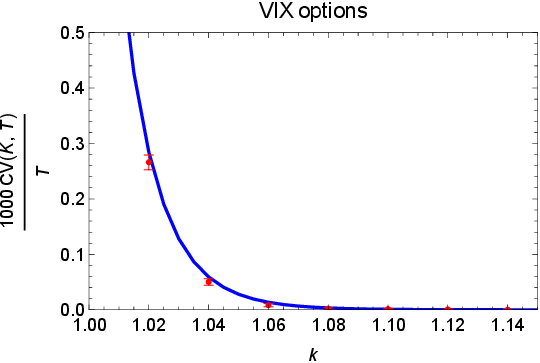}
\caption{Test for the asymptotic results for VIX calls under the Eraker model. 
The blue curve shows the asymptotic result for 
$1000 \frac{a_{V,C}(K)}{\lambda^{C}}$ and the red dots the result of an MC simulation.
The test uses VIX options with maturity $T=0.1$ and $\sigma_V=0.001$ under the Eraker model.}
\label{Fig:aV}
\end{figure}

\subsection{Kou-type model}

We present in this section the testing for the
alternative model which replaces the normal assumption for $Y_1^{C,S}|Y_1^{C,V}$ with a doubly-exponential distribution.
The asymptotic predictions for this model for VIX and European options were given in Corollary~\ref{cor:Kou:European} and Corollary~\ref{cor:Kou:VIX}.
As mentioned, the double-exponential jumps are used for the common jumps of $S,V$, 
and not for modeling the idiosyncratic $S$ jumps as in the classical application in \cite{kou2002jump}.

In this model, the common $V$ jumps are 
exponential as in the Eraker model
\begin{equation}
Y_1^{C,V} \sim\mathrm{Exp}(\eta_{C,V}).
\end{equation}
Conditional on  $Y_1^{C,V}$, the underlying $S$ has a double-exponential jump
\begin{align}
Y_1^{C,S} =
\begin{cases}
\mu^{S} + \rho_J Y_1^{C,V} + X_1 & \text{ with probability } \alpha, \\
\mu^{S} + \rho_J Y_1^{C,V} - X_1 & \text{ with probability } 1- \alpha.
\end{cases}
\end{align}
where $X_1 \sim\mathrm{Exp}(\eta_{C,S})$ is an exponentially distributed random variable. 
This model corresponds to the joint distribution of the common jumps
\begin{align}\label{pC}
p^{C}(x,y) &= \eta_{C,V} e^{-\eta_{C,V} y} \Big\{
\alpha \eta_{C,S} e^{-\eta_{C,S} (x-\mu^{S}-\rho_J y)} 1_{x-\mu^{S}-\rho_J y\geq 0}\\
& \qquad\qquad\qquad\qquad+ (1-\alpha)
\eta_{C,S} e^{-\eta_{C,S} (\mu^{S} + \rho_J y - x)} 1_{x-\mu^{S}-\rho_J y < 0}
\Big\}. \nonumber
\end{align}

\subsubsection{Testing.}
We present next the results of the testing for the Kou-type model. 
For simplicity, we take $\eta(\cdot)\equiv 1$ and $S_0=1, r=q=0$.
The model has the following common jump parameters:
\begin{align}
C \mbox{ jumps }: \lambda^{C},\, \eta_{C,V},\, \mu^{S},\, \rho_J,\, \eta_{C,S},\,
\alpha \,.
\end{align}

The numerical values of the model parameters used for testing are given in Table~\ref{tab:paramsK}.
The common jump intensity is chosen as in the Eraker model $\lambda^{C}=0.47$. We choose the asymmetry parameter $\alpha=0.5$ such that the jump size distribution is symmetric with respect to its mean, similar to the Eraker model.


\begin{table}[h!]
\centering
\caption{The common jumps and volatility process parameters for the Kou-type model. The intensity of the common jumps is taken as in the Eraker model $\lambda^{C}=0.47$.
The variance process parameters are $\mu^{V}=0,\sigma_V=0.01, V_0=0.0076$.}
    \begin{tabular}{|c||c|c|c|c||c|}
    \hline
parameter & $\eta_{C,V}$ & $\eta_{C,S}$ & $\mu^{S}$ & $\rho_{J}$ & $\alpha$ \\
    \hline
value & $20.0$ &  $10.0$ & $-0.11$ & $-0.38$ & $0.5$ \\
\hline
    \end{tabular}%
  \label{tab:paramsK}%
\end{table}%

The compensators and expectations introduced above 
that correspond to these parameters are $\mu^{C,S}=-0.112, m^{C,S}=-0.129$. 
This gives $\kappa = 2\lambda^{C}(\mu^{C,S}-m^{C,S}) = 0.016$.

\subsubsection{European options}
We illustrate the short-maturity asymptotic prediction for European options for the Kou-type model that is 
provided in Corollary~\ref{cor:Kou:European}.
The numerical results for 
$a_{E,C}(K)$ and $a_{E,P}(K)$ for OTM European options are shown in 
Table~\ref{tab:KT0p01} for several strikes, comparing with the results of an MC simulation for options with maturity $T=0.01$. 
The MC simulation used $10^5$ MC paths, and an Euler discretization with 100 time steps. The MC simulation results agree reasonably well within MC errors, with the asymptotic predictions for the Kou-type model in Corollary~\ref{cor:Kou:European}.

\begin{table}[htbp]
\centering
\caption{Numerical tests for European option pricing in the Kou-type model.
The options have maturity $T = 0.01$. 
The second and fifth columns show the asymptotic prediction, and the third and sixth columns shows the MC simulation results. 
The MC simulation uses 100k paths.}
\label{tab:kou_euro_t001_rhoNeg}
\renewcommand{\arraystretch}{1.2}
\setlength{\tabcolsep}{10pt}
\begin{tabular}{|c|c|c||c|c|c|}
\hline
$K$ & $1000a_{E,C}/\lambda^{C}$ & MC simulation  &
$K$ & $1000a_{E,P}/\lambda^{C}$ & MC simulation \\
\hline
1.05 & 10.0174 & 10.037 $\pm$ 1.099 &
0.95 & 86.6465 & 85.461 $\pm$ 2.487 \\
1.10 &  6.5906 &  6.646 $\pm$ 0.795 &
0.90 & 52.1012 & 51.362 $\pm$ 1.826 \\
1.15 &  4.4175 &  5.044 $\pm$ 0.805 &
0.85 & 28.1740 & 25.754 $\pm$ 1.264 \\
1.20 &  3.0118 &  2.554 $\pm$ 0.552 &
0.80 & 14.4798 & 13.738 $\pm$ 0.867 \\
1.25 &  2.0858 &  2.455 $\pm$ 0.858 &
0.75 & 7.1201 & 7.115 $\pm$ 0.632 \\
1.30 &  1.4654 &  1.076 $\pm$ 0.281 &
0.70 & 3.3335 & 3.543 $\pm$ 0.429 \\
1.35 &  1.0434 &  0.858 $\pm$ 0.266 &
0.65 & 1.4752 & 1.424 $\pm$ 0.232 \\
\hline
\end{tabular}
\label{tab:KT0p01}%
\end{table}

Figure~\ref{Fig:EK} shows the same results in graphical form, comparing the asymptotic and MC simulation results.

\begin{figure}
\includegraphics[scale=0.45]{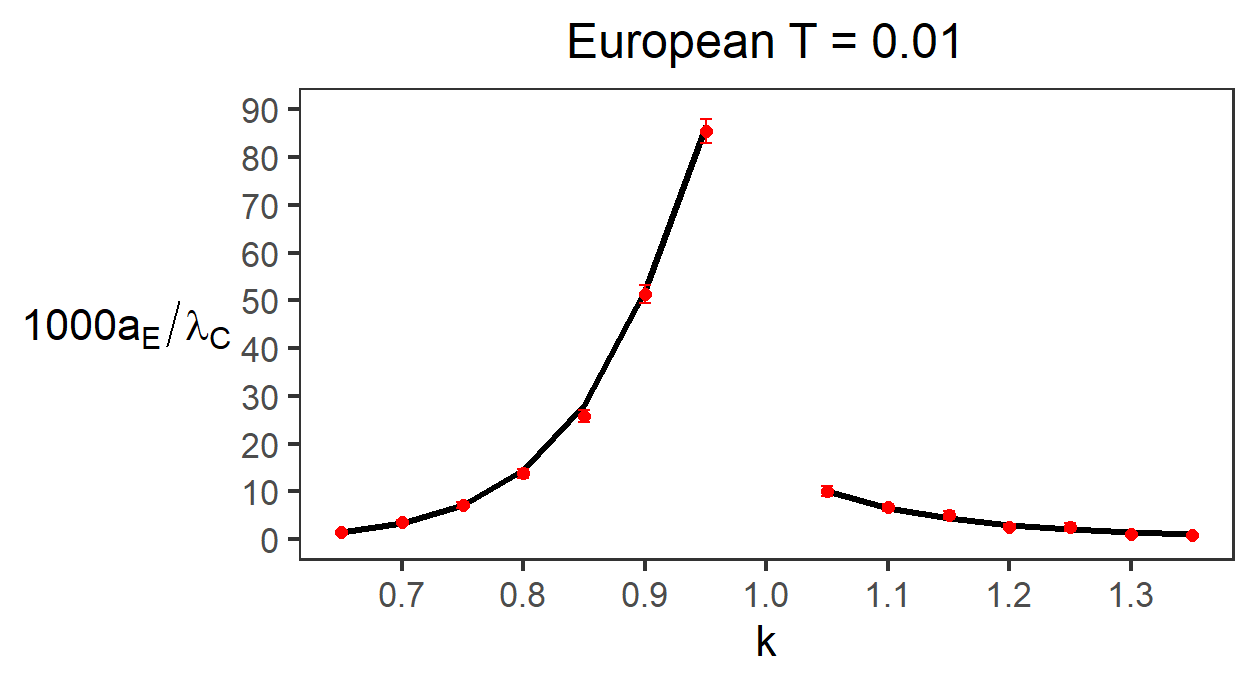}
\includegraphics[scale=0.45]{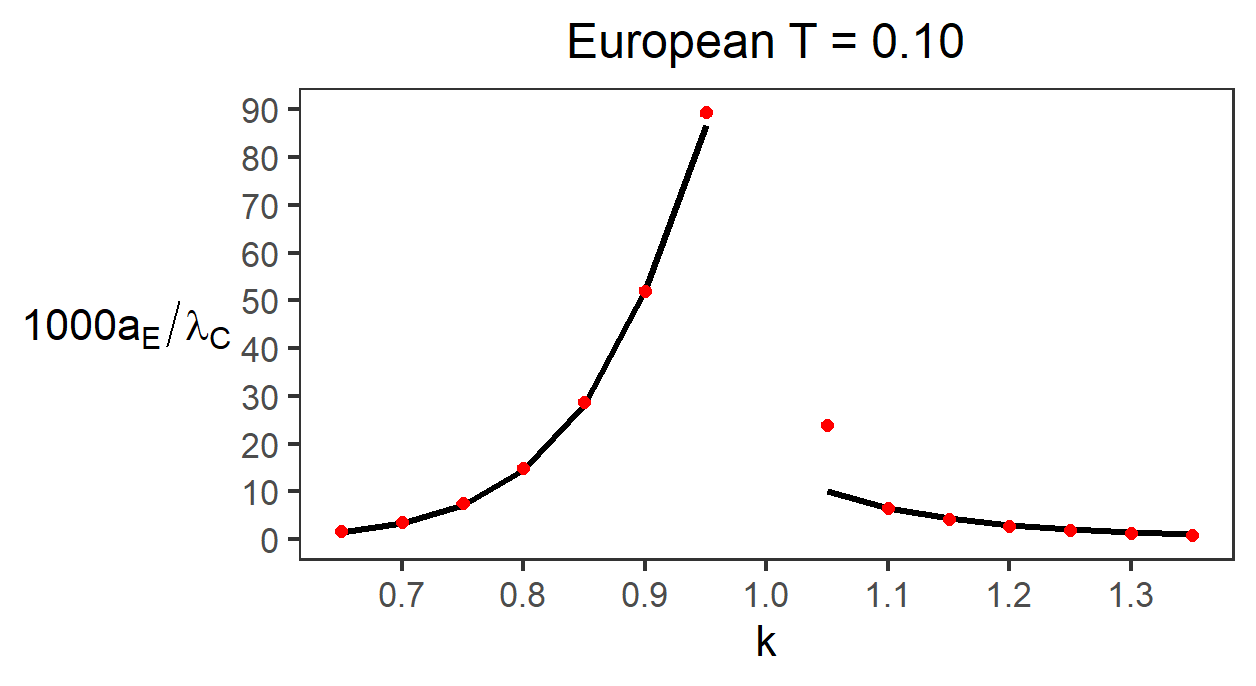}
\caption{Plots of the asymptotic results for European calls and puts with maturity $T=0.01$ (left) and  $T=0.1$ (right) in the Kou-type model.
The black curves show the asymptotic prediction for
$1000 \frac{a_{E,C}(K)}{\lambda^{C}}$ and
$1000 \frac{a_{E,P}(K)}{\lambda^{C}}$. The theoretical results are compared with 
the MC simulation (red dots) described in text.}
\label{Fig:EK}
\end{figure}


\subsubsection{VIX options}

We illustrate the short-maturity asymptotic prediction for VIX options for the Kou-type model that is 
provided in Corollary~\ref{cor:Kou:VIX}.

Using the numerical 
values of the parameters in Table~\ref{tab:paramsK} we get
$\kappa = 0.034 \lambda^{C} = 0.016$.
As a result, we obtain for the strike of an ATM VIX option 
$K_{\mathrm{VIX,ATM}} = \sqrt{V_0 + \kappa} = 0.1536$.
Numerical results for $a_{V,C}(K)$ are shown in Table~\ref{tab:VK} 
for several values of the moneyness $k=K/K_{\mathrm{ATM}} > 1$ corresponding to OTM VIX call options. We compare them with the results of an MC simulation.

  \begin{table}[h!]
  \centering
  \caption{Short maturity predictions for the VIX call options under the Kou-type model. 
  The last column shows the results of simulations 
  for VIX options with maturity $T=0.01$ (third column) and $T=0.1$ (fourth column) under the Kou-type model.
  Model parameters as in Table~\ref{tab:paramsK} and $\mu^{V}=0, \sigma_V=0.01$. }
    \begin{tabular}{|c|c||c|c|}
    \hline
$K/K_{\mathrm{ATM}}$ & $1000a_{V,C}/\lambda^{C}$ 
    & $1000 C_V^{\mathrm{MC}}(K,T)/(\lambda^{C} T)$ 
    & $1000 C_V^{\mathrm{MC}}(K,T)/(\lambda^{C} T)$\\
    & & $T=0.01$ & $T=0.1$ \\
    \hline\hline
1.00 & 1.2908 &  2.845 $\pm$ 0.039 & 1.422 $\pm$ 0.012 \\
1.02 & 0.1340  & 0.134 $\pm$ 0.013 & 0.140 $\pm$ 0.004 \\
1.04 & 0.0170 &  0.014 $\pm$ 0.004 & 0.016 $\pm$ 0.001 \\
1.06 &  0.0025 &  0.001 $\pm$ 0.001 & 0.002 $\pm$ 0.000 \\
1.08 & 0.0004 & 0.000 $\pm$ 0.000   &0.000 $\pm$ 0.000 \\
1.10 & 0.0001 & 0.000 $\pm$ 0.000 & 0.000 $\pm$ 0.000\\
    \hline
    \end{tabular}%
  \label{tab:VK}%
\end{table}%

We test the asymptotic prediction for the Kou-type model in Corollary~\ref{cor:Kou:VIX} by MC simulation of VIX option prices. 
For this simulation, the VIX index $\mathrm{VIX}_T$ was approximated using the small-$\tau$ approximation 
$\sqrt{V_T + \kappa}$; see Corollary~\ref{cor:VIXsmalltau}.
Taking $T=0.1$ and $N_{\mathrm{MC}}=10^5$ paths, the VIX forward is estimated as 
\begin{equation}
F_V(T) = \mathbb{E}[\mathrm{VIX}_T] = 0.13046\pm 0.00002\,,
\end{equation}
which is close to the VIX ATM strike $K_{\mathrm{ATM}} = 0.1308$.

VIX call option prices obtained by simulation are used to compute $\frac{C_V(K,T)}{\lambda^{C} T}$. The results are shown in the last column of Table~\ref{tab:VK}, and in Figure~\ref{Fig:aV:2}, and are consistent with the asymptotic predictions in Corollary~\ref{cor:Kou:VIX}. 
The corresponding values for VIX puts are very small and consistent with zero 
(not shown) which agrees with the prediction from the asymptotic results of Corollary~\ref{cor:Kou:VIX}.

\begin{figure}
\centering
\includegraphics[scale=0.45]{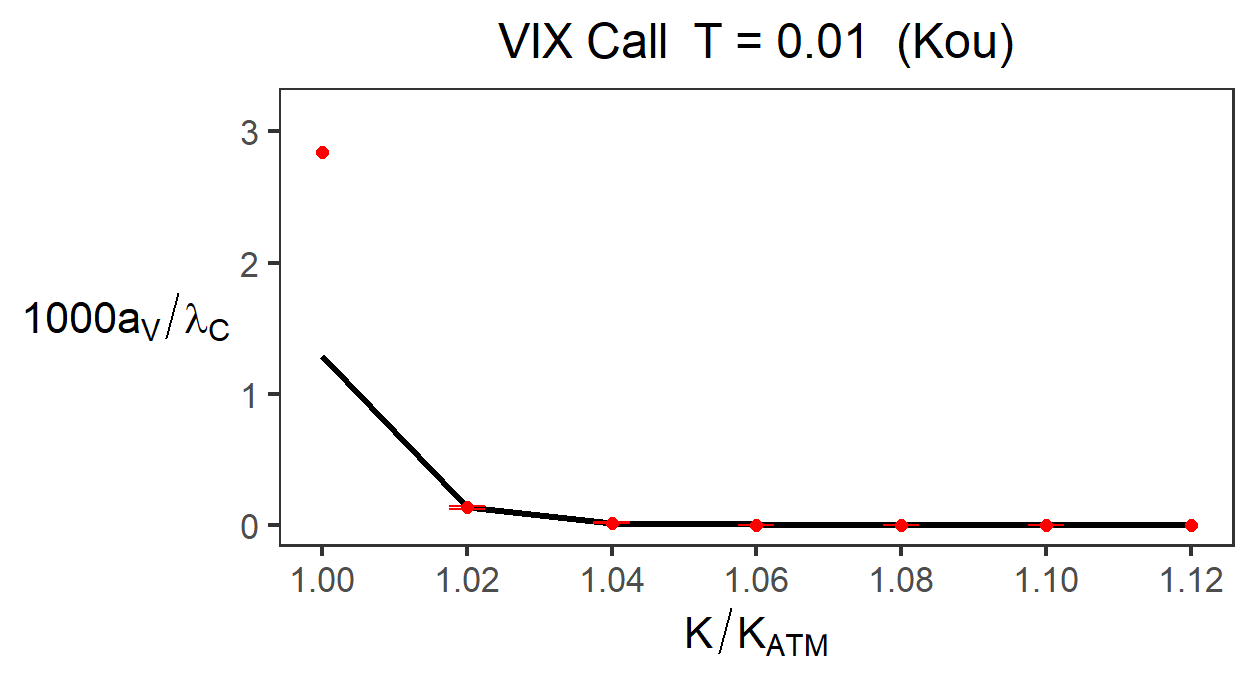}
\includegraphics[scale=0.45]{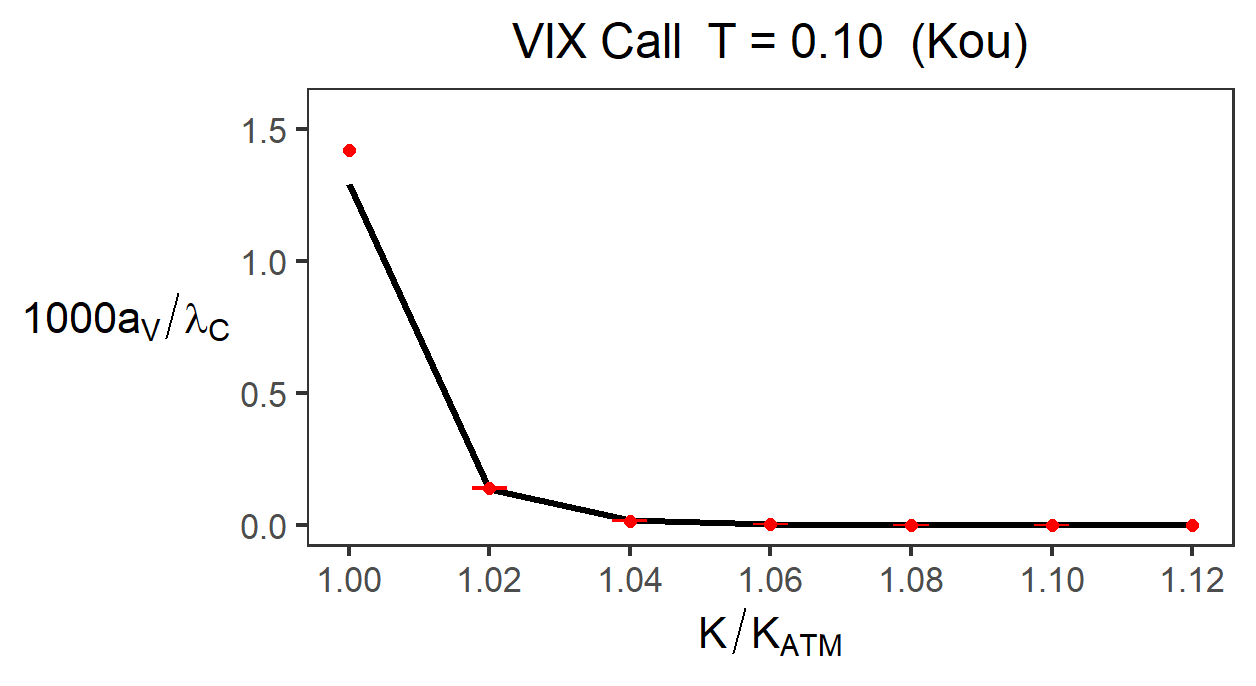}
\caption{Test for the asymptotic results for VIX calls in the Kou-type model. 
The black curve shows the asymptotic result for 
$1000 \frac{a_{V,C}(K)}{\lambda^{C}}$ and the red dots the result of an MC simulation.
The test uses VIX options with maturity $T=0.01$ (left) and $T=0.1$ (right).}
\label{Fig:aV:2}
\end{figure}


\subsection{Folded normal model}

\subsubsection{Parameters} 
The model has the following common jump parameters:
\begin{align}
C \mbox{ jumps }: \lambda^{C},\, \eta_{C,S},\, \sigma_{C,V},\, \mu^{S},\, \rho_J,\, \sigma_{C,S}\,.
\end{align}

We use the parameters of the folded normal model shown in Table~\ref{tab:paramsFN}. The parameter $\sigma_{C,V}$ was estimated such that the mean of the 
common jump in $V$ is the same as in the Eraker model. 
This gives the relation
\begin{equation}
\mathbb{E}\left[Y_1^{C,V}\right]  = \sqrt{\frac{2}{\pi}} \sigma_{C,V} = \frac{1}{\eta_{C,V}} \,.
\end{equation}
Using $\eta_{C,V}=20.0$ from the Eraker model yields $\sigma_{C,V}=0.063$. The remaining parameters $\rho_J,\mu^{S}$ were taken to have the same values as those used for the Kou-type model.

Using the parameters shown in Table~\ref{tab:paramsFN}, we obtain the numerical values of the compensators $\mu^{C,S}=-0.121, m^{C,S}=-0.129$. Using the relation 
$\kappa = 2\lambda^{S} (\mu^{S} - m^{S}) + 2\lambda^{C} (\mu^{C,S} - m^{C,S}) $ and $\lambda^{S}=0$ 
gives $\kappa = 0.0169\lambda^{C} = 0.0079$.

\begin{table}[h!]
\centering
\caption{The common jumps parameters for the 
folded normal model. The intensity of the common jumps is taken as 
in the Eraker model $\lambda^{C}=0.47$.
The variance process parameters are $\mu^{V}=0,\sigma_V=0.01, V_0=0.0076$.}
    \begin{tabular}{|c||c|c|c|c|c|}
    \hline
parameter & $\sigma_{C,V}$ & $\eta_{C,S}$ & $\mu^{S}$ & $\rho_{J}$ & $\sigma_{C,S}$  \\
    \hline
value & $0.063$ &  $10.0$ & $-0.11$ & $-0.38$ & 0.1 \\
\hline
    \end{tabular}%
  \label{tab:paramsFN}%
\end{table}%
\subsubsection{European options}
The asymptotic predictions for European options in the folded normal jump model are given in Corollary~\ref{cor:FN:Eur}. The numerical predictions are shown in Table~\ref{tab:FN_euro} for both calls and puts. These predictions are compared with the results of an MC simulation for European options with maturity $T=0.1$.

The MC simulation results are in reasonably good agreement with the asymptotic prediction for the European put options, but we note large differences for the call options. In order to investigate the reason for these differences we show in Table \ref{tab:FN_euro:2} also the test results for maturity $T=0.01$. The agreement is much better, for both call and put options. We conclude that large sub-leading corrections are responsible for the differences observed for $T=0.1$.

\begin{table}[htbp]
\centering
\caption{Numerical tests for European option pricing in the folded normal model.
The asymptotic results are shown in the second column (calls) and fifth column (puts). The options used for the MC simulation have maturity $T = 0.1$. 
The MC simulation uses 100k paths.}
\label{tab:FN_euro}
\renewcommand{\arraystretch}{1.2}
\setlength{\tabcolsep}{10pt}
\begin{tabular}{|c|c|c||c|c|c|}
\hline
$K$ & $1000a_{E,C}/\lambda^{C}$ & MC simulation  &
$K$ & $1000a_{E,P}/\lambda^{C}$ & MC simulation \\
\hline
1.02 & 2.970 & 112.903 $\pm$ 0.523 &
0.98 & 107.742 & 152.807 $\pm$ 1.386 \\
1.04 & 1.913 & 33.184 $\pm$ 0.113 &
0.96 & 90.800 & 94.924 $\pm$ 1.272 \\
1.06 & 1.208 & 7.786 $\pm$ 0.133 &
0.94 & 74.935 & 70.907 $\pm$ 1.148 \\
1.08 & 0.748 & 1.965 $\pm$ 0.140 &
0.92 & 60.371 & 56.245 $\pm$ 1.063 \\
1.10 & 0.454 & 0.786 $\pm$ 0.086 &
0.90 & 47.318 & 44.144 $\pm$ 0.934 \\
\hline
\end{tabular}
\end{table}

\begin{table}[htbp]
\centering
\caption{The options used for the MC simulation have maturity $T = 0.01$. 
}
\label{tab:FN_euro:2}
\renewcommand{\arraystretch}{1.2}
\setlength{\tabcolsep}{10pt}
\begin{tabular}{|c|c|c||c|c|c|}
\hline
$K$ & $1000a_{E,C}/\lambda^{C}$ & MC simulation  &
$K$ & $1000a_{E,P}/\lambda^{C}$ & MC simulation \\
\hline
1.02 & 2.970 & 12.696 $\pm$ 0.240 &
0.98 & 107.742 & 112.020 $\pm$ 4.042 \\
1.04 & 1.913 & 2.301 $\pm$ 0.121 &
0.96 & 90.800 & 89.587 $\pm$ 3.222 \\
1.06 & 1.208 & 1.429 $\pm$ 0.140 &
0.94 & 74.935 & 73.627 $\pm$ 2.484 \\
1.08 & 0.748 & 0.886 $\pm$ 0.193 &
0.92 & 60.371 & 59.204 $\pm$ 1.920 \\
1.10 & 0.454 & 0.515 $\pm$ 0.233 &
0.90 & 47.318 & 46.300 $\pm$ 1.413 \\
\hline
\end{tabular}
\end{table}
\subsubsection{VIX options}
The asymptotic predictions for VIX options in the folded normal model are given in Corollary~\ref{cor:FN:VIX}. 
The ATM strike of the VIX options in the $T\to 0$ limit is $K_{\mathrm{ATM}}=\sqrt{V_0+\kappa} = 0.1247$. The asymptotic predictions for OTM VIX call options with several strikes are shown in the second column in Table~\ref{tab:VFN}. These are compared with MC simulation results for VIX options with maturities $T=0.01$ (third column) and $T=0.1$ (fourth column).
\begin{table}[!ht]
  \centering
  \caption{Short maturity predictions for the VIX call options under the folded normal model.
  The last columns show the results of simulations 
  for VIX options with maturity $T=0.01$ (third column) and $T=0.1$ (fourth column).
  Model parameters as in Table~\ref{tab:paramsFN} and $\mu^{V}=0, \sigma_V=0.01$. }
    \begin{tabular}{|c|c||c|c|}
    \hline
$K/K_{\mathrm{ATM}}$ & $10^4\cdot a_{V,C}/\lambda^{C}$ 
    & $10^4\cdot C_V^{\mathrm{MC}}(K,T)/(\lambda^{C} T)$ 
    & $10^4\cdot C_V^{\mathrm{MC}}(K,T)/(\lambda^{C} T)$\\
    & & $T=0.01$ & $T=0.1$ \\
    \hline\hline
1.01 & 6.44 & 6.11 $\pm$ 0.16 & 6.10 $\pm$ 0.11 \\
1.02 & 2.08 & 1.96 $\pm$ 0.01 & 2.07 $\pm$ 0.04 \\
1.03 & 0.53 & 0.40 $\pm$ 0.03 & 0.60 $\pm$ 0.01 \\
1.04 & 0.10 & 0.02 $\pm$ 0.00 & 0.15 $\pm$ 0.01 \\
1.05 & 0.02 & 0.00 $\pm$ 0.00 & 0.04 $\pm$ 0.00 \\
1.06 & 0.00 & 0.00 $\pm$ 0.00 & 0.01 $\pm$ 0.00 \\
    \hline
    \end{tabular}%
  \label{tab:VFN}%
\end{table}%

We note good agreement of the MC simulation results with the asymptotic predictions for VIX call options. The MC simulations for the VIX put options give very small results, compatible with zero (not shown), which agree with the asymptotic prediction $a_{P,V}(K)=0$.

\section*{Acknowledgements}
Xiaoyu Wang is supported by the Guangzhou-HKUST(GZ) Joint Funding Program \\(No.2024A03J0630), 
Guangzhou Municipal Key Laboratory of Financial Technology\\ Cutting-Edge Research and the Guangzhou-HKUST(GZ) Joint Funding Program \\(No.2025A03J3556).
Lingjiong Zhu is partially supported by the grants NSF DMS-2053454, NSF DMS-2208303.



\bibliographystyle{plain}
\bibliography{ShortMaturityVIX}


\appendix

\section{Technical Lemmas}

In this appendix, we introduce a few technical lemmas
that will be used in the proofs of the main results in the paper.

\begin{lemma}\label{lem:p:bound}
For any $p>1$ such that $\mathbb{E}\left[e^{\frac{p}{2}Y_{1}^{V}}\right]<\infty$ and $\mathbb{E}\left[e^{\frac{p}{2}Y_{1}^{C,V}}\right]<\infty$, we have
\begin{align}
&\mathbb{E}\left[\left(\sqrt{\eta^{2}(S_{T})V_{T}+\kappa}+K\right)^{p}\right]
\nonumber
\\
&\leq 2^{p-1}M_{\eta}^{p}V_{0}^{p/2}e^{\lambda^{V}T\left(\mathbb{E}\left[e^{\frac{p}{2}Y_{1}^{V}}\right]-1\right)}e^{\lambda^{C}T\left(\mathbb{E}\left[e^{\frac{p}{2}Y_{1}^{C,V}}\right]-1\right)}e^{\frac{p}{2}M_{\mu}T+\frac{p^{2}}{8}M_{\sigma}^{2}T}
+2^{p-1}\left(\sqrt{\kappa}+K\right)^{p}.
\end{align}
\end{lemma}

Let $(\bar{V}_{t})_{t\geq 0}$ satisfy the following SDE:
\begin{eqnarray}
\frac{d\bar{V}_t}{\bar{V}_t} =\sigma(\bar{V}_t) d\bar{Z}_t + \mu(\bar{V}_t) dt \,,\nonumber
\end{eqnarray}
where $\bar{Z}_{t}$ is a standard Brownian motion, 
with $\bar{V}_{0}=V_{0}$.
We have the following two technical lemma.s

\begin{lemma}\label{lem:1}
For any $0\leq t\leq T$:
\begin{align}
\mathbb{E}[| \bar{V}_t - \bar{V}_0|]
\leq
V_{0}\left(e^{2TM_{\mu}}
e^{4TM_{\sigma}^{2}}+1-2e^{-TM_{\mu}-\frac{1}{2}T M_{\sigma}^{2}}\right)^{1/2}.
\end{align}
\end{lemma}

\begin{lemma}\label{lem:2}
For any $0\leq t\leq T$:
\begin{equation}
\mathbb{E}[\bar{V}_{t}]\leq V_{0}e^{M_{\mu}t},
\end{equation}
and
\begin{equation}
\mathbb{E}[V_{t}]\leq V_{0}e^{\lambda^{V}t\mu^{V}+\lambda^{C}t\mu^{C,V}}e^{M_{\mu}t}.
\end{equation}
\end{lemma}

\section{Proofs of the Results in Section~\ref{sec:main}}

\begin{proof}[Proof of Proposition~\ref{prop:VIXsmalltau}]
First of all, we have
\begin{align}
\left| \mathrm{VIX}_T^2 - V_T \eta^2(S_T)-\kappa \right|  & = \left|\frac{1}{\tau} \int_T^{T+\tau} \mathbb{E}[V_s \eta^2(S_s)|\mathcal{F}_{T}] ds+\kappa 
- V_T \eta^2(S_T)-\kappa \right|\nonumber \\
& \leq
\frac{1}{\tau} \int_T^{T+\tau} \left| \mathbb{E}[V_s \eta^2(S_s) |\mathcal{F}_{T}]  - V_T \eta^2(S_T) \right| ds\,. 
\end{align}

The integrand in the above equation can be bounded as 
\begin{equation}\label{18}
\left| \mathbb{E}[V_s \eta^2(S_s) |\mathcal{F}_{T}]  - V_T \eta^2(S_T) \right| 
\leq 
\left| \mathbb{E}[( V_s - V_T) \eta^2(S_s) |\mathcal{F}_{T}] \right| +
\left| \mathbb{E}[V_T (\eta^2(S_s) - \eta^2(S_T)) |\mathcal{F}_{T}] \right| \,.
\end{equation}
Next, we bound each term on the right-hand side in \eqref{18}.

\textbf{Step 1.} \textit{First term in (\ref{18}).}
The first term in \eqref{18} can be bounded as 
\begin{align}
\left|\mathbb{E}\left[V_{s}\eta^{2}(S_{s})|\mathcal{F}_{T}\right]
-\mathbb{E}\left[V_{T}\eta^{2}(S_{s})|\mathcal{F}_{T}\right]\right|
&\leq
M_{\eta}^{2}\left[\mathbb{E}|V_{s}-V_{T}||\mathcal{F}_{T}\right]
\nonumber
\\
&\leq
M_{\eta}^{2}\left(\left[\mathbb{E}|V_{s}-V_{T}|^{2}|\mathcal{F}_{T}\right]\right)^{1/2}.
\end{align}
We can further compute that
\begin{equation}
V_{s}=V_{T}e^{\int_{T}^{s}(\mu(V_{u})-\frac{1}{2}\sigma^{2}(V_{u}))du+\int_{T}^{s}\sigma(V_{u})dZ_{u}+R_{s}^{V}-R_{T}^{V}+R_{s}^{C,V}-R_{T}^{C,V}},
\end{equation}
where, for any $t$,
\begin{equation}
R_{t}^{V}:=\sum_{i=1}^{N_{t}^{S}}Y_{i}^{V},
\qquad
R_{t}^{C,V}:=\sum_{i=1}^{N_{t}^{C}}Y_{i}^{C,V}.
\end{equation}
We have
\begin{align}
\mathbb{E}[V_{s}|\mathcal{F}_{T}]
&=V_{T}\mathbb{E}\left[\mathbb{E}\left[e^{\int_{T}^{s}(\mu(V_{u})-\frac{1}{2}\sigma^{2}(V_{u}))du+\int_{T}^{s}\sigma(V_{u})dZ_{u}}|\mathcal{F}_{T},\mathcal{F}_{s}^{R,V}\right]e^{R_{s}^{V}-R_{T}^{V}+R_{s}^{C,V}-R_{T}^{C,V}}\right],
\end{align}
where 
$\mathcal{F}_{t}^{R,V}$ is the natural filtration of $(R_{t}^{V},R_{t}^{C,V})$ process.
By Jensen's inequality and Assumption~\ref{assump:bounded},
we have
\begin{align}
\mathbb{E}\left[e^{\int_{T}^{s}(\mu(V_{u})-\frac{1}{2}\sigma^{2}(V_{u}))du+\int_{T}^{s}\sigma(V_{u})dZ_{u}}|\mathcal{F}_{T},\mathcal{F}_{s}^{R,V}\right]
&\geq
e^{\mathbb{E}[\int_{T}^{s}(\mu(V_{u})-\frac{1}{2}\sigma^{2}(V_{u}))du
+\int_{T}^{s}\sigma(V_{u})dZ_{u}|\mathcal{F}_{T},\mathcal{F}_{s}^{R,V}]}
\nonumber
\\
&=
e^{\mathbb{E}[\int_{T}^{s}(\mu(V_{u})-\frac{1}{2}\sigma^{2}(V_{u}))du|\mathcal{F}_{T},\mathcal{F}_{s}^{R,V}]}
\nonumber
\\
&\geq
e^{-(s-T)M_{\mu}-\frac{1}{2}(s-T)M_{\sigma}^{2}}.
\end{align}
Hence, we conclude that
\begin{align}
\mathbb{E}[V_{s}|\mathcal{F}_{T}]
&\geq
V_{T}e^{-(s-T)M_{\mu}-\frac{1}{2}(s-T)M_{\sigma}^{2}}\mathbb{E}\left[e^{R_{s}^{V}-R_{T}^{V}+R_{s}^{C,V}-R_{T}^{C,V}}\right]
\nonumber
\\
&=
V_{T}e^{\lambda^{V}(s-T)(\mathbb{E}[e^{Y_{1}^{V}}]-1)}e^{\lambda^{C}(s-T)(\mathbb{E}[e^{Y_{1}^{C,V}}]-1)}e^{-(s-T)M_{\mu}-\frac{1}{2}(s-T)M_{\sigma}^{2}}.
\end{align}

On the other hand, we have
\begin{align}
\mathbb{E}[V_{s}^{2}|\mathcal{F}_{T}]
&=V_{T}^{2}\mathbb{E}\left[\mathbb{E}\left[e^{\int_{T}^{s}(2\mu(V_{u})-\sigma^{2}(V_{u}))du+2\int_{T}^{s}\sigma(V_{u})dZ_{u}}|\mathcal{F}_{T},\mathcal{F}_{s}^{R,V}\right]e^{2R_{s}^{V}-2R_{T}^{V}+2R_{s}^{C,V}-2R_{T}^{C,V}}\right],
\end{align}
where 
$\mathcal{F}_{t}^{R,V}$ is the natural filtration of $(R_{t}^{V},R_{t}^{C,V})$ process.
By Assumption~\ref{assump:bounded}, 
\begin{align}
\mathbb{E}\left[e^{\int_{T}^{s}(2\mu(V_{u})-\sigma^{2}(V_{u}))du
+2\int_{T}^{s}\sigma(V_{u})dZ_{u}}\Big|\mathcal{F}_{T},\mathcal{F}_{s}^{R,V}\right]
\leq
e^{2(s-T)M_{\mu}}
\mathbb{E}\left[e^{2\int_{T}^{s}\sigma(V_{u})dZ_{u}}\Big|\mathcal{F}_{T},\mathcal{F}_{s}^{R,V}\right],
\end{align}
and by Cauchy-Schwarz inequality, we can further compute that
\begin{align}
&\mathbb{E}\left[e^{2\int_{T}^{s}\sigma(V_{u})dZ_{u}}\Big|\mathcal{F}_{T},\mathcal{F}_{s}^{R,V}\right]
\nonumber
\\
&=\mathbb{E}\left[e^{\int_{T}^{s}2\sigma(V_{u})dZ_{u}-\int_{T}^{s}4\sigma^{2}(V_{u})du}e^{\int_{T}^{s}4\sigma^{2}(V_{u})du}\Big|\mathcal{F}_{T},\mathcal{F}_{s}^{R,V}\right]
\nonumber
\\
&\leq
\left(\mathbb{E}\left[e^{\int_{T}^{s}4\sigma(V_{u})dZ_{u}-\frac{1}{2}\int_{T}^{s}(4\sigma)^{2}(V_{u})du}\Big|\mathcal{F}_{T},\mathcal{F}_{s}^{R,V}\right]\right)^{1/2}
\left(\mathbb{E}\left[e^{8\int_{T}^{s}\sigma^{2}(V_{u})du}\Big|\mathcal{F}_{T},\mathcal{F}_{s}^{R,V}\right]\right)^{1/2}
\nonumber
\\
&=\left(\mathbb{E}\left[e^{8\int_{T}^{s}\sigma^{2}(V_{u})du}\Big|\mathcal{F}_{T},\mathcal{F}_{s}^{R,V}\right]\right)^{1/2}
\leq e^{4(s-T)M_{\sigma}^{2}}.
\end{align}
Hence, 
\begin{align}
\mathbb{E}[V_{s}^{2}|\mathcal{F}_{T}]
&\leq
V_{T}^{2}\mathbb{E}\left[e^{2R_{s}^{V}-2R_{T}^{V}+2R_{s}^{C,V}-2R_{T}^{C,V}}\right]e^{2(s-T)M_{\mu}}
e^{4(s-T)M_{\sigma}^{2}}
\nonumber
\\
&=V_{T}^{2}e^{\lambda^{V}(s-T)(\mathbb{E}[e^{2Y_{1}^{V}}]-1)}e^{\lambda^{C}(s-T)(\mathbb{E}[e^{2Y_{1}^{C,V}}]-1)}e^{2(s-T)M_{\mu}}
e^{4(s-T)M_{\sigma}^{2}}.
\end{align}
Therefore, for any $T\leq s\leq T+\tau$,
\begin{align}
&\left[\mathbb{E}|V_{s}-V_{T}|^{2}|\mathcal{F}_{T}\right]
\nonumber
\\
&=\mathbb{E}\left[V_{s}^{2}|\mathcal{F}_{T}\right]
+V_{T}^{2}
-2V_{T}\mathbb{E}[V_{s}|\mathcal{F}_{T}]
\nonumber
\\
&\leq
V_{T}^{2}\Bigg(e^{\lambda^{V}(s-T)(\mathbb{E}[e^{2Y_{1}^{V}}]-1)}e^{\lambda^{C}(s-T)(\mathbb{E}[e^{2Y_{1}^{C,V}}]-1)}e^{2(s-T)M_{\mu}}
e^{4(s-T)M_{\sigma}^{2}}+1
\nonumber
\\
&\qquad\qquad\qquad
-2e^{-\lambda^{V}(s-T)(\mathbb{E}[e^{Y_{1}^{V}}]-1)}e^{-\lambda^{C}(s-T)(\mathbb{E}[e^{Y_{1}^{C,V}}]-1)}e^{-(s-T)M_{\mu}-\frac{1}{2}(s-T)M_{\sigma}^{2}}\Bigg)
\nonumber
\\
&\leq
V_{T}^{2}\Bigg(e^{\lambda^{V}\tau(\mathbb{E}[e^{2Y_{1}^{V}}]-1)}e^{\lambda^{C}\tau(\mathbb{E}[e^{2Y_{1}^{C,V}}]-1)}e^{2\tau M_{\mu}}
e^{4\tau M_{\sigma}^{2}}+1
\nonumber
\\
&\qquad\qquad\qquad-2e^{-\lambda^{V}\tau(\mathbb{E}[e^{Y_{1}^{V}}]-1)}e^{-\lambda^{C}\tau(\mathbb{E}[e^{Y_{1}^{C,V}}]-1)}e^{-\tau M_{\mu}-\frac{1}{2}\tau M_{\sigma}^{2}}\Bigg).
\end{align}
Hence, we conclude that, for any $T\leq s\leq T+\tau$,
\begin{align}
&\left|\mathbb{E}\left[V_{s}\eta^{2}(S_{s})|\mathcal{F}_{T}\right]
-\mathbb{E}\left[V_{T}\eta^{2}(S_{s})|\mathcal{F}_{T}\right]\right|
\nonumber
\\
&\leq
M_{\eta}^{2}V_{T}\Bigg(e^{\lambda^{V}\tau(\mathbb{E}[e^{2Y_{1}^{V}}]-1)}e^{\lambda^{C}\tau(\mathbb{E}[e^{2Y_{1}^{C,V}}]-1)}e^{2\tau M_{\mu}}
e^{4\tau M_{\sigma}^{2}}+1
\nonumber
\\
&\qquad\qquad\qquad\qquad\qquad-2e^{-\lambda^{V}\tau(\mathbb{E}[e^{Y_{1}^{V}}]-1)}e^{-\lambda^{C}\tau(\mathbb{E}[e^{Y_{1}^{C,V}}]-1)}e^{-\tau M_{\mu}-\frac{1}{2}\tau M_{\sigma}^{2}}\Bigg)^{1/2}.
\end{align}

\textbf{Step 2.}
\textit{The second term in (\ref{18})}.
The second term in (\ref{18}) is bounded further as
\begin{align}
\left|\mathbb{E}\left[V_{T}\eta^{2}(S_{s})|\mathcal{F}_{T}\right]
-V_{T}\eta^{2}(S_{T})\right|
=V_{T}
\left|\mathbb{E}\left[\eta^{2}(S_{s})|\mathcal{F}_{T}\right]
-\eta^{2}(S_{T})\right|.
\end{align}
By It\^{o}'s formula and taking expectations, we obtain 
\begin{align}
&\mathbb{E}\left[\eta^{2}(S_{s})|\mathcal{F}_{T}\right]
-\eta^{2}(S_{T})
\nonumber
\\
&=\int_{T}^{s}\mathbb{E}\left[(\eta^{2})'(S_{t})(r-q-\lambda^{S}\mu^{S}-\lambda^{C}\mu^{C,s})S_{t}
+\frac{1}{2}(\eta^{2})''(S_{t})\eta^{2}(S_{t})S_{t}^{2}V_{t}\bigg|\mathcal{F}_{T}\right]dt
\nonumber
\\
&\qquad\qquad+\mathbb{E}\left[\int_{T}^{s}(\eta^{2})'(S_{t})\eta(S_{t})S_{t}\sqrt{V_{t}}dW_{t}\bigg|\mathcal{F}_{T}\right]
\nonumber
\\
&\qquad\qquad\qquad
+\mathbb{E}\left[\int_{T}^{s}\int_{\mathbb{R}}\left(\eta^{2}(S_{t}e^{x})-\eta^{2}(S_{t})\right)P^{S}(dx)dN_{t}^{S}\bigg|\mathcal{F}_{T}\right]
\nonumber
\\
&\qquad\qquad\qquad\qquad
+\mathbb{E}\left[\int_{T}^{s}\int_{\mathbb{R}}\left(\eta^{2}(S_{t}e^{x})-\eta^{2}(S_{t})\right)p^{C,S}(x)dxdN_{t}^{C}\bigg|\mathcal{F}_{T}\right],
\end{align}
which further implies that
\begin{align}
&\mathbb{E}\left[\eta^{2}(S_{s})|\mathcal{F}_{T}\right]
-\eta^{2}(S_{T})
\nonumber
\\
&=\int_{T}^{s}\mathbb{E}\left[(\eta^{2})'(S_{t})(r-q-\lambda^{S}\mu^{S}-\lambda^{C}\mu^{C,s})S_{t}
+\frac{1}{2}(\eta^{2})''(S_{t})\eta^{2}(S_{t})S_{t}^{2}V_{t}\bigg|\mathcal{F}_{T}\right]dt
\nonumber
\\
&\qquad\qquad
+\lambda^{S}\int_{T}^{s}\int_{\mathbb{R}}\mathbb{E}\left[\left(\eta^{2}(S_{t}e^{x})-\eta^{2}(S_{t})\right)\bigg|\mathcal{F}_{T}\right]P^{S}(dx)dt
\nonumber
\\
&\qquad\qquad\qquad
+\lambda^{C}\int_{T}^{s}\int_{\mathbb{R}}\mathbb{E}\left[\left(\eta^{2}(S_{t}e^{x})-\eta^{2}(S_{t})\right)\bigg|\mathcal{F}_{T}\right]p^{C,S}(x)dxdt.
\end{align}

By our assumption, $\eta$ is $L$-Lipschitz, 
so that $(\eta^{2})' = 2\eta\eta'$, which implies that $\eta^{2}$ is $2LM_{\eta}$-Lipschitz.
Also we have the bound (\ref{eta:2:assumption}) on the second derivative $(\eta^2)''(s)$.
Furthermore, since $\eta(\cdot)$ is uniformly bounded by $M_{\eta}$, 
we have $\left|\mathbb{E}\left[\left(\eta^{2}(S_{t}e^{x})-\eta^{2}(S_{t})\right)\bigg|\mathcal{F}_{T}\right]\right|\leq M_{\eta}^{2}$.
Therefore, for any $T\leq s\leq T+\tau$, we have
\begin{align}
&\left|\mathbb{E}\left[\eta^{2}(S_{s})|\mathcal{F}_{T}\right]
-\eta^{2}(S_{T})\right|
\nonumber
\\
&\leq
\int_{T}^{s}\mathbb{E}\left[2LM_{\eta}|r-q-\lambda^{S}\mu^{S}-\lambda^{C}\mu^{C,s}|S_{t}
+\frac{1}{2}M_{\eta,2}M_{\eta}^{2}V_{t}\bigg|\mathcal{F}_{T}\right]dt
+(\lambda^{S}+\lambda^{C})M_{\eta}^{2}\tau
\nonumber
\\
&
=\int_{T}^{s}\left[2LM_{\eta}|r-q-\lambda^{S}\mu^{S}-\lambda^{C}\mu^{C,s}|S_{T}e^{(r-q)(t-T)}
+\frac{1}{2}M_{\eta,2}M_{\eta}^{2}\mathbb{E}[V_{t}|\mathcal{F}_{T}]\right]dt
+(\lambda^{S}+\lambda^{C})M_{\eta}^{2}\tau
\nonumber
\\
&\leq\int_{T}^{s}\Bigg[2LM_{\eta}\left|r-q-\lambda^{S}\mu^{S}-\lambda^{C}\mu^{C,s}\right|S_{T}e^{(r-q)(t-T)}
\nonumber
\\
&\qquad\qquad\qquad\qquad\qquad
+\frac{1}{2}M_{\eta,2}M_{\eta}^{2}V_{T}e^{\lambda^{V}\tau\mu^{V}}e^{\lambda^{C}\tau\mu^{C,V}}e^{\tau M_{\mu}}\Bigg]dt
+(\lambda^{S}+\lambda^{C})M_{\eta}^{2}\tau
\nonumber
\\
&\leq
\tau
\Bigg[2LM_{\eta}\left|r-q-\lambda^{S}\mu^{S}-\lambda^{C}\mu^{C,s}\right|S_{T}e^{|r-q|\tau}
\nonumber
\\
&\qquad\qquad\qquad\qquad
+\frac{1}{2}M_{\eta,2}M_{\eta}^{2}V_{T}e^{\lambda^{V}\tau\mu^{V}}e^{\lambda^{C}\tau\mu^{C,V}}e^{\tau M_{\mu}}
+(\lambda^{S}+\lambda^{C})M_{\eta}^{2}\Bigg],
\end{align}
where we also applied Lemma~\ref{lem:2}.

Hence, we conclude that for any $T\leq s\leq T+\tau$,
\begin{align}
&\left|\mathbb{E}\left[V_{s}\eta^{2}(S_{s})|\mathcal{F}_{T}\right]
-V_{T}\eta^{2}(S_{T})\right|
\nonumber
\\
&\leq
M_{\eta}^{2}V_{T}\Bigg(e^{\lambda^{V}\tau(\mathbb{E}[e^{2Y_{1}^{V}}]-1)}e^{\lambda^{C}\tau(\mathbb{E}[e^{2Y_{1}^{C,V}}]-1)}e^{2\tau M_{\mu}}
e^{4\tau M_{\sigma}^{2}}+1
\nonumber
\\
&\qquad\qquad\qquad\qquad\qquad-2e^{-\lambda^{V}\tau(\mathbb{E}[e^{Y_{1}^{V}}]-1)}e^{-\lambda^{C}\tau(\mathbb{E}[e^{Y_{1}^{C,V}}]-1)}e^{-\tau M_{\mu}-\frac{1}{2}\tau M_{\sigma}^{2}}\Bigg)^{1/2}
\nonumber
\\
&\qquad\qquad
+\tau\Bigg[2LM_{\eta}\left|r-q-\lambda^{S}\mu^{S}-\lambda^{C}\mu^{C,s}\right|S_{T}e^{|r-q|\tau}
\nonumber
\\
&\qquad\qquad\qquad\qquad
+\frac{1}{2}M_{\eta,2}M_{\eta}^{2}V_{T}e^{\lambda^{V}\tau\mu^{V}}e^{\lambda^{C}\tau\mu^{C,V}}e^{\tau M_{\mu}}
+(\lambda^{S}+\lambda^{C})M_{\eta}^{2}\Bigg].
\end{align}
By recalling the formula in \eqref{VIX:formula}, we conclude that
\begin{equation}
\left|\mathrm{VIX}_T^2-V_{T}\eta^{2}(S_{T})\right|
\leq C_{1}(\tau)S_{T}+C_{2}(\tau)V_{T}+(\lambda^{S}+\lambda^{C})M_{\eta}^{2}\tau,
\end{equation}
where $C_{1}(\tau),C_{2}(\tau)$ are defined in \eqref{C:1:tau}-\eqref{C:2:tau}. 
Finally, we can compute that
\begin{align}
&\mathbb{E}\left|\mathrm{VIX}_T^2-V_{T}\eta^{2}(S_{T})\right|
\nonumber
\\
&\leq C_{1}(\tau)\mathbb{E}[S_{T}]+C_{2}(\tau)\mathbb{E}[V_{T}]
+\left(\lambda^{S}+\lambda^{C}\right)M_{\eta}^{2}\tau
\nonumber
\\
&\leq
C_{1}(\tau)S_{0}e^{(r-q)T}
+C_{2}(\tau)V_{0}e^{\lambda^{V}T\mu^{V}}e^{\lambda^{C}T\mu^{C,V}}e^{TM_{\mu}}+\left(\lambda^{S}+\lambda^{C}\right)M_{\eta}^{2}\tau,
\end{align}
where we recall fromn \eqref{mu:defn} that $\mu^{V}=\mathbb{E}[e^{Y_{1}^{V}}]-1$
and $\mu^{C,V}=\mathbb{E}[e^{Y_{1}^{C,V}}]-1$ and we also applied Lemma~\ref{lem:2}.
This completes the proof.
\end{proof}

\begin{proof}[Proof of Corollary~\ref{cor:VIXsmalltau}]
Notice that
\begin{align}
\left|\mathrm{VIX}_T-\sqrt{\eta^{2}(S_{T})V_{T}+\kappa}\right|
=\frac{\left|\mathrm{VIX}_T^2
-\eta^{2}(S_{T})V_{T}-\kappa\right|}{\mathrm{VIX}_T+\sqrt{\eta^{2}(S_{T})V_{T}+\kappa}}
\leq\frac{\left|\mathrm{VIX}_T^2
-\eta^{2}(S_{T})V_{T}-\kappa\right|}{\sqrt{\kappa}}.
\end{align}
Then the conclusion follows from Proposition~\ref{prop:VIXsmalltau}.
\end{proof}

\begin{proof}[Proof of Theorem~\ref{thm:OTM}]
Let us prove the asymptotics for the call option.
The derivation for the put option is similar and omitted here.

First, notice that we have the following decomposition:
\begin{align}
\hat{C}_{V}(K,T)
&=e^{-rT}\mathbb{E}\left[\left(\sqrt{\eta^{2}(S_{T})V_{T}+\kappa}-K\right)^{+}1_{N_{T}^{S}=0}1_{N_{T}^{C}=0}1_{N_{T}^{V}=0}\right]
\nonumber
\\
&\quad
+e^{-rT}\mathbb{E}\left[\left(\sqrt{\eta^{2}(S_{T})V_{T}+\kappa}-K\right)^{+}1_{N_{T}^{S}=1}1_{N_{T}^{C}=0}1_{N_{T}^{V}=0}\right]
\nonumber
\\
&\quad\quad
+e^{-rT}\mathbb{E}\left[\left(\sqrt{\eta^{2}(S_{T})V_{T}+\kappa}-K\right)^{+}1_{N_{T}^{S}=0}1_{N_{T}^{C}=1}1_{N_{T}^{V}=0}\right]
\nonumber
\\
&\quad\quad\quad
+e^{-rT}\mathbb{E}\left[\left(\sqrt{\eta^{2}(S_{T})V_{T}+\kappa}-K\right)^{+}1_{N_{T}^{S}=0}1_{N_{T}^{C}=0}1_{N_{T}^{V}=1}\right]
\nonumber
\\
&\quad\quad\quad\quad
+e^{-rT}\mathbb{E}\left[\left(\sqrt{\eta^{2}(S_{T})V_{T}+\kappa}-K\right)^{+}1_{N_{T}^{S}+N_{T}^{C}+N_{T}^{V}\geq 2}\right].
\end{align}

\textbf{Step 1.} By using the large deviations for the local-stochastic volatility model without jumps, we can show that
\begin{equation}\label{to:show:1}
e^{-rT}\mathbb{E}\left[\left(\sqrt{\eta^{2}(S_{T})V_{T}+\kappa}-K\right)^{+}1_{N_{T}^{S}=0}1_{N_{T}^{C}=0}1_{N_{T}^{V}=0}\right]=o(T),
\end{equation}
as $T\rightarrow 0$. Next, let us prove \eqref{to:show:1}.
Notice that
\begin{align*}
&\mathbb{E}\left[\left(\sqrt{\eta^{2}(S_{T})V_{T}+\kappa}-K\right)^{+}1_{N_{T}^{S}=0}1_{N_{T}^{C}=0}1_{N_{T}^{V}=0}\right]
\\
&=\mathbb{E}\left[\left(\sqrt{\eta^{2}(S_{T})V_{T}+\kappa}-K\right)^{+}\Big|N_{T}^{S}=N_{T}^{C}=N_{T}^{V}=0\right]
\mathbb{Q}\left(N_{T}^{S}=N_{T}^{C}=N_{T}^{V}=0\right)
\\
&\leq\mathbb{E}\left[\left(\sqrt{\eta^{2}(S_{T})V_{T}+\kappa}-K\right)^{+}\Big|N_{T}^{S}=N_{T}^{C}=N_{T}^{V}=0\right]
\\
&=\mathbb{E}\left[\left(\sqrt{\eta^{2}(\tilde{S}_{T})\tilde{V}_{T}+\kappa}-K\right)^{+}\right],
\end{align*}
where $(\tilde{S}_{t},\tilde{V}_{t})_{t\geq 0}$ is the local-stochastic volatility model without jumps:
\begin{eqnarray}
&& \frac{d\tilde{S}_t}{\tilde{S}_t} = \eta(\tilde{S}_t) \sqrt{\tilde{V}_t} dW_t +\left(r - q-\lambda^{S}\mu^{S}-\lambda^{C}\mu^{C,S}\right) dt\,, \\
&& \frac{d\tilde{V}_t}{\tilde{V}_t} =\sigma(\tilde{V}_t) dZ_t + \mu(\tilde{V}_t) dt \,,\nonumber
\end{eqnarray}
and it is proved in Theorem~4.1. in \cite{VIXpaper} that
\begin{equation}
\mathbb{E}\left[\left(\sqrt{\eta^{2}(\tilde{S}_{T})\tilde{V}_{T}+\kappa}-K\right)^{+}\right]=e^{-O(1/T)},
\end{equation}
as $T\rightarrow 0$. Therefore, we proved \eqref{to:show:1}.

\textbf{Step 2.} We can show that
\begin{align}
&\lim_{T\rightarrow 0}\frac{1}{T}e^{-rT}\mathbb{E}\left[\left(\sqrt{\eta^{2}(S_{T})V_{T}+\kappa}-K\right)^{+}1_{N_{T}^{S}=1}1_{N_{T}^{C}=0}1_{N_{T}^{V}=0}\right]
\nonumber
\\
&=\lambda^{S}\int_{\mathbb{R}}\left(\sqrt{\eta^{2}\left(S_{0}e^{x}\right)V_{0}+\kappa}-K\right)^{+}P^{S}(dx).\label{to:show:2}
\end{align}
Next, let us prove \eqref{to:show:2}.
We can compute that
\begin{align*}
&\mathbb{E}\left[\left(\sqrt{\eta^{2}(S_{T})V_{T}+\kappa}-K\right)^{+}1_{N_{T}^{S}=1}1_{N_{T}^{C}=0}1_{N_{T}^{V}=0}\right]
\\
&=\mathbb{E}\left[\left(\sqrt{\eta^{2}(S_{T})V_{T}+\kappa}-K\right)^{+}\Big|N_{T}^{S}=1,N_{T}^{C}=N_{T}^{V}=0\right]
\mathbb{Q}\left(N_{T}^{S}=1,N_{T}^{C}=N_{T}^{V}=0\right).
\end{align*}
Since $N_{t}^{S}$ is a Poisson process with constant intensity $\lambda^{S}$, 
conditional on $N_{T}^{S}=1$, the location of the jump is distributed
uniformly on the time interval $[0,T]$. 
Therefore, we have
\begin{align*}
&\mathbb{E}\left[\left(\sqrt{\eta^{2}(S_{T})V_{T}+\kappa}-K\right)^{+}\Big|N_{T}^{S}=1,N_{T}^{C}=N_{T}^{V}=0\right]
\\
&=\frac{1}{T}\int_{0}^{T}\int_{\mathbb{R}}\mathbb{E}\left[\mathbb{E}\left[\left(\sqrt{\eta^{2}(\tilde{S}_{T})\tilde{V}_{T}+\kappa}-K\right)^{+}\Big|\tilde{S}_{t}=\bar{S}_{t}e^{x},\tilde{V}_{t}=\bar{V}_{t}\right]\right]P^{S}(dx)dt,
\end{align*}
where $(\tilde{S}_{t},\tilde{V}_{t})_{t\geq 0}$ and $(\bar{S}_{t},\bar{V}_{t})_{t\geq 0}$ are two independent local-stochastic volatility models without jumps, i.e.
\begin{eqnarray}
&& \frac{d\tilde{S}_t}{\tilde{S}_t} = \eta(\tilde{S}_t) \sqrt{\tilde{V}_t} dW_t +\left(r - q-\lambda^{S}\mu^{S}-\lambda^{C}\mu^{C,S}\right) dt\,,\label{tilde:S:defn} \\
&& \frac{d\tilde{V}_t}{\tilde{V}_t} =\sigma(\tilde{V}_t) dZ_t + \mu(\tilde{V}_t) dt \,,\nonumber
\end{eqnarray}
with $\tilde{S}_{0}=S_{0}$, $\tilde{V}_{0}=V_{0}$ and
\begin{eqnarray}
&& \frac{d\bar{S}_t}{\bar{S}_t} = \eta(\bar{S}_t) \sqrt{\bar{V}_t} d\bar{W}_t +\left(r - q-\lambda^{S}\mu^{S}-\lambda^{C}\mu^{C,S}\right) dt\,,\label{bar:S:defn} \\
&& \frac{d\bar{V}_t}{\bar{V}_t} =\sigma(\bar{V}_t) d\bar{Z}_t + \mu(\bar{V}_t) dt \,,\nonumber
\end{eqnarray}
where $\bar{W}_{t},\bar{Z}_{t}$ are two standard Brownian motions
with correlation $\rho$, independent of $W_{t},Z_{t}$, 
with $\bar{S}_{0}=S_{0}$, $\bar{V}_{0}=V_{0}$.
By the time-homogeneity, we have
\begin{align*}
&\mathbb{E}\left[\mathbb{E}\left[\left(\sqrt{\eta^{2}(\tilde{S}_{T})\tilde{V}_{T}+\kappa}-K\right)^{+}\Big|\tilde{S}_{t}=\bar{S}_{t}e^{x},\tilde{V}_{t}=\bar{V}_{t}\right]\right]
\\
&=\mathbb{E}\left[\mathbb{E}\left[\left(\sqrt{\eta^{2}(\tilde{S}_{T-t})\tilde{V}_{T-t}+\kappa}-K\right)^{+}\Big|\tilde{S}_{0}=\bar{S}_{t}e^{x},\tilde{V}_{0}=\bar{V}_{t}\right]\right].
\end{align*}
We can show that
\begin{align}
&\frac{1}{T}\int_{0}^{T}\int_{\mathbb{R}}\mathbb{E}\left[\mathbb{E}\left[\left(\sqrt{\eta^{2}(\tilde{S}_{T-t})\tilde{V}_{T-t}+\kappa}-K\right)^{+}\Big|\tilde{S}_{0}=\bar{S}_{t}e^{x},\tilde{V}_{0}=\bar{V}_{t}\right]\right]P^{S}(dx)dt\nonumber
\\
&\rightarrow\int_{\mathbb{R}}\left(\sqrt{\eta^{2}\left(S_{0}e^{x}\right)V_{0}+\kappa}-K\right)^{+}P^{S}(dx),\label{to:show:convergence}    
\end{align}
as $T\rightarrow 0$.

Next, let us prove \eqref{to:show:convergence}.
First, we have
\begin{align}
&\Bigg|\mathbb{E}\left[\mathbb{E}\left[\left(\sqrt{\eta^{2}(\tilde{S}_{T-t})\tilde{V}_{T-t}+\kappa}-K\right)^{+}\Big|\tilde{S}_{0}=\bar{S}_{t}e^{x},\tilde{V}_{0}=\bar{V}_{t}\right]\right]
\nonumber
\\
&\qquad\qquad
-\mathbb{E}\left[\left(\sqrt{\eta^{2}(\bar{S}_{t}e^{x})\bar{V}_{t}+\kappa}-K\right)^{+}\right]\Bigg|
\nonumber
\\
&=\Bigg|\mathbb{E}\left[\mathbb{E}\left[\left(\sqrt{\eta^{2}(\tilde{S}_{T-t})\tilde{V}_{T-t}+\kappa}-K\right)^{+}\Big|\tilde{S}_{0}=\bar{S}_{t}e^{x},\tilde{V}_{0}=\bar{V}_{t}\right]\right]
\nonumber
\\
&\qquad\qquad
-\mathbb{E}\left[\mathbb{E}\left[\left(\sqrt{\eta^{2}(\tilde{S}_{0})\tilde{V}_{0}+\kappa}-K\right)^{+}\Big|\tilde{S}_{0}=\bar{S}_{t}e^{x},\tilde{V}_{0}=\bar{V}_{t}\right]\right]\Bigg|
\nonumber
\\
&\leq\mathbb{E}\left[\mathbb{E}\left[\left|\sqrt{\eta^{2}(\tilde{S}_{T-t})\tilde{V}_{T-t}+\kappa}-\sqrt{\eta^{2}(\tilde{S}_{0})\tilde{V}_{0}+\kappa}\right|\Big|\tilde{S}_{0}=\bar{S}_{t}e^{x},\tilde{V}_{0}=\bar{V}_{t}\right]\right]
\nonumber
\\
&\leq\frac{1}{2\sqrt{\kappa}}\mathbb{E}\left[\mathbb{E}\left[\left|\eta^{2}(\tilde{S}_{T-t})\tilde{V}_{T-t}-\eta^{2}(\tilde{S}_{0})\tilde{V}_{0}\right|\Big|\tilde{S}_{0}=\bar{S}_{t}e^{x},\tilde{V}_{0}=\bar{V}_{t}\right]\right],\label{first:ineq:to:use:1}
\end{align}
where we used the $1$-Lipschitzness of the map $x\mapsto x^{+}$
and $\frac{1}{2\sqrt{\kappa}}$-Lipschitzness of the map $x\mapsto\sqrt{x+\kappa}$
with $x\geq 0$.
By \eqref{first:ineq:to:use:1}, we have
\begin{align*}
&\Bigg|\mathbb{E}\left[\mathbb{E}\left[\left(\sqrt{\eta^{2}(\tilde{S}_{T-t})\tilde{V}_{T-t}+\kappa}-K\right)^{+}\Big|\tilde{S}_{0}=\bar{S}_{t}e^{x},\tilde{V}_{0}=\bar{V}_{t}\right]\right]
\\
&\qquad\qquad
-\mathbb{E}\left[\left(\sqrt{\eta^{2}(\bar{S}_{t}e^{x})\bar{V}_{t}+\kappa}-K\right)^{+}\right]\Bigg|
\\
&\leq\frac{1}{2\sqrt{\kappa}}\mathbb{E}\left[\mathbb{E}\left[\left|\eta^{2}(\tilde{S}_{T-t})\tilde{V}_{T-t}-\eta^{2}(\tilde{S}_{T-t})\tilde{V}_{0}\right|\Big|\tilde{S}_{0}=\bar{S}_{t}e^{x},\tilde{V}_{0}=\bar{V}_{t}\right]\right]
\\
&\qquad\qquad\qquad
+\frac{1}{2\sqrt{\kappa}}\mathbb{E}\left[\mathbb{E}\left[\left|\eta^{2}(\tilde{S}_{T-t})\tilde{V}_{0}-\eta^{2}(\tilde{S}_{0})\tilde{V}_{0}\right|\Big|\tilde{S}_{0}=\bar{S}_{t}e^{x},\tilde{V}_{0}=\bar{V}_{t}\right]\right].
\end{align*}

By Lemma~\ref{lem:1}, we have
\begin{align*}
&\frac{1}{2\sqrt{\kappa}}\mathbb{E}\left[\mathbb{E}\left[\left|\eta^{2}(\tilde{S}_{T-t})\tilde{V}_{T-t}-\eta^{2}(\tilde{S}_{T-t})\tilde{V}_{0}\right|\Big|\tilde{S}_{0}=\bar{S}_{t}e^{x},\tilde{V}_{0}=\bar{V}_{t}\right]\right]
\\
&\leq\frac{1}{2\sqrt{\kappa}}M_{\eta}^{2}\mathbb{E}\left[\mathbb{E}\left[\left|\tilde{V}_{T-t}-\tilde{V}_{0}\right|\Big|\tilde{S}_{0}=\bar{S}_{t}e^{x},\tilde{V}_{0}=\bar{V}_{t}\right]\right]
\\
&\leq\frac{1}{2\sqrt{\kappa}}M_{\eta}^{2}\left(e^{2(T-t)M_{\mu}}
e^{4(T-t)M_{\sigma}^{2}}+1-2e^{-(T-t)M_{\mu}-\frac{1}{2}(T-t) M_{\sigma}^{2}}\right)^{1/2}\mathbb{E}\left[\mathbb{E}\left[\tilde{V}_{0}\Big|\tilde{S}_{0}=\bar{S}_{t}e^{x},\tilde{V}_{0}=\bar{V}_{t}\right]\right]
\\
&=\frac{1}{2\sqrt{\kappa}}M_{\eta}^{2}\left(e^{2(T-t)M_{\mu}}
e^{4(T-t)M_{\sigma}^{2}}+1-2e^{-(T-t)M_{\mu}-\frac{1}{2}(T-t) M_{\sigma}^{2}}\right)^{1/2}\mathbb{E}\left[\bar{V}_{t}\right]
\\
&\leq\frac{1}{2\sqrt{\kappa}}M_{\eta}^{2}\left(e^{2TM_{\mu}}
e^{4TM_{\sigma}^{2}}+1-2e^{-TM_{\mu}-\frac{1}{2}T M_{\sigma}^{2}}\right)^{1/2}V_{0}e^{TM_{\mu}},
\end{align*}
for any $0\leq t\leq T$, where we applied Lemma~\ref{lem:2}.
Furthermore, 
\begin{align*}
&=\frac{1}{2\sqrt{\kappa}}\mathbb{E}\left[\mathbb{E}\left[\left|\eta^{2}(\tilde{S}_{T-t})\tilde{V}_{0}-\eta^{2}(\tilde{S}_{0})\tilde{V}_{0}\right|\Big|\tilde{S}_{0}=\bar{S}_{t}e^{x},\tilde{V}_{0}=\bar{V}_{t}\right]\right]
\\
&=\frac{1}{2\sqrt{\kappa}}\mathbb{E}\left[\mathbb{E}\left[\left|\eta^{2}(\tilde{S}_{T-t})-\eta^{2}(\tilde{S}_{0})\right|\Big|\tilde{S}_{0}=\bar{S}_{t}e^{x},\tilde{V}_{0}=\bar{V}_{t}\right]\right]\mathbb{E}[\bar{V}_{t}]
\\
&\leq\frac{1}{2\sqrt{\kappa}}V_{0}e^{TM_{\mu}}\mathbb{E}\left[\mathbb{E}\left[\left|\eta^{2}(\tilde{S}_{T-t})-\eta^{2}(\tilde{S}_{0})\right|\Big|\tilde{S}_{0}=\bar{S}_{t}e^{x},\tilde{V}_{0}=\bar{V}_{t}\right]\right],
\end{align*}
for any $0\leq t\leq T$, where we applied Lemma~\ref{lem:2}. Since $\eta(\cdot)$ is uniformly bounded, 
by bounded convergence theorem, 
\begin{equation}
\frac{1}{T}\int_{0}^{T}\int_{\mathbb{R}}\mathbb{E}\left[\mathbb{E}\left[\left|\eta^{2}(\tilde{S}_{T-t})-\eta^{2}(\tilde{S}_{0})\right|\Big|\tilde{S}_{0}=\bar{S}_{t}e^{x},\tilde{V}_{0}=\bar{V}_{t}\right]\right]P^{S}(dx)dt\rightarrow 0,
\end{equation}
as $T\rightarrow 0$. 
Hence, we conclude that
\begin{align*}
&\frac{1}{T}\int_{0}^{T}\int_{\mathbb{R}}\Bigg|\mathbb{E}\left[\mathbb{E}\left[\left(\sqrt{\eta^{2}(\tilde{S}_{T-t})\tilde{V}_{T-t}+\kappa}-K\right)^{+}\Big|\tilde{S}_{0}=\bar{S}_{t}e^{x},\tilde{V}_{0}=\bar{V}_{t}\right]\right]
\\
&\qquad\qquad\qquad\qquad\qquad
-\mathbb{E}\left[\left(\sqrt{\eta^{2}(\bar{S}_{t}e^{x})\bar{V}_{t}+\kappa}-K\right)^{+}\right]\Bigg|
P^{S}(dx)dt\rightarrow 0,
\end{align*}
as $T\rightarrow 0$.

Moreover, we have
\begin{align}
&\left|\mathbb{E}\left[\left(\sqrt{\eta^{2}(\bar{S}_{t}e^{x})\bar{V}_{t}+\kappa}-K\right)^{+}\right]
-\left(\sqrt{\eta^{2}(S_{0}e^{x})V_{0}+\kappa}-K\right)^{+}\right|
\nonumber
\\
&\leq
\mathbb{E}\left|\sqrt{\eta^{2}(\bar{S}_{t}e^{x})\bar{V}_{t}+\kappa}-\sqrt{\eta^{2}(S_{0}e^{x})V_{0}+\kappa}\right|
\nonumber
\\
&\leq
\frac{1}{2\sqrt{\kappa}}\mathbb{E}\left|\eta^{2}(\bar{S}_{t}e^{x})\bar{V}_{t}-\eta^{2}(S_{0}e^{x})V_{0}\right|,\label{first:ineq:to:use:2}
\end{align}
where we used the $1$-Lipschitzness of the map $x\mapsto x^{+}$
and $\frac{1}{2\sqrt{\kappa}}$-Lipschitzness of the map $x\mapsto\sqrt{x+\kappa}$
with $x\geq 0$.
By \eqref{first:ineq:to:use:2} and Lemma~\ref{lem:1}, we have
\begin{align}
&\left|\mathbb{E}\left[\left(\sqrt{\eta^{2}(\bar{S}_{t}e^{x})\bar{V}_{t}+\kappa}-K\right)^{+}\right]
-\left(\sqrt{\eta^{2}(S_{0}e^{x})V_{0}+\kappa}-K\right)^{+}\right|
\nonumber
\\
&\leq
\frac{1}{2\sqrt{\kappa}}\mathbb{E}\left|\eta^{2}(\bar{S}_{t}e^{x})\bar{V}_{t}-\eta^{2}(\bar{S}_{t}e^{x})V_{0}\right|
+\frac{V_{0}}{2\sqrt{\kappa}}\mathbb{E}\left|\eta^{2}(\bar{S}_{t}e^{x})-\eta^{2}(S_{0}e^{x})\right|
\nonumber
\\
&\leq
\frac{1}{2\sqrt{\kappa}}M_{\eta}^{2}\mathbb{E}\left|\bar{V}_{t}-V_{0}\right|
+\frac{V_{0}}{2\sqrt{\kappa}}\mathbb{E}\left|\eta^{2}(\bar{S}_{t}e^{x})-\eta^{2}(S_{0}e^{x})\right|
\nonumber
\\
&\leq
\frac{1}{2\sqrt{\kappa}}M_{\eta}^{2}V_{0}\left(e^{2TM_{\mu}}
e^{4TM_{\sigma}^{2}}+1-2e^{-TM_{\mu}-\frac{1}{2}T M_{\sigma}^{2}}\right)^{1/2}
+\frac{V_{0}}{2\sqrt{\kappa}}\mathbb{E}\left|\eta^{2}(\bar{S}_{t}e^{x})-\eta^{2}(S_{0}e^{x})\right|,
\end{align}
for any $0\leq t\leq T$.
Note that $\bar{S}_{t}\rightarrow S_{0}$ a.s. for any $0\leq t\leq T$
as $T\rightarrow 0$. Since $\eta(\cdot)$ is uniformly bounded,
by the bounded convergence theorem, we have
\begin{equation}
\frac{1}{T}\int_{0}^{T}\int_{\mathbb{R}}\mathbb{E}\left|\eta^{2}(\bar{S}_{t}e^{x})-\eta^{2}(S_{0}e^{x})\right|P^{S}(dx)dt
\rightarrow 0,
\end{equation}
as $T\rightarrow 0$.
Hence, we conclude that
\begin{equation}
\frac{1}{T}\int_{0}^{T}\int_{\mathbb{R}}\left|\mathbb{E}\left[\left(\sqrt{\eta^{2}(\bar{S}_{t}e^{x})\bar{V}_{t}+\kappa}-K\right)^{+}\right]
-\left(\sqrt{\eta^{2}(S_{0}e^{x})V_{0}+\kappa}-K\right)^{+}\right|P^{S}(dx)dt
\rightarrow 0,
\end{equation}
as $T\rightarrow 0$.

Finally, we can compute that
\begin{align*}
\mathbb{Q}\left(N_{T}^{S}=1,N_{T}^{C}=N_{T}^{V}=0\right)
=\lambda^{S}Te^{-\lambda^{S}T-\lambda^{C}T-\lambda^{V}T}
=\lambda^{S}T+O(T^{2}),
\end{align*}
as $T\rightarrow 0$.
Hence, we proved \eqref{to:show:2}.

\textbf{Step 3.} We can show that
\begin{align}
&\lim_{T\rightarrow 0}\frac{1}{T}e^{-rT}\mathbb{E}\left[\left(\sqrt{\eta^{2}(S_{T})V_{T}+\kappa}-K\right)^{+}1_{N_{T}^{S}=0}1_{N_{T}^{C}=1}1_{N_{T}^{V}=0}\right]
\nonumber
\\
&=\lambda^{C}\int_{\mathbb{R}}\int_{\mathbb{R}}\left(\sqrt{\eta^{2}\left(S_{0}e^{x}\right)V_{0}e^{y}+\kappa}-K\right)^{+}P^{C}(dx,dy).\label{to:show:3}
\end{align}
Next, let us prove \eqref{to:show:3}.
We can compute that
\begin{align*}
&\mathbb{E}\left[\left(\sqrt{\eta^{2}(S_{T})V_{T}+\kappa}-K\right)^{+}1_{N_{T}^{S}=0}1_{N_{T}^{C}=1}1_{N_{T}^{V}=0}\right]
\\
&=\mathbb{E}\left[\left(\sqrt{\eta^{2}(S_{T})V_{T}+\kappa}-K\right)^{+}\Big|N_{T}^{C}=1,N_{T}^{S}=N_{T}^{V}=0\right]
\mathbb{Q}\left(N_{T}^{C}=1,N_{T}^{S}=N_{T}^{V}=0\right).
\end{align*}
Furthermore, we can compute that
\begin{align*}
&\mathbb{E}\left[\left(\sqrt{\eta^{2}(S_{T})V_{T}+\kappa}-K\right)^{+}\Big|N_{T}^{C}=1,N_{T}^{S}=N_{T}^{V}=0\right]
\\
&=\frac{1}{T}\int_{0}^{T}\int_{\mathbb{R}}\int_{\mathbb{R}}\mathbb{E}\left[\mathbb{E}\left[\left(\sqrt{\eta^{2}(\tilde{S}_{T})\tilde{V}_{T}+\kappa}-K\right)^{+}\Big|\tilde{S}_{t}=\bar{S}_{t}e^{x},\tilde{V}_{t}=\bar{V}_{t}e^{y}\right]\right]P^{C}(dx,dy)dt,
\end{align*}
where $(\tilde{S}_{t},\tilde{V}_{t})_{t\geq 0}$ and $(\bar{S}_{t},\bar{V}_{t})_{t\geq 0}$ are two independent local-stochastic volatility models without jump given in 
\eqref{tilde:S:defn} and \eqref{bar:S:defn}.

Similar as in \textbf{Step 2}, we can show that 
\begin{align}
&\frac{1}{T}\int_{0}^{T}\int_{\mathbb{R}}\int_{\mathbb{R}}\mathbb{E}\left[\mathbb{E}\left[\left(\sqrt{\eta^{2}(\tilde{S}_{T-t})\tilde{V}_{T-t}+\kappa}-K\right)^{+}\Big|\tilde{S}_{0}=\bar{S}_{t}e^{x},\tilde{V}_{0}=\bar{V}_{t}e^{y}\right]\right]P^{C}(dx,dy)dt\nonumber
\\
&\rightarrow\int_{\mathbb{R}}\int_{\mathbb{R}}\left(\sqrt{\eta^{2}\left(S_{0}e^{x}\right)V_{0}e^{y}+\kappa}-K\right)^{+}P^{C}(dx,dy),
\end{align}
as $T\rightarrow 0$, and 
\begin{align*}
\mathbb{Q}\left(N_{T}^{C}=1,N_{T}^{S}=N_{T}^{V}=0\right)
=\lambda^{C}Te^{-\lambda^{S}T-\lambda^{C}T-\lambda^{V}T}
=\lambda^{C}T+O(T^{2}),
\end{align*}
as $T\rightarrow 0$, which yields \eqref{to:show:3}.

\textbf{Step 4.} We can show that
\begin{align}
&\lim_{T\rightarrow 0}\frac{1}{T}e^{-rT}\mathbb{E}\left[\left(\sqrt{\eta^{2}(S_{T})V_{T}+\kappa}-K\right)^{+}1_{N_{T}^{S}=0}1_{N_{T}^{C}=0}1_{N_{T}^{V}=1}\right]
\nonumber
\\
&=\lambda^{V}\int_{\mathbb{R}}\left(\sqrt{\eta^{2}\left(S_{0}\right)V_{0}e^{x}+\kappa}-K\right)^{+}P^{V}(dx).\label{to:show:4}
\end{align}
Next, let us prove \eqref{to:show:4}.
We can compute that
\begin{align*}
&\mathbb{E}\left[\left(\sqrt{\eta^{2}(S_{T})V_{T}+\kappa}-K\right)^{+}1_{N_{T}^{S}=0}1_{N_{T}^{C}=0}1_{N_{T}^{V}=1}\right]
\\
&=\mathbb{E}\left[\left(\sqrt{\eta^{2}(S_{T})V_{T}+\kappa}-K\right)^{+}\Big|N_{T}^{V}=1,N_{T}^{S}=N_{T}^{C}=0\right]
\mathbb{Q}\left(N_{T}^{V}=1,N_{T}^{S}=N_{T}^{C}=0\right).
\end{align*}
Furthermore, we can compute that
\begin{align*}
&\mathbb{E}\left[\left(\sqrt{\eta^{2}(S_{T})V_{T}+\kappa}-K\right)^{+}\Big|N_{T}^{V}=1,N_{T}^{S}=N_{T}^{C}=0\right]
\\
&=\frac{1}{T}\int_{0}^{T}\int_{\mathbb{R}}\mathbb{E}\left[\mathbb{E}\left[\left(\sqrt{\eta^{2}(\tilde{S}_{T})\tilde{V}_{T}+\kappa}-K\right)^{+}\Big|\tilde{S}_{t}=\bar{S}_{t},\tilde{V}_{t}=\bar{V}_{t}e^{x}\right]\right]P^{V}(dx)dt,
\end{align*}
where $(\tilde{S}_{t},\tilde{V}_{t})_{t\geq 0}$ and $(\bar{S}_{t},\bar{V}_{t})_{t\geq 0}$ are two independent local-stochastic volatility models without jump given in 
\eqref{tilde:S:defn} and \eqref{bar:S:defn}.

Similar as in \textbf{Step 2}, we can show that 
\begin{align}
&\frac{1}{T}\int_{0}^{T}\int_{\mathbb{R}}\mathbb{E}\left[\mathbb{E}\left[\left(\sqrt{\eta^{2}(\tilde{S}_{T-t})\tilde{V}_{T-t}+\kappa}-K\right)^{+}\Big|\tilde{S}_{0}=\bar{S}_{t},\tilde{V}_{0}=\bar{V}_{t}e^{x}\right]\right]P^{V}(dx)dt\nonumber
\\
&\rightarrow\int_{\mathbb{R}}\int_{\mathbb{R}}\left(\sqrt{\eta^{2}\left(S_{0}\right)V_{0}e^{x}+\kappa}-K\right)^{+}P^{V}(dx),
\end{align}
as $T\rightarrow 0$, and 
\begin{align*}
\mathbb{Q}\left(N_{T}^{V}=1,N_{T}^{S}=N_{T}^{C}=0\right)
=\lambda^{V}Te^{-\lambda^{S}T-\lambda^{C}T-\lambda^{V}T}
=\lambda^{V}T+O(T^{2}),
\end{align*}
as $T\rightarrow 0$, which yields \eqref{to:show:4}.

\textbf{Step 5.} We can show that
\begin{equation}
e^{-rT}\mathbb{E}\left[\left(\sqrt{\eta^{2}(S_{T})V_{T}+\kappa}-K\right)^{+}1_{N_{T}^{S}+N_{T}^{C}+N_{T}^{V}\geq 2}\right]=o(T),\label{to:show:5}
\end{equation}
as $T\rightarrow 0$. Next, let us prove \eqref{to:show:5}.
By H\"{o}lder's inequality, for any $p,q>1$ with $\frac{1}{p}+\frac{1}{q}=1$, we have
\begin{align*}
&\mathbb{E}\left[\left(\sqrt{\eta^{2}(S_{T})V_{T}+\kappa}-K\right)^{+}1_{N_{T}^{S}+N_{T}^{C}+N_{T}^{V}\geq 2}\right]
\\
&\leq\left(\mathbb{E}\left[\left|\sqrt{\eta^{2}(S_{T})V_{T}+\kappa}-K\right|^{p}\right]\right)^{1/p}\left(\mathbb{E}\left[1_{N_{T}^{S}+N_{T}^{C}+N_{T}^{V}\geq 2}\right]\right)^{1/q}
\\
&\leq\left(\mathbb{E}\left[\left(\sqrt{\eta^{2}(S_{T})V_{T}+\kappa}+K\right)^{p}\right]\right)^{1/p}\left(\mathbb{Q}\left(N_{T}^{S}+N_{T}^{C}+N_{T}^{V}\geq 2\right)\right)^{1/q}.
\end{align*}
By Lemma~\ref{lem:p:bound}, there exists some $p>2$ such that $\mathbb{E}\left[\left(\sqrt{\eta^{2}(S_{T})V_{T}+\kappa}+K\right)^{p}\right]=O(1)$ as $T\rightarrow 0$
and 
\begin{align*}
&\mathbb{Q}\left(N_{T}^{S}+N_{T}^{C}+N_{T}^{V}\geq 2\right)
\\
&=1-\mathbb{Q}\left(N_{T}^{S}+N_{T}^{C}+N_{T}^{V}=0\right)
-\mathbb{Q}\left(N_{T}^{S}+N_{T}^{C}+N_{T}^{V}=1\right)
\\
&=1-e^{-\lambda^{S}T-\lambda^{C}T-\lambda^{V}T}
-\lambda^{S}Te^{-\lambda^{S}T-\lambda^{C}T-\lambda^{V}T}
-\lambda^{C}Te^{-\lambda^{S}T-\lambda^{C}T-\lambda^{V}T}
-\lambda^{V}Te^{-\lambda^{S}T-\lambda^{C}T-\lambda^{V}T}
\\
&=O(T^{2}),
\end{align*}
as $T\rightarrow 0$.  
By taking $1<q<2$ (for some $p>2$), we proved \eqref{to:show:5}.

By combining \textbf{Step 1}-\textbf{Step 5}, we complete the proof.
\end{proof}


\begin{proof}[Proof of Theorem~\ref{thm:ATM}]
Let us prove the asymptotics for the call option. 
The derivation for the put option is similar and omitted here.
First, notice that we have the following decomposition:
\begin{align}
\hat{C}_{V}(K,T)
&=e^{-rT}\mathbb{E}\left[\left(\sqrt{\eta^{2}(S_{T})V_{T}+\kappa}-K\right)^{+}1_{N_{T}^{S}+N_{T}^{C}+N_{T}^{V}=0}\right]
\nonumber
\\
&\quad
+e^{-rT}\mathbb{E}\left[\left(\sqrt{\eta^{2}(S_{T})V_{T}+\kappa}-K\right)^{+}1_{N_{T}^{S}+N_{T}^{C}+N_{T}^{V}\geq 1}\right].
\end{align}

\textbf{Step 1.} We can show that
\begin{equation}\label{to:show:1:ATM}
e^{-rT}\mathbb{E}\left[\left(\sqrt{\eta^{2}(S_{T})V_{T}+\kappa}-K\right)^{+}1_{N_{T}^{S}+N_{T}^{C}+N_{T}^{V}\geq 1}\right]=o\left(\sqrt{T}\right),    
\end{equation}
as $T\rightarrow 0$. Next, let us prove \eqref{to:show:1:ATM}.
By H\"{o}lder's inequality, for any $p,q>1$ with $\frac{1}{p}+\frac{1}{q}=1$, we have
\begin{align*}
&\mathbb{E}\left[\left(\sqrt{\eta^{2}(S_{T})V_{T}+\kappa}-K\right)^{+}1_{N_{T}^{S}+N_{T}^{C}+N_{T}^{V}\geq 1}\right]
\\
&\leq\left(\mathbb{E}\left[\left|\sqrt{\eta^{2}(S_{T})V_{T}+\kappa}-K\right|^{p}\right]\right)^{1/p}\left(\mathbb{E}\left[1_{N_{T}^{S}+N_{T}^{C}+N_{T}^{V}\geq 1}\right]\right)^{1/q}
\\
&\leq\left(\mathbb{E}\left[\left(\sqrt{\eta^{2}(S_{T})V_{T}+\kappa}+K\right)^{p}\right]\right)^{1/p}\left(\mathbb{Q}\left(N_{T}^{S}+N_{T}^{C}+N_{T}^{V}\geq 1\right)\right)^{1/q}.
\end{align*}
By Lemma~\ref{lem:p:bound}, we have $\mathbb{E}\left[\left(\sqrt{\eta^{2}(S_{T})V_{T}+\kappa}+K\right)^{p}\right]=O(1)$ as $T\rightarrow 0$
and 
\begin{equation*}
\mathbb{Q}\left(N_{T}^{S}+N_{T}^{C}+N_{T}^{V}\geq 1\right)
=1-\mathbb{Q}\left(N_{T}^{S}+N_{T}^{C}+N_{T}^{V}=0\right)
=1-e^{-\lambda^{S}T-\lambda^{C}T-\lambda^{V}T}=O(T),
\end{equation*}
as $T\rightarrow 0$. By taking $1<q<2$, we proved \eqref{to:show:1:ATM}.

\textbf{Step 2.} Conditional on $N_{T}^{S}+N_{T}^{C}+N_{T}^{V}=0$, 
$(S_{t},V_{t})$ is the local-stochastic volatility model without jumps,
and we can show that
\begin{align}
&\Bigg|e^{-rT}\mathbb{E}\left[\left(\sqrt{\eta^{2}(S_{T})V_{T}+\kappa}-K\right)^{+}1_{N_{T}^{S}+N_{T}^{C}+N_{T}^{V}=0}\right]
\nonumber
\\
&\qquad
-\mathbb{E}\left[\left(\sqrt{\eta^{2}(\hat{S}_{T})\hat{V}_{T}+\kappa}-K\right)^{+}\right]\Bigg|=o\left(\sqrt{T}\right),\label{to:show:2:ATM}
\end{align}
as $T\rightarrow 0$, where
\begin{align}
&\hat{S}_{T}=S_{0}+\eta(S_{0})S_{0}\sqrt{V_{0}}\left(\sqrt{1-\rho^{2}}B_{T}+\rho Z_{T}\right)\,,
\\
&\hat{V}_{T}=V_{0}+\sigma(V_{0})V_{0}Z_{T}\,,
\end{align}
where $B_{t},Z_{t}$ are independent standard Brownian motions.
Next, let us prove \eqref{to:show:2:ATM}.
Notice that
\begin{align*}
&\mathbb{E}\left[\left(\sqrt{\eta^{2}(S_{T})V_{T}+\kappa}-K\right)^{+}1_{N_{T}^{S}+N_{T}^{C}+N_{T}^{V}=0}\right]
\\
&=\mathbb{E}\left[\left(\sqrt{\eta^{2}(S_{T})V_{T}+\kappa}-K\right)^{+}\Big|N_{T}^{S}+N_{T}^{C}+N_{T}^{V}=0\right]
\mathbb{Q}\left(N_{T}^{S}+N_{T}^{C}+N_{T}^{V}=0\right)
\\
&=\mathbb{E}\left[\left(\sqrt{\eta^{2}(\tilde{S}_{T})\tilde{V}_{T}+\kappa}-K\right)^{+}\right]
\mathbb{Q}\left(N_{T}^{S}+N_{T}^{C}+N_{T}^{V}=0\right),
\end{align*}
where $(\tilde{S}_{t},\tilde{V}_{t})_{t\geq 0}$ is the local-stochastic volatility model without jumps:
\begin{eqnarray}
&& \frac{d\tilde{S}_t}{\tilde{S}_t} = \eta(\tilde{S}_t) \sqrt{\tilde{V}_t} dW_t + \left(r - q-\lambda^{S}\mu^{S}-\lambda^{C}\mu^{C,S}\right) dt \,, \\
&& \frac{d\tilde{V}_t}{\tilde{V}_t} =\sigma(\tilde{V}_t) dZ_t + \mu(\tilde{V}_t) dt \,.\nonumber
\end{eqnarray}
Moreover, 
\begin{align*}
&\left|\mathbb{E}\left[\left(\sqrt{\eta^{2}(\tilde{S}_{T})\tilde{V}_{T}+\kappa}-K\right)^{+}\right]
\mathbb{Q}\left(N_{T}^{S}+N_{T}^{C}+N_{T}^{V}=0\right)
-\mathbb{E}\left[\left(\sqrt{\eta^{2}(\tilde{S}_{T})\tilde{V}_{T}+\kappa}-K\right)^{+}\right]\right|
\\
&=\mathbb{E}\left[\left(\sqrt{\eta^{2}(\tilde{S}_{T})\tilde{V}_{T}+\kappa}-K\right)^{+}\right]\mathbb{Q}\left(N_{T}^{S}+N_{T}^{C}+N_{T}^{V}\geq 1\right),
\end{align*}
where one can show that $\mathbb{E}\left[\left(\sqrt{\eta^{2}(\tilde{S}_{T})\tilde{V}_{T}+\kappa}-K\right)^{+}\right]=O(1)$
as $T\rightarrow 0$ and
\begin{equation}
\mathbb{Q}\left(N_{T}^{S}+N_{T}^{C}+N_{T}^{V}\geq 1\right)
=1-e^{-\lambda^{S}T-\lambda^{C}T-\lambda^{V}T}=O(T),
\end{equation}
as $T\rightarrow 0$. Moreover, by Theorem~4.2 in \cite{VIXpaper}, 
\begin{align}
\Bigg|e^{-rT}\mathbb{E}\left[\left(\sqrt{\eta^{2}(\tilde{S}_{T})\tilde{V}_{T}+\kappa}-K\right)^{+}\right]
-\mathbb{E}\left[\left(\sqrt{\eta^{2}(\hat{S}_{T})\hat{V}_{T}+\kappa}-K\right)^{+}\right]\Bigg|=o\left(\sqrt{T}\right),
\end{align}
which completes the proof of \eqref{to:show:2:ATM}.

\textbf{Step 3.} Next, we can compute that
\begin{align}
&\mathbb{E}\left[\left(\sqrt{\eta^{2}(\hat{S}_{T})\hat{V}_{T}+\kappa}-K\right)^{+}\right]\nonumber
\\
&=\mathbb{E}\left[\left(\sqrt{\eta^{2}(S_{0})V_{0}+\kappa}
+\frac{\eta^{2}(\hat{S}_{T})\hat{V}_{T}-\eta^{2}(S_{0})V_{0}}{2\sqrt{\eta^{2}(S_{0})V_{0}+\kappa}}-K\right)^{+}\right]
+o\left(\sqrt{T}\right)
\nonumber
\\
&=\frac{\sqrt{\eta^{2}(S_{0})V_{0}}}{\sqrt{\eta^{2}(S_{0})V_{0}+\kappa}}\mathbb{E}\left[\left(\frac{\eta^{2}(\hat{S}_{T})\hat{V}_{T}-\eta^{2}(S_{0})V_{0}}{2\sqrt{\eta^{2}(S_{0})V_{0}}}\right)^{+}\right]
+o\left(\sqrt{T}\right)
\end{align}
as $T\rightarrow 0$.
It is proved in the proof of Theorem~4.2 in \cite{VIXpaper} that
\begin{align}
&\Bigg|\mathbb{E}\left[\left(\frac{\eta^{2}(\hat{S}_{T})\hat{V}_{T}-\eta^{2}(S_{0})V_{0}}{2\sqrt{\eta^{2}(S_{0})V_{0}}}\right)^{+}\right]
\nonumber
\\
&-\mathbb{E}\left[\left(\eta(S_{0})\frac{1}{2\sqrt{V_{0}}}\sigma(V_{0})V_{0}Z_{T}
+\sqrt{V_{0}}\eta'(S_{0})\eta(S_{0})S_{0}\sqrt{V_{0}}\left(\sqrt{1-\rho^{2}}B_{T}+\rho Z_{T}\right)\right)^{+}\right]\Bigg|
\nonumber
\\
&=o\left(\sqrt{T}\right),
\end{align}
as $T\rightarrow 0$.

\textbf{Step 4.} Finally, we recall from the proof of Theorem~4.2 in \cite{VIXpaper} that
\begin{align}
&\mathbb{E}\left[\left(\eta(S_{0})\frac{1}{2\sqrt{V_{0}}}\sigma(V_{0})V_{0}Z_{T}
+\sqrt{V_{0}}\eta'(S_{0})\eta(S_{0})S_{0}\sqrt{V_{0}}\left(\sqrt{1-\rho^{2}}B_{T}+\rho Z_{T}\right)\right)^{+}\right]
\nonumber
\\
&=\sqrt{\left((\eta(S_{0})\frac{1}{2}\sigma(V_{0})\sqrt{V_{0}}+\eta'(S_{0})\eta(S_{0})S_{0}V_{0}\rho\right)^{2}+\left(\eta'(S_{0})\eta(S_{0})S_{0}V_{0}\sqrt{1-\rho^{2}}\right)^{2}}
\sqrt{T}\frac{1}{\sqrt{2\pi}}.
\end{align}
By combining \textbf{Step 1}-\textbf{Step 4}, we complete the proof.
\end{proof}


\begin{proof}[Proof of Theorem~\ref{thm:OTM:European}]
Let us prove the asymptotics for the call option.
The derivation for the put option is similar and omitted here.
The proof is very similar to the proof of Theorem~\ref{thm:OTM}.
As a result, we only provide an outline of the proof and omit
the details here.

First, notice that we have the following decomposition:
\begin{align}
C_{E}(K,T)
&=e^{-rT}\mathbb{E}\left[\left(S_{T}-K\right)^{+}1_{N_{T}^{S}+N_{T}^{C}+N_{T}^{V}=0}\right]
\nonumber
\\
&\qquad
+e^{-rT}\mathbb{E}\left[\left(S_{T}-K\right)^{+}1_{N_{T}^{S}+N_{T}^{C}+N_{T}^{V}=1}\right]
\nonumber
\\
&\qquad\qquad
+e^{-rT}\mathbb{E}\left[\left(S_{T}-K\right)^{+}1_{N_{T}^{S}+N_{T}^{C}+N_{T}^{V}\geq 2}\right].
\end{align}

\textbf{Step 1.} We can show that
\begin{equation}
e^{-rT}\mathbb{E}\left[\left(S_{T}-K\right)^{+}1_{N_{T}^{S}+N_{T}^{C}+N_{T}^{V}=0}\right]=o(T),
\end{equation}
as $T\rightarrow 0$.

\textbf{Step 2.} We can show that
\begin{align}
&\lim_{T\rightarrow 0}\frac{1}{T}e^{-rT}\mathbb{E}\left[\left(S_{T}-K\right)^{+}1_{N_{T}^{S}+N_{T}^{C}+N_{T}^{V}=1}\right]
\nonumber
\\
&=\lambda^{S}\int_{\mathbb{R}}\left(S_{0}e^{x}-K\right)^{+}P^{S}(dx)
+\lambda^{C}\int_{\mathbb{R}}\int_{\mathbb{R}}\left(S_{0}e^{x}-K\right)^{+}P^{C}(dx,dy).
\end{align}

\textbf{Step 3.} We can show that
\begin{equation}
e^{-rT}\mathbb{E}\left[\left(S_{T}-K\right)^{+}1_{N_{T}^{S}+N_{T}^{C}+N_{T}^{V}\geq 2}\right]=o(T),
\end{equation}
as $T\rightarrow 0$.

By combining \textbf{Step 1}-\textbf{Step 3}, we complete the proof.
\end{proof}


\begin{proof}[Proof of Theorem~\ref{thm:ATM:European}]
Let us prove the asymptotics for the call option. 
The derivation for the put option is similar and omitted here.
The proof is very similar to the proof of Theorem~\ref{thm:ATM}.
As a result, we only provide an outline of the proof and omit
the details here.
First, notice that we have the following decomposition:
\begin{align}
C_{E}(K,T)
=e^{-rT}\mathbb{E}\left[\left(S_{T}-K\right)^{+}1_{N_{T}^{S}+N_{T}^{C}+N_{T}^{V}=0}\right]
+e^{-rT}\mathbb{E}\left[\left(S_{T}-K\right)^{+}1_{N_{T}^{S}+N_{T}^{C}+N_{T}^{V}\geq 1}\right].
\end{align}

\textbf{Step 1.} We can show that
\begin{equation}
e^{-rT}\mathbb{E}\left[\left(S_{T}-K\right)^{+}1_{N_{T}^{S}+N_{T}^{C}+N_{T}^{V}\geq 1}\right]=o\left(\sqrt{T}\right),    
\end{equation}
as $T\rightarrow 0$. 

\textbf{Step 2.} Conditional on $N_{T}^{S}+N_{T}^{C}+N_{T}^{V}=0$, 
$(S_{t},V_{t})$ is the local-stochastic volatility model without jumps,
and we can show that
\begin{align}
\Bigg|e^{-rT}\mathbb{E}\left[\left(S_{T}-K\right)^{+}1_{N_{T}^{S}+N_{T}^{C}+N_{T}^{V}=0}\right]
-\mathbb{E}\left[\left(\hat{S}_{T}-K\right)^{+}\right]\Bigg|=o\left(\sqrt{T}\right),
\end{align}
as $T\rightarrow 0$, where
\begin{align}
&\hat{S}_{T}=S_{0}+\eta(S_{0})S_{0}\sqrt{V_{0}}\left(\sqrt{1-\rho^{2}}B_{T}+\rho Z_{T}\right)\,,
\\
&\hat{V}_{T}=V_{0}+\sigma(V_{0})V_{0}Z_{T}\,,
\end{align}
where $B_{t},Z_{t}$ are independent standard Brownian motions.

\textbf{Step 3.} Finally, by the proof of Theorem~3.2 in \cite{VIXpaper}, we can compute that
\begin{align}
\mathbb{E}\left[\left(\hat{S}_{T}-K\right)^{+}\right]=\frac{\eta(S_{0})\sqrt{V_{0}}}{\sqrt{2\pi}}\sqrt{T}.
\end{align}
By combining \textbf{Step 1}-\textbf{Step 3}, we complete the proof.
\end{proof}


\section{Proofs of the Results in Section~\ref{sec:models}}

\begin{proof}[Proof of Corollary~\ref{cor:Eraker:European}]

a) By applying Theorem~\ref{thm:OTM:European}, we can compute that
\begin{align}
a_{E,C}(K) &:= \lim_{T\to 0}
\frac{C_{E}(K,T)}{T}\nonumber 
\\
&= 
\lambda^{C} \int_0^\infty \left(S_0 e^x - K\right)^+ p^{C}(x,y) dx dy
+ \lambda_S \int_0^\infty \left( S_0 e^x - K\right)^+ P^S(dx) \,.\label{two:integrals}
\end{align}

The first integral in \eqref{two:integrals} is evaluated as
\begin{align}
\int_0^\infty \left(S_0 e^x - K\right)^+ p^{C}(x,y) dx dy = \eta_{C,V} \int_0^\infty e^{-\eta_{C,V} y}
c_{\mathrm{BS}}\left(K,S_0e^{\mu^{S}+\rho_J y}, \sigma_{C,S}\right) dy\,, 
\end{align}
where $c_{\mathrm{BS}}(K,F,v) = 
F N\left(-\frac{1}{v} \log(K/F)+\frac12 v\right) - 
K N\left(-\frac{1}{v} \log(K/F) - \frac12 v\right)$. 

The second integral in \eqref{two:integrals} is the contribution from the idiosyncratic jumps in $S$ and is evaluated in closed form as
\begin{align}
\int_0^\infty \left( S_0 e^x - K\right)^+ P_S(dx) =
S_0 e^{\alpha_S + \frac12 \sigma_S^2} \Phi((-k + \alpha_S + \sigma_S^2)/\sigma_S) - K \Phi((-k+\alpha_S)/\sigma_S)\,.
\end{align}

b) The derivation for put options is similar
and hence omitted here. 
This completes the proof.
\end{proof}

\begin{proof}[Proof of Corollary~\ref{cor:Eraker:VIX}]

a) By applying Theorem~\ref{thm:OTM}, 
we can compute that
\begin{align}\label{two:integrals:2}
a_{V,C}(K) := \lim_{T\to 0}
\frac{C_V(K,T)}{T} &= \lambda^{C} \int_0^\infty \left(\sqrt{V_0 e^y + \kappa} - K\right)^+ p^{C}(x,y) dx dy \\
&\qquad\qquad+ \lambda^{V} \int_0^\infty \left(\sqrt{V_0 e^x + \kappa} - K\right)^+ p^{V}(y) dy\,.\nonumber
\end{align}

The two integrals in \eqref{two:integrals:2} have the same form. The first one in \eqref{two:integrals:2} is evaluated as
\begin{align}
\int_0^\infty \left(\sqrt{V_0 e^y + \kappa} - K\right)^+ p^{C}(x,y) dx dy
&=  \eta_{C,V} \int_0^\infty \left(\sqrt{V_0 e^y + \kappa} - K\right)^+ e^{-\eta_{C,V} y} dy\nonumber\\
&=   e^{-\eta_{C,V} y_0} \left[\eta_{C,V} I_1\left(\kappa, V_0 e^{y_0}, \eta_{C,V}\right) - K 
\right] \,,
\end{align}
where $y_0(K) = \log\frac{K^2-\kappa}{V_0}$. The integral is evaluated in terms of the function
$I_1(a,b,\eta)=\frac{2\sqrt{b}}{2\eta-1} {}_2F_1(-1/2, -1/2+\eta, 1/2+\eta;-a/b)$, where ${}_2 F_1(a,b,c;z)$ is the Gauss hypergeometric function. 

The second integral in \eqref{two:integrals:2} is evaluated in a similar way with the result
\begin{align}
\int_0^\infty \left(\sqrt{V_0 e^x + \kappa} - K\right)^+ p^{V}(y) dy =
e^{-\eta_{V} y_0} \left[\eta_{V} I_1\left(\kappa, V_0 e^{y_0}, \eta_{V}\right) - K 
\right] \,.
\end{align}


The contribution of the idiosyncratic jumps in $S$ vanishes since the integrand vanishes for OTM VIX call options
\begin{align}
S \mbox{ jumps}: \int_0^\infty \Big( \sqrt{V_0 + \kappa} - K \Big)^+ P^s(dx) =
\Big( \sqrt{V_0 + \kappa} - K \Big)^+ = 0\,.
\end{align}

b) The short maturity asymptotics for OTM VIX put options in this model vanishes
\begin{align}
a_{V,P}(K) := \lim_{T\to 0}
\frac{P_V(K,T)}{T} &= \lambda^{C} \int_0^\infty \left(K - \sqrt{V_0 e^y + \kappa}\right)^+ p^{C}(x,y) dx dy \\
&\qquad\qquad+ \lambda^V \int_0^\infty \left(K - \sqrt{V_0 e^y + \kappa}\right)^+ p^{V}(y) dy \nonumber \\
&= \lambda^{C} \eta_{C,V} \int_0^\infty \left(K - \sqrt{V_0 e^y + \kappa}\right)^+ e^{-\eta_{C,V} y} dy 
\nonumber \\
&\qquad\qquad+ \lambda^{V} \eta_{V} \int_0^\infty \left(K - \sqrt{V_0 e^y + \kappa}\right)^+ e^{-\eta_{V} y} dy  =0 \,.
\nonumber
\end{align}
This completes the proof.
\end{proof}


\begin{proof}[Proof of Corollary~\ref{cor:Kou:European}]

a) For call options with $k>0$, the application of Theorem~\ref{thm:OTM:European}
leads to the evaluation of the integral
\begin{align}\label{aECp}
a_{E,C}(K) &=  \lambda^{C}  \int_{\mathbb{R}} \int_{\mathbb{R}} (S_0 e^x - K)^+ p^{C}(x,y)  dx dy 
+ \lambda^S \int_{\mathbb{R}} (S_0 e^x - K)^+ p^S(x) dx \\
&= 
\lambda^{C} \int_0^\infty (S_0 e^x - K)^+ p(x) dx
+ \lambda^S \int_{\mathbb{R}} (S_0 e^x - K)^+ p^S(x) dx \,,\nonumber
\end{align}
where we denoted $p(x) := \int_0^\infty p^{C}(x,y) dy$ the marginal distribution of the common jump size for the underlying asset.

We start by evaluating the marginal density $p(x)$. The ranges of integration over $y$ are constrained by the indicator functions in $p^{C}(x,y)$ and depend on the sign of the correlation $\rho_J$. We give next the result for $\rho_J<0$ which is relevant for our application.

i) For $x>\mu^{S}$ we have
\begin{align}\label{p:x:1}
p(x) = c_R \eta_{C,S} e^{-\eta_{C,S} x}\,,\quad
c_R = \alpha \frac{\eta_{C,V}}{\eta_{C,V} + \eta_{C,S} |\rho_J |} e^{\eta_{C,S} \mu^{S}} \,.
\end{align}

ii) For $x<\mu^{S}$ we have
\begin{align}\label{p:x:2}
p(x) = c_{1L} \frac{\eta_{C,V}}{|\rho_J |} e^{ \frac{\eta_{C,V}}{|\rho_J |}  x} +
c_{2L} \eta_{C,S} e^{\eta_{C,S} x}\,,
\end{align}
where 
\begin{align}
c_{1L} &= \alpha \frac{\eta_{C,S} |\rho_J|}{\eta_{C,V} + \eta_{C,S} |\rho_J|} e^{-\frac{\eta_{C,V}}{|\rho_J|} \mu^{S}} - (1-\alpha)
\frac{\eta_{C,S} |\rho_J|}{\eta_{C,V} - \eta_{C,S} |\rho_J|} e^{-\frac{\eta_{C,V}}{|\rho_J|} \mu^{S}}\,,
\\
c_{2L} &=(1-\alpha)
\frac{\eta_{C,V} }{\eta_{C,V} - \eta_{C,S} |\rho_J|} e^{- \eta_{C,S}  \mu^{S}} \,.
\end{align}

For both cases the marginal distribution $p(x)$ is a linear combination of exponentials. The resulting $x$ integral are evaluated as
\begin{align}\label{f:c:eqn0}
f_c(k,a) := \int_0^\infty (S_0 e^x - K)^+ a e^{-a x} dx = S_0 \frac{1}{a-1} e^{-(a-1)k} \,.
\end{align}
Since $\mu^{S}<0$ the terms proportional to $c_{1L},c_{2L}$ do not contribute to the integral in \eqref{aECp}.
Substituting into \eqref{aECp} gives the stated result.

b) For put options $k<0$ we distinguish two cases. We consider only the case $\mu^{S}<0$, which corresponds to predominantly negative jumps. This requires that we consider separately the two ranges of strikes $k\in (-\infty, \mu^{S})$ and $k\in (\mu^{S},0)$.

The analog of \eqref{aECp} for put options is
\begin{align}\label{aEPp}
a_{E,P}(K) &=  \lambda^{C}  \int_{\mathbb{R}} \int_{\mathbb{R}} (K - S_0 e^x )^+ p^{C}(x,y)  dx dy 
+ \lambda^S \int_{-\infty}^0 (K - S_0 e^x)^+ p^S(x) dx \\
&= \lambda^{C} \int_{-\infty}^0 (K - S_0 e^x)^+ p(x) dx
+ \lambda^S \int_{-\infty}^0 (K - S_0 e^x)^+ p^S(x) dx\,. \nonumber
\end{align}

We substitute here the marginal distribution $p(x)$ computed above in \eqref{p:x:1} and \eqref{p:x:2}.
For $k < \mu^{S}$ the integral receives only contributions from the $c_{1L},c_{2L}$ terms in $p(x)$, and we get
\begin{align}
a_{E,P}(K) = \lambda^{C} \left(  c_{1L} f_p\left(k, \frac{\eta_{C,V}}{|\rho_J|}\right) + c_{2L} f_p(k, \eta_{C,S}) 
\right)\,,
\end{align}
with
\begin{align}
f_p(k,a) := \int_{-\infty}^0 (K - S_0 e^x )^+ a e^{a x} dx = S_0 \frac{1}{a+1} e^{(a+1)k} \,.
\end{align}

For put options with $\mu^{S} < k < 0$, there is also a contribution from the $c_R$ term in the density $p(x)$ and we get
\begin{align}
a_{E,P}(K) &= \lambda^{C} \left( c_{1L} f_p\left(\mu^{S}, \eta_{C,V}/|\rho_J|\right) + c_{2L}  f_p(k, \eta_{C,S}) + 
c_{R} \left[ f_c\left(\mu^{S},\eta_{C,S}\right) - f_c(k,\eta_{C,S})\right] \right) \nonumber \\
&\qquad\qquad+ \lambda^S \int_{-\infty}^0 (K - S_0 e^x)^+ p^S(x) dx \,,
\end{align}
where $f_{c}(\mu,\eta)$ is defined above in \eqref{f:c:eqn0}. 
This completes the proof.
\end{proof}


\begin{proof}[Proof of Corollary~\ref{cor:Kou:VIX}]

The proof is identical to that of Corollary~\ref{cor:Eraker:VIX} as all the integrals are identical.

\end{proof}

\begin{proof}[Proof of Corollary~\ref{cor:FN:Eur}]
The proof is analogous to that of Corollary~\ref{cor:Eraker:European}.
\end{proof}

\begin{proof}[Proof of Corollary~\ref{cor:FN:VIX}]
The proof is analogous to that of Corollary~\ref{cor:Eraker:VIX}.
\end{proof}

\section{Proofs of the Technical Lemmas}

\begin{proof}[Proof of Lemma~\ref{lem:p:bound}]
First, we have
\begin{equation}
\mathbb{E}\left[\left(\sqrt{\eta^{2}(S_{T})V_{T}+\kappa}+K\right)^{p}\right]
\leq
\mathbb{E}\left[\left(\eta(S_{T})\sqrt{V_{T}}+\sqrt{\kappa}+K\right)^{p}\right],
\end{equation}
where we used the inequality that $\sqrt{x+y}\leq\sqrt{x}+\sqrt{y}$
for any $x,y\geq 0$. 

Next, since $p>1$, by Jensen's inequality, we have
\begin{align}
\mathbb{E}\left[\left(\eta(S_{T})\sqrt{V_{T}}+\sqrt{\kappa}+K\right)^{p}\right]
&\leq 2^{p-1}\mathbb{E}\left[\left(\eta(S_{T})\sqrt{V_{T}}\right)^{p}+\left(\sqrt{\kappa}+K\right)^{p}\right]
\nonumber
\\
&\leq 2^{p-1}M_{\eta}^{p}\mathbb{E}\left[V_{T}^{p/2}\right]
+2^{p-1}\left(\sqrt{\kappa}+K\right)^{p}.
\end{align}
Next, we can compute that
\begin{equation}
V_{T}=V_{0}e^{\int_{0}^{T}(\mu(V_{u})-\frac{1}{2}\sigma^{2}(V_{u}))du+\int_{0}^{T}\sigma(V_{u})dZ_{u}+R_{T}^{V}+R_{T}^{C,V}},
\end{equation}
where, for any $t\geq 0$,
\begin{equation}
R_{t}^{V}:=\sum_{i=1}^{N_{t}^{S}}Y_{i}^{V},
\qquad
R_{t}^{C,V}:=\sum_{i=1}^{N_{t}^{C}}Y_{i}^{C,V}.
\end{equation}
and therefore
\begin{align}
\mathbb{E}\left[V_{T}^{p/2}\right]
&=V_{0}^{p/2}\mathbb{E}\left[e^{\int_{0}^{T}(\frac{p}{2}\mu(V_{u})-\frac{p}{4}\sigma^{2}(V_{u}))du+\frac{p}{2}\int_{0}^{T}\sigma(V_{u})dZ_{u}+\frac{p}{2}R_{T}^{V}+\frac{p}{2}R_{T}^{C,V}}\right]
\nonumber
\\
&=V_{0}^{p/2}\mathbb{E}\left[\mathbb{E}\left[e^{\int_{0}^{T}(\frac{p}{2}\mu(V_{u})-\frac{p}{4}\sigma^{2}(V_{u}))du+\frac{p}{2}\int_{0}^{T}\sigma(V_{u})dZ_{u}}|\mathcal{F}_{T}^{R,V}\right]e^{\frac{p}{2}R_{T}^{V}+\frac{p}{2}R_{T}^{C,V}}\right],
\end{align}
where $\mathcal{F}_{t}^{R,V}$ is the natural filtration of $(R_{t}^{V},R_{t}^{C,V})$ process.
We can compute that
\begin{align}
&\mathbb{E}\left[e^{\int_{0}^{T}(\frac{p}{2}\mu(V_{u})-\frac{p}{4}\sigma^{2}(V_{u}))du+\frac{p}{2}\int_{0}^{T}\sigma(V_{u})dZ_{u}}|\mathcal{F}_{T}^{R,V}\right]
\nonumber
\\
&\leq
e^{\frac{p}{2}M_{\mu}T+\frac{p^{2}}{8}M_{\sigma}^{2}T}
\mathbb{E}\left[e^{-\int_{0}^{T}\frac{p^{2}}{8}\sigma^{2}(V_{u})du+\frac{p}{2}\int_{0}^{T}\sigma(V_{u})dZ_{u}}|\mathcal{F}_{T}^{R,V}\right]
\nonumber
\\
&\leq e^{\frac{p}{2}M_{\mu}T+\frac{p^{2}}{8}M_{\sigma}^{2}T},
\end{align}
where we used Assumption~\ref{assump:bounded}
and the fact that $\left\{e^{-\int_{0}^{t}\frac{p^{2}}{2}\sigma^{2}(V_{u})du+p\int_{0}^{t}\sigma(V_{u})dZ_{u}}\right\}_{t\geq 0}$
is a non-negative local martingale and thus a supermartingale.
Thus,
\begin{align}
\mathbb{E}\left[V_{T}^{p/2}\right]
&\leq
V_{0}^{p/2}
e^{\frac{p}{2}M_{\mu}T+\frac{p^{2}}{8}M_{\sigma}^{2}T}
\mathbb{E}\left[e^{\frac{p}{2}R_{T}^{V}+\frac{p}{2}R_{T}^{C,V}}\right]
\nonumber
\\
&=V_{0}^{p/2}e^{\lambda^{V}T\left(\mathbb{E}\left[e^{\frac{p}{2}Y_{1}^{V}}\right]-1\right)}e^{\lambda^{C}T\left(\mathbb{E}\left[e^{\frac{p}{2}Y_{1}^{C,V}}\right]-1\right)}e^{\frac{p}{2}M_{\mu}T+\frac{p^{2}}{8}M_{\sigma}^{2}T},
\end{align}
provided that $\mathbb{E}\left[e^{\frac{p}{2}Y_{1}^{V}}\right]<\infty$ and $\mathbb{E}\left[e^{\frac{p}{2}Y_{1}^{C,V}}\right]<\infty$.
Hence, we conclude that for any $p>1$,
\begin{align}
&\mathbb{E}\left[\left(\sqrt{\eta^{2}(S_{T})V_{T}+\kappa}+K\right)^{p}\right]
\nonumber
\\
&\leq 2^{p-1}M_{\eta}^{p}V_{0}^{p/2}e^{\lambda^{V}T\left(\mathbb{E}\left[e^{\frac{p}{2}Y_{1}^{V}}\right]-1\right)}e^{\lambda^{C}T\left(\mathbb{E}\left[e^{\frac{p}{2}Y_{1}^{C,V}}\right]-1\right)}e^{\frac{p}{2}M_{\mu}T+\frac{p^{2}}{8}M_{\sigma}^{2}T}
+2^{p-1}\left(\sqrt{\kappa}+K\right)^{p}.
\end{align}
This completes the proof.
\end{proof}

\begin{proof}[Proof of Lemma~\ref{lem:1}]
We can compute that
\begin{align}
\mathbb{E}[| \bar{V}_t - \bar{V}_0|]
\leq
\left[\mathbb{E}|\bar{V}_{t}-V_{0}|\right]
\leq
\left(\left[\mathbb{E}|\bar{V}_{t}-V_{0}|^{2}\right]\right)^{1/2}.
\end{align}
We can further compute that
\begin{equation}
\bar{V}_{t}=V_{0}e^{\int_{0}^{t}(\mu(\bar{V}_{u})-\frac{1}{2}\sigma^{2}(\bar{V}_{u}))du+\int_{0}^{t}\sigma(\bar{V}_{u})d\bar{Z}_{u}},
\end{equation}
and by Jensen's inequality and Assumption~\ref{assump:bounded},
\begin{align}
\mathbb{E}[\bar{V}_{t}]
&=V_{0}\mathbb{E}\left[e^{\int_{0}^{t}(\mu(\bar{V}_{u})-\frac{1}{2}\sigma^{2}(\bar{V}_{u}))du
+\int_{0}^{t}\sigma(\bar{V}_{u})d\bar{Z}_{u}}\right]
\nonumber
\\
&\geq
V_{0}
e^{\mathbb{E}[\int_{0}^{t}(\mu(\bar{V}_{u})-\frac{1}{2}\sigma^{2}(\bar{V}_{u}))du
+\int_{0}^{t}\sigma(\bar{V}_{u})d\bar{Z}_{u}]}
=
V_{0}
e^{\mathbb{E}[\int_{0}^{t}(\mu(\bar{V}_{u})-\frac{1}{2}\sigma^{2}(\bar{V}_{u}))du]}
\geq
V_{0}e^{-tM_{\mu}-\frac{1}{2}tM_{\sigma}^{2}}.
\end{align}
On the other hand, by Assumption~\ref{assump:bounded}, 
\begin{align}
\mathbb{E}[\bar{V}_{t}^{2}]
=V_{0}^{2}\mathbb{E}\left[e^{\int_{0}^{t}(2\mu(\bar{V}_{u})-\sigma^{2}(\bar{V}_{u}))du
+2\int_{0}^{t}\sigma(\bar{V}_{u})d\bar{Z}_{u}}\right]
\leq
V_{0}^{2}e^{2tM_{\mu}}
\mathbb{E}\left[e^{2\int_{0}^{t}\sigma(\bar{V}_{u})d\bar{Z}_{u}}\right],
\end{align}
and by Cauchy-Schwarz inequality, we can further compute that
\begin{align}
\mathbb{E}\left[e^{2\int_{0}^{t}\sigma(\bar{V}_{u})d\bar{Z}_{u}}\right]
&=\mathbb{E}\left[e^{\int_{0}^{t}2\sigma(\bar{V}_{u})d\bar{Z}_{u}-\int_{0}^{t}4\sigma^{2}(\bar{V}_{u})du}e^{\int_{0}^{t}4\sigma^{2}(\bar{V}_{u})du}\right]
\nonumber
\\
&\leq
\left(\mathbb{E}\left[e^{\int_{0}^{t}4\sigma(\bar{V}_{u})d\bar{Z}_{u}-\frac{1}{2}\int_{0}^{t}(4\sigma)^{2}(\bar{V}_{u})du}\right]\right)^{1/2}
\left(\mathbb{E}\left[e^{8\int_{0}^{t}\sigma^{2}(\bar{V}_{u})du}\right]\right)^{1/2}
\nonumber
\\
&=\left(\mathbb{E}\left[e^{8\int_{0}^{t}\sigma^{2}(\bar{V}_{u})du}\right]\right)^{1/2}
\leq e^{4tM_{\sigma}^{2}}.
\end{align}
Hence, 
\begin{equation}
\mathbb{E}[\bar{V}_{t}^{2}]
\leq
V_{0}^{2}e^{2tM_{\mu}}
e^{4tM_{\sigma}^{2}}.
\end{equation}
Therefore, for any $0\leq t\leq T$,
\begin{align}
\left[\mathbb{E}|\bar{V}_{t}-V_{0}|^{2}\right]
&=\mathbb{E}\left[\bar{V}_{t}^{2}\right]
+V_{0}^{2}
-2V_{0}\mathbb{E}[\bar{V}_{t}]
\nonumber
\\
&\leq
V_{0}^{2}\left(e^{2tM_{\mu}}
e^{4tM_{\sigma}^{2}}+1-2e^{-tM_{\mu}-\frac{1}{2}tM_{\sigma}^{2}}\right)
\nonumber
\\
&\leq
V_{0}^{2}\left(e^{2TM_{\mu}}
e^{4TM_{\sigma}^{2}}+1-2e^{-TM_{\mu}-\frac{1}{2}TM_{\sigma}^{2}}\right).
\end{align}
Hence, we conclude that, for any $0\leq t\leq T$,
\begin{align}
\mathbb{E}[| \bar{V}_t - \bar{V}_0|]
\leq
V_{0}\left(e^{2TM_{\mu}}
e^{4TM_{\sigma}^{2}}+1-2e^{-TM_{\mu}-\frac{1}{2}T M_{\sigma}^{2}}\right)^{1/2}.
\end{align}
This completes the proof.
\end{proof}

\begin{proof}[Proof of Lemma~\ref{lem:2}]
First, we can compute that
\begin{equation}
\bar{V}_{t}=V_{0}e^{\int_{0}^{t}(\mu(\bar{V}_{u})-\frac{1}{2}\sigma^{2}(\bar{V}_{u}))du+\int_{0}^{t}\sigma(\bar{V}_{u})d\bar{Z}_{u}},
\end{equation}
which implies that
\begin{align*}
\mathbb{E}[\bar{V}_{t}]
&=\mathbb{E}\left[V_{0}e^{\int_{0}^{t}(\mu(\bar{V}_{u})-\frac{1}{2}\sigma^{2}(\bar{V}_{u}))du+\int_{0}^{t}\sigma(\bar{V}_{u})d\bar{Z}_{u}}\right]
\\
&\leq V_{0}e^{M_{\mu}t}\mathbb{E}\left[e^{-\int_{0}^{t}\frac{1}{2}\sigma^{2}(\bar{V}_{u}))du+\int_{0}^{t}\sigma(\bar{V}_{u})d\bar{Z}_{u}}\right]
\leq V_{0}e^{M_{\mu}t},
\end{align*}
where we used the fact that $\left\{e^{-\int_{0}^{t}\frac{p^{2}}{2}\sigma^{2}(\bar{V}_{u})du+p\int_{0}^{t}\sigma(\bar{V}_{u})d\bar{Z}_{u}}\right\}_{t\geq 0}$
is a non-negative local martingale and thus a supermartingale.

Next, we can compute that
\begin{equation}
V_{t}=V_{0}e^{\int_{0}^{t}(\mu(V_{u})-\frac{1}{2}\sigma^{2}(V_{u}))du+\int_{0}^{t}\sigma(V_{u})dZ_{u}+R_{t}^{V}+R_{t}^{C,V}},
\end{equation}
where, for any $t$,
$R_{t}^{V}:=\sum_{i=1}^{N_{t}^{S}}Y_{i}^{V}$ and
$R_{t}^{C,V}:=\sum_{i=1}^{N_{t}^{C}}Y_{i}^{C,V}$.
Therefore
\begin{align*}
\mathbb{E}[V_{t}]
&=\mathbb{E}\left[V_{0}e^{\int_{0}^{t}(\mu(V_{u})-\frac{1}{2}\sigma^{2}(V_{u}))du+\int_{0}^{t}\sigma(V_{u})dZ_{u}+R_{t}^{V}+R_{t}^{C,V}}\right]
\\
&=V_{0}\mathbb{E}\left[\mathbb{E}\left[e^{\int_{0}^{t}(\mu(V_{u})-\frac{1}{2}\sigma^{2}(V_{u}))du+\int_{0}^{t}\sigma(V_{u})dZ_{u}}|\mathcal{F}_{t}^{R,V}\right]e^{R_{t}^{V}+R_{t}^{C,V}}\right],
\end{align*}
where $\mathcal{F}_{t}^{R,V}$ is the natural filtration of $(R_{t}^{V},R_{t}^{C,V})$ process.
Moreover, 
\begin{align*}
&\mathbb{E}\left[e^{\int_{0}^{t}(\mu(V_{u})-\frac{1}{2}\sigma^{2}(V_{u}))du+\int_{0}^{t}\sigma(V_{u})dZ_{u}}|\mathcal{F}_{t}^{R,V}\right]
\\
&\leq 
e^{M_{\mu}t}\mathbb{E}\left[e^{-\int_{0}^{t}\frac{1}{2}\sigma^{2}(V_{u}))du+\int_{0}^{t}\sigma(V_{u})dZ_{u}}|\mathcal{F}_{t}^{R,V}\right]
\leq e^{M_{\mu}t},
\end{align*}
where we used the fact that $\left\{e^{-\int_{0}^{t}\frac{p^{2}}{2}\sigma^{2}(V_{u})du+p\int_{0}^{t}\sigma(V_{u})dZ_{u}}\right\}_{t\geq 0}$
is a non-negative local martingale and thus a supermartingale.
Hence, we conclude that
\begin{align}
\mathbb{E}[V_{t}]
\leq V_{0}e^{M_{\mu}t}
\mathbb{E}\left[e^{R_{t}^{V}+R_{t}^{C,V}}\right]
=V_{0}e^{\lambda^{V}t\mu^{V}+\lambda^{C}t\mu^{C,V}}e^{M_{\mu}t},
\end{align}
where we recall from \eqref{mu:defn} the definition that $\mu^{V}=\mathbb{E}[e^{Y_{1}^{V}}]-1$, $\mu^{C,V}=\mathbb{E}[e^{Y_{1}^{C,V}}]-1$ and 

This completes the proof.
\end{proof}

\end{document}